%% file: nc_arxiv.tex
\documentclass[aps,pra,superscriptaddress,twocolumn,nofootinbib,notitlepage,tightenlines]{revtex4-1}

\include{definitions_manuscript}

\usepackage{float} 
\usepackage{titlesec}
\titleformat{\section}
  {\normalfont\centering\bfseries\small\MakeUppercase}
  {\thesection.}{.5em}{}
 
\begin{document}

 \title{Fundamental limitations on distillation of quantum channel resources}

\author{Bartosz Regula}
\email{bartosz.regula@gmail.com}
\affiliation{School of Physical and Mathematical Sciences, Nanyang Technological University, 637371, Singapore}
\author{Ryuji Takagi}
\email{ryuji.takagi@ntu.edu.sg}
\affiliation{School of Physical and Mathematical Sciences, Nanyang Technological University, 637371, Singapore}
\affiliation{Center for Theoretical Physics and Department of Physics, Massachusetts Institute of Technology, Cambridge, Massachusetts 02139, USA}
 
\begin{abstract}
Quantum channels underlie the dynamics of quantum systems, but in many practical settings it is the channels themselves that require processing. 
We establish universal limitations on the processing of both quantum states and channels, expressed in the form of no-go theorems and quantitative bounds for the manipulation of general quantum channel resources under the most general transformation protocols. 
Focusing on the class of distillation tasks --- which can be understood either as the purification of noisy channels into unitary ones, or the extraction of state-based resources from channels --- we develop fundamental restrictions on the error incurred in such transformations and comprehensive lower bounds for the overhead of any distillation protocol. In the asymptotic setting, our results yield broadly applicable bounds for rates of distillation. We demonstrate our results through applications to fault-tolerant quantum computation, where we obtain state-of-the-art lower bounds for the overhead cost of magic state distillation, as well as to quantum communication, where we recover a number of strong converse bounds for quantum channel capacity.
\end{abstract}

 \maketitle

\addtocontents{toc}{\protect\setcounter{tocdepth}{0}}

\section*{Introduction}

One of the central aims of quantum information science is to precisely understand the limitations governing the use of quantum systems and develop the most efficient ways to take advantage of the laws of quantum physics. At the heart of such questions lies the study of quantum channels, which enable the manipulation of quantum states. However, in order to most effectively exploit quantum resources, it is important to be able to manipulate quantum channels themselves~\cite{chiribella_2008,chiribella_2008-1,gour_2019-2}. Channel transformations form the basis of some of the most pressing problems in quantum science, including for instance devising efficient schemes for quantum communication and key distribution for use in quantum networks~\cite{bennett_1996,gisin_2007,wehner_2018,wilde_2017}, or processing quantum circuits to aid in the mitigation and correction of errors in computation~\cite{shor_1996,campbell_2017,preskill_2018}.

Among such tasks, a particularly important class of problems known as \emph{channel distillation} can be distinguished. Depending on the resource in consideration, distillation can be understood either as channel purification, i.e.\ the conversion of noisy channel resources into pure (unitary) ones, or as the extraction of state-based resources from quantum channels. The motivation for such transformations comes from the fact that, just as in the case of maximally entangled singlets in entanglement theory~\cite{bennett_1996-1,bennett_1996}, pure resources can be necessary for the efficient utilisation of a given resource. This is the case in quantum computation, where one aims to synthesise unitary quantum gates which can be employed in a quantum circuit~\cite{campbell_2017}, or in quantum communication, where transfer of quantum information can be understood as the distillation of noiseless channels~\cite{wilde_2017}. However, the practical realisation of such distillation protocols can incur large costs in terms of the required resource overhead. Due to the importance of distillation schemes in mitigating the effects of noise, the study of their limitations is therefore vital in many fundamental quantum information processing tasks. 
A major obstacle to understanding the capabilities of channel manipulation protocols is that general strategies for transforming channels can be highly complex, using ancillary systems and the outputs from successive channel uses in order to adaptively improve the transformations~\cite{chiribella_2008-1}, or even processing channels in ways that do not enforce a definite causal order~\cite{chiribella_2013,oreshkov_2012}. 
Additionally, the limits of channel manipulation can be understood in different ways: in settings such as quantum computation, it is crucial to precisely understand and minimise the error incurred in manipulating gates and circuits, while in the study of quantum communication, it is often of interest to characterise asymptotic transformations and bound their achievable rates. We set to describe all such limitations in a common framework.

In this work, we establish a comprehensive approach to bounding the efficacy of manipulating the resources of quantum channels under general free transformations. We introduce universal lower bounds on the error of channel distillation, establishing precise quantitative limitations on the achievable performance of any distillation protocol. We reveal broad no-go results in multi-copy channel transformations under the most general manipulation protocols --- adaptive schemes whose causal order structure is not necessarily fixed --- allowing us to establish fundamental bounds on the overhead of any physical protocol for channel distillation and simulation. We furthermore use our results to provide strong converse bounds for asymptotic transformations, establishing sharp thresholds on the achievable distillation rates and characterising the ultimate limits of channel manipulation. All of our bounds rely on trade-off relations between the transformation accuracy and two resource quantifiers: the resource robustness and resource weight. By adopting such a general resource-theoretic approach, our methods are readily applicable in a wide variety of practical settings. This allows us not only to unify, consolidate, and extend results that have appeared in specialised settings, but also to develop methods and bounds that have not previously found use in characterising resource transformations.

Furthermore, since quantum states can be regarded as a special case of quantum channels, our results apply also to state manipulation tasks. Our framework significantly improves on and extends the applicability of previous methods which characterised state transformations, including a recent general approach to no-go theorems and bounds for quantum state purification introduced in Ref.~\cite{fang_2020}.

Our results can be applied in the characterisation of general quantum resources, encompassing both intrinsic properties of quantum channels as well as dynamical resources based on the underlying properties of quantum states. We showcase this broad applicability with two different applications to the most pertinent settings: fault-tolerant quantum computation with magic states, as well as quantum communication. First, we connect the tasks of magic state distillation and gate synthesis through the underlying resource theory of magic, and study the similarities and differences between the two tasks. We show that our results yield substantially improved bounds in this setting, providing in particular state-of-the-art general lower bounds on the overhead of magic state distillation. We then develop further the resource-theoretic approach to quantum communication assisted by no-signalling correlations, where we show how our bounds can be used to understand both one-shot and asymptotic transformations as well as to recover the strong converse property of no-signalling coding~\cite{Fang2020max,takagi_2020}. Adapting our methods to the study of communication assisted by separable and positive partial transpose (PPT) operations, we recover a number of leading single-letter strong converse bounds on quantum capacity~\cite{christandl_2017,wang_2019-3,berta_2017}, providing a simplification of proof methods employed in specialised approaches.
Furthermore, we formalise the trade-off relations between the success probability and transformation accuracy in probabilistic distillation protocols where post-selection is allowed. Here, our results indicate a qualitative difference in achievable accuracy between deterministic and probabilistic settings, and suggest potential advantages of employing probabilistic distillation protocols.

\section*{Results}

\subsection*{Setting}%
Quantum information processing can often be understood as the interplay of various physical resources~\cite{devetak_2008,chitambar_2019}.
In order to describe different quantum phenomena in a unified manner and establish methods which can apply to a broad variety of physical settings, we will employ the formalism of quantum resource theories~\cite{chitambar_2019}. The recent years have seen an active development of general resource-theoretic approaches to state manipulation and distillation problems, but the study of quantum channel manipulation in this setting is still in its infancy~\cite{gour_2019-2,takagi_2019,liu_2020,liu_2019-1,gour_2019-1}. In particular, not much is known about constraining one-shot transformations of channels beyond specific settings, and questions such as transformation rates have previously only been addressed under specific assumptions on the structure of the involved resources and protocols. Our approach will be to instead employ broad resource-theoretic methods which avoid presupposing any particular properties of the considered setting.

A resource theory is a general framework concerned with the manipulation of quantum states or channels under some physical restrictions~\cite{chitambar_2019}. The restrictions determine which states or channels are \emph{free}, in that they carry no resource and can be regarded as freely accessible under the physical constraints. The primary object of study of our work will be channel resources, so we assume that in the given physical setting, a particular subset of quantum channels $\OO$ has been singled out as the free channels. A large number of very different settings and resources can be described with a suitable choice of $\OO$, motivating us to establish methods which apply to any such choice. Therefore, to remain as general as possible, we will only make two natural assumptions about the set $\OO$: that it is closed, meaning that no resource can be generated by taking a sequence of resourceless channels, and that is convex, which means that simply probabilistically mixing free channels cannot generate any resource.

\begin{figure}[t]
\centering
\includegraphics[width=\columnwidth]{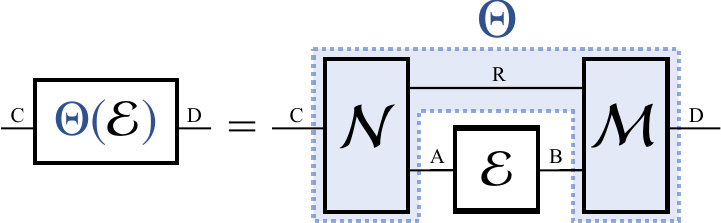}
\caption{\textbf{The general structure of a superchannel.} Given the Hilbert spaces of two quantum systems $A$ and $B$, we write ${\mathrm{CPTP}}(A \to B)$ to denote the set of quantum channels, i.e.\ completely positive and trace-preserving (CPTP) maps between operators acting on spaces $A$ and $B$. We associate with each channel $\E : A \to B$ its Choi matrix $J_{\E} \coloneqq {\mathrm{id}} \otimes \E(\Phi^+)$, where $\Phi^+ = \sum_{i,j} \ket{ii}\!\bra{jj}$ is the unnormalised maximally entangled state and ${\mathrm{id}}$ is the identity channel. Transformations of quantum channels are then maps from ${\mathrm{CPTP}}(A \to B)$ to ${\mathrm{CPTP}}(C \to D)$, or $(A \to B) \to (C \to D)$ in short. Under the necessary physical requirement that any such mapping should always take a quantum channel to a valid quantum channel, the most general form of a channel transformation is given by a quantum superchannel~\cite{chiribella_2008}. Such a transformation can be written as $\Theta(\E) = \M_{RB \to D} \circ ( {\mathrm{id}}_{R} \otimes \E ) \circ \N_{C \to RA}$ where $\N, \M$ are some pre- and post-processing quantum channels and $R$ denotes an ancillary system. For simplicity of notation, we often do not state explicitly which spaces the channels are acting on. }
\label{fig:sc-def}
\end{figure}

The most general way to manipulate a quantum channel is represented by a \emph{quantum superchannel}~\cite{chiribella_2008}, which we introduce in Figure~\ref{fig:sc-def}.  We are then interested in manipulating quantum channels with transformations which can be regarded as free within the constraints of the given theory.
In order to apply our results to all possible settings, we will make no specific assumptions about the considered set of free superchannels, save for the weakest possible constraint: that a free transformation $\Theta$ does not generate any resource by itself; that is, for any free channel $\M \in \OO$, it holds that $\Theta(\M) \in \OO$. We use $\mathbb{S}$ to denote the set of all such resource-preserving superchannels.
By studying these transformations, we will therefore obtain the most general bounds on the achievable performance of \emph{any} free channel manipulation protocol, since any physically motivated choice of free transformations will necessarily be a subset of $\mathbb{S}$.

We stress that, as a special case, all of our results apply also to the manipulation of the static resources of quantum states: they can be viewed as quantum channels that act on a trivial input space. For clarity, we will use $\FF$ instead of $\OO$ to denote the set of free states when discussing state-specific applications.

\subsection*{No-go theorems for resource distillation}
The task of distillation can be understood as the transformation of a noisy resource channel $\E$ into `pure' or `perfect' resources, which are represented by some target channel $\T$. Importantly, two distinct types of channel resource theories can be distinguished. The first type is concerned with the investigation of intrinsic channel resources; this includes various resource theories of quantum communication and the related setting of quantum memories. In such cases, it is often natural to regard some unitary channel $\U(\cdot) = U \cdot U^\dagger$ as the target of distillation protocols, representing noiseless dynamical resources. The other type is concerned with an underlying state-based resource and the manipulation of channels in order to extract or utilise the state resource more effectively; this includes, for instance, quantum entanglement, coherence, or thermodynamics. The target can then be a replacement (or preparation) channel $\R_\phi (\cdot) = \Tr(\cdot) \phi$ which prepares a given resourceful pure state. All of our results below apply to either of these settings, with $\T$ denoting a unitary or a replacement channel accordingly. Our task then is to understand when one can achieve transformations such that $F(\Theta(\E), \T)\geq 1-\ve$, where we use the worst-case fidelity~\cite{belavkin_2005,gilchrist_2005}
\begin{equation}\begin{aligned}
F(\E, \F) = \min_{\rho} F({\mathrm{id}} \otimes \E (\rho), {\mathrm{id}} \otimes \F (\rho))
\end{aligned}\end{equation}
to effectively benchmark the error of the transformation.
The choice of the worst-case fidelity as our figure of merit guarantees that the fidelity between the outputs of the channels will be large for any input state $\rho$, even when the channels are applied to a part of the system.

We endeavour to characterise the ultimate restrictions on the achievable performance of distillation by studying the trade-offs between three different quantities: the transformation error $\ve$, the resources contained in the input channel $\E$, and the resources of the target channel. To this end, we will employ two different resource measures. The \emph{resource robustness} $R_\OO$~\cite{vidal_1999,brandao_2015,diaz_2018-2,takagi_2019} and the \emph{resource weight} $W_\OO$~\cite{lewenstein_1998,uola_2020-1} are defined as
\begin{align}
  {R_{\OO}} (\E) &\coloneqq \min \lset \lambda \bar J_\E \leq \lambda J_\M,\; \M \in \OO \rset,\\
  W_\OO (\E) &\coloneqq \max \lset \lambda \bar J_\E \geq \lambda J_\M,\; \M \in \OO \rset, 
\end{align}
where $J_\E$ is the Choi matrix of the given channel, and the inequality is  understood in terms of positive semidefiniteness. The simple structure of the two quantities allows for a number of useful properties to be shown, such as their monotonicity under all free superchannels and submultiplicativity (see Supplementary Notes~\ref{sec:setting} and \ref{sec:many-copy}). The measures correspond to convex optimisation problems, in many relevant cases even reducing to efficiently computable semidefinite programs. $R_\OO$ and $W_\OO$ are natural generalisations of quantities defined at the level of quantum states, e.g.\ $R_\FF(\rho) = \min \lset \lambda \bar \rho \leq \lambda \sigma,\; \sigma \in \FF \rset$, where we recall that $\FF$ denotes free states in a considered resource theory. The robustness previously appeared in various ways in the characterisation of state transformations~\cite{brandao_2010-1,brandao_2015,regula_2020,liu_2019}, but the weight measure --- although a known geometric resource quantifier --- has not been connected with resource manipulation before.

To quantify the resources of the target channel, we will use the fidelity-based measure of the overlap with free channels:
\begin{equation}\begin{aligned}\label{eqm:fidelity_def}
	F_\OO(\T) \coloneqq& \max_{\M\in \OO} F(\T, \M).
\end{aligned}\end{equation}
This can be thought of as a parameter that determines how difficult a given target is to distil.
Although we will show that this parameter can be straightforwardly computed in most cases of practical interest,  in some contexts (such as quantum communication) an alternative figure of merit is often encountered: the Choi-state fidelity~\cite{barnum_2000,kretschmann_2004}
\begin{equation}\begin{aligned}
 \wt F_\OO(\T) \coloneqq  \max_{\M\in \OO} \, F \left(\wt J_\T, \wt J_\M \right),
\end{aligned}\end{equation}
where we denoted by $\wt J_\E$ the Choi matrix of a channel normalised so that $\Tr \wt J_\E = 1$.
In our discussion below, we will state our results using the parameter $F_\OO(\T)$ as this leads to the tightest bounds, but the bounds remain valid also if one replaces  $F_\OO(\T)$ with $\wt F_\OO(\T)$ everywhere.

We now give universally applicable, fundamental limitations on the performance of any resource distillation protocol.

\begin{theorem}\label{qip:thm1}
If there exists a free superchannel $\Theta \in \mathbb{S}$ such that $F(\Theta(\E), \T) \geq 1-\ve$ for a target channel $\T$, then
\begin{equation}\begin{aligned}\label{eqm:main-nogo-rob}
  \ve \geq 1 - F_\OO(\T) \, {R_{\OO}}(\E)
\end{aligned}\end{equation}
and
\begin{equation}\begin{aligned}\label{eqm:main-nogo-weight}
\ve \geq [1 - F_\OO(\T)]\, W_\OO(\E).
\end{aligned}\end{equation}
\end{theorem}
The bounds can be understood in two different ways: either as a general no-go result establishing the minimal error allowed within the constraints of the given resource theory, or, when $\ve$ is fixed, as a bound for the resources of $\E$ necessary for the distillation to be possible. The two bounds in Eqs.~\eqref{eqm:main-nogo-rob} and \eqref{eqm:main-nogo-weight} are very different from each other --- in both a quantitative and qualitative sense --- and can complement each other in various settings. We will aim to elucidate this with explicit examples and discussions in the following sections and in the Supplementary Notes.

As an immediate consequence of the Theorem, we see that the exact transformation with $\ve=0$ is impossible whenever $W_\OO(\E) > 0$, which is true e.g.\ for generic noisy channels with a full-rank Choi matrix.  Importantly, channels with  $W_\OO(\E) > 0$ cannot be distilled to a pure target $\T$ even when the target is less resourceful. This indicates strong constraints on distillation characterised by the resource weight $W_\OO$ and establishes a general no-go result in channel manipulation, extending earlier partial results~\cite{fang_2020}.

One important difference between the two bounds is that, when $\E$ is a pure (unitary or replacement) channel itself, then $W_\OO(\E) = 0$ and we gain no information from the weight bound. However, $R_\OO(\E)$ can provide a non-trivial error threshold even in this case, making it useful also in unitary-to-unitary or pure-to-pure transformations.

The result of Theorem~\ref{qip:thm1} directly applies also to the manipulation of quantum states, where now the free transformations $\mathbb{S}$ are in the form of quantum channels. Specifically, the bounds
\begin{align}
  \ve &\geq 1 - F_\FF(\phi) \, R_\FF(\rho)\\
\ve &\geq [1 - F_\FF(\phi)]\, W_\FF(\rho)
\end{align}
hold for any state $\rho$ undergoing a distillation protocol with a pure state $\phi$ as target.
This gives general error bounds on transformations of state-based resources. While the state-based robustness bound has previously appeared in Ref.~\cite{regula_2020}, the weight bound constitutes an improvement over previously known results, and in particular over a different approach to no-go theorems for resource purification which was recently introduced in Ref.~\cite{fang_2020}. In contrast to the framework of Ref.~\cite{fang_2020}, our results can characterise the manipulation of all quantum states (not only full-rank input states) and our quantitative bounds are strictly better than the previously known ones. This allows us to reveal substantially refined limitations on state-to-state transformations, as we will shortly demonstrate in explicit comparisons.

We will find that the bounds can tightly characterise one-shot transformations for specific cases of channels.
However, a major strength of the bounds lies not simply in estimating the errors in single-shot channel manipulation, but also in their applicability to multi-copy and asymptotic manipulation protocols: we now show that the bounds of Theorem~\ref{qip:thm1} can reveal powerful restrictions on distillation when multiple uses of a quantum channel are considered.

\begin{figure*}[!]
\centering
\includegraphics[width=.9\textwidth]{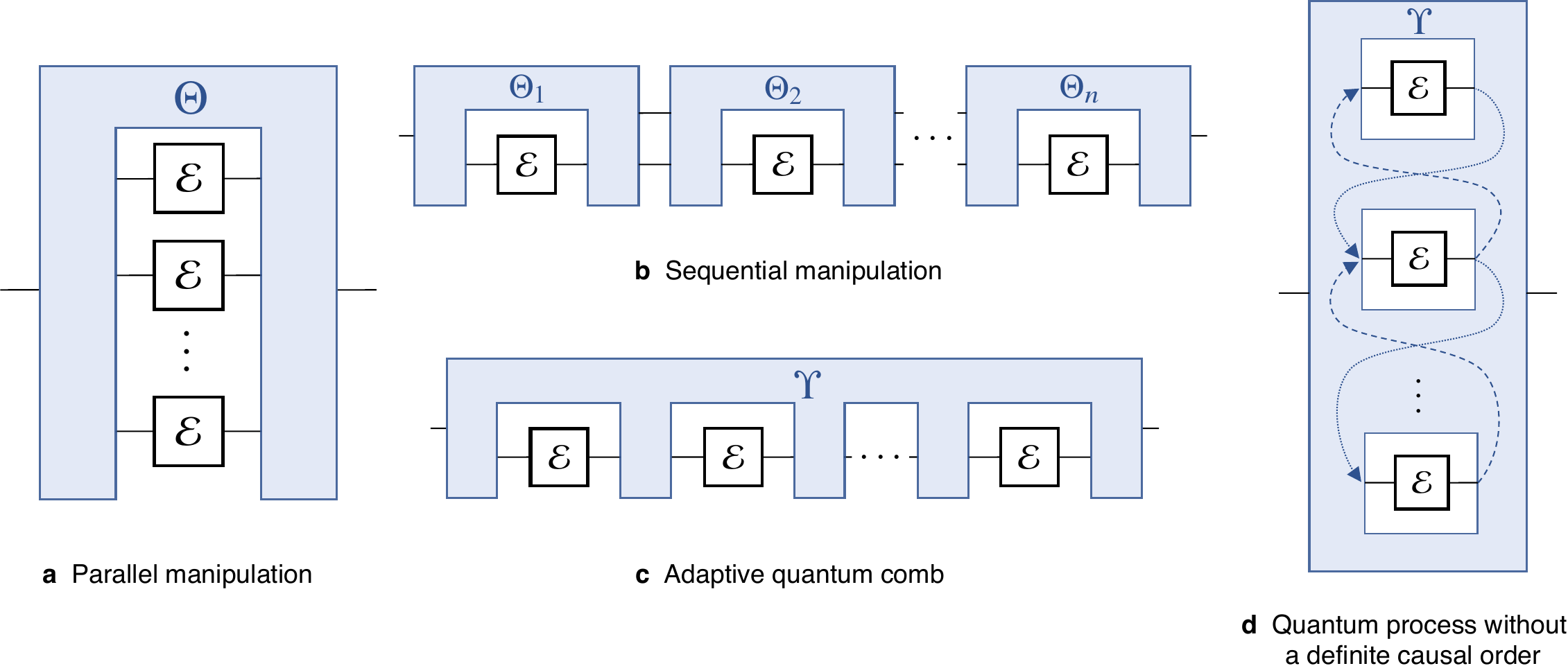}
\caption{\textbf{The different ways to manipulate many copies of a quantum channel with free transformations.} For quantum states, having access to multiple copies of a state $\rho$ is equivalent to acting on the tensor product $\rho^{\otimes n}$. A naive way to employ $n$ copies of a given channel is to consider them in parallel as $\E^{\otimes n}$ \textbf{(a)}. However, a more general way to manipulate $n$ copies of a channel is to employ a sequential (iterative) protocol \textbf{(b)}, which can be understood as the channel $\E$ being fed into a sequence of $n$ free superchannels one after another, allowing one to use the output of previous channel uses to improve the transformation. Indeed, such protocols are known to provide advantages over parallel ones in some settings~\cite{chiribella_2008-2,harrow_2010-1,yuan_2016,quintino_2019,liu_2020}, and this approach is commonly employed to transform channels in the setting of various resource theories such as quantum communication~\cite{Cooney2016constant,christandl_2017,kaur_2017,pirandola_2017}, entanglement~\cite{wilde_2018,gour_2019}, or magic~\cite{wang_2019-1}. However, even this does not represent the most general way to manipulate multiple channels within the setting of the given resource theory. When the causal order of the channels is fixed, any $n$-channel transformation scheme by means of a quantum circuit can be expressed as a so-called \emph{quantum comb}~\cite{chiribella_2008-1} \textbf{(c)}. Even more complex manipulation strategies are possible if one does not assume a definite causal order between the channel transformations, that is, when one is not able to say in what order the channels will be used throughout the protocol \textbf{(d)}. Such an approach allows one to treat the transformation trajectories themselves as quantum objects, leading to concepts such as superpositions of different causal orders~\cite{chiribella_2013,oreshkov_2012} which can indeed provide advantages over standard, causally ordered transformation methods~\cite{chiribella_2013,ebler_2018,quintino_2019}. These transformations are dubbed \emph{quantum processes}, and we will use them to characterise the most general physically realisable manipulation protocols involving multiple quantum channels.}
\label{fig:sc-many-copy}
\end{figure*}

\subsection*{Many-copy manipulation}
In contrast to transformations of quantum states, it does not suffice to consider channel manipulation as acting on the tensor product $\E^{\otimes n}$, but more complex protocols need to be considered. The most general form of such manipulation schemes are referred to as quantum processes (see Figure~\ref{fig:sc-many-copy}). We then use $\mathbb{S}_{(n)}$ to denote all free quantum processes, that is, all transformations $\Upsilon$ of $n$ channels such that the output channel $\Upsilon(\M_1,\ldots,\M_n)$ is a free channel whenever $\M_1, \ldots, \M_n \in \OO$. This approach will allow us to
characterise the performance of the most general protocols for manipulating channels or states within the physical constraints of the given resource.
\begin{theorem}\label{qip:thm2}
Given any distillation protocol $\Upsilon \in \mathbb{S}_{(n)}$ --- parallel, sequential, or adaptive, with or without a definite causal order --- which transforms $n$ uses of a channel $\E$ to some target channel $\T$ up to accuracy $\ve > 0$, it necessarily holds that
\begin{align}\label{eqm:main-copies-weight}
  n &\geq \log_{1/W_\OO(\E)}\, \frac{1- F_\OO(\T)}{\ve}
  \end{align}
  \vspace{-\baselineskip}
 and
  \begin{align}
  n &\geq \log_{{R_{\OO}}(\E)}\, \frac{1-\ve}{F_\OO(\T)}.
\end{align}
\end{theorem}
This gives general lower bounds on the overhead of distillation that must be obeyed by any physical transformation protocol. Once again, the bounds exhibit different behaviour: intuitively, the regime of $\ve$ very close to $0$ will be characterised more precisely by the bound based on resource weight $W_\OO$, while the robustness $R_\OO$ will perform better for larger error and for input channels $\E$ which are close to pure (unitary or replacement) channels.

An important aspect of the bound in Theorem~\ref{qip:thm2} is that it holds regardless of the structure of the involved channel manipulation process $\Upsilon$. This allows us to go beyond methods previously employed in settings such as quantum communication, which applied only to sequential protocols with a restricted structure.

As an immediate consequence of this result, the weight-based bound in Theorem~\ref{qip:thm2} shows that the number of uses of the channel $\E$ needed to perform distillation must scale as $\log(1/\ve)$ as $\ve \to 0$, establishing a universal limit on the overhead of distillation protocols such as quantum gate synthesis or noisy quantum communication.

\subsection*{Asymptotic manipulation}
The ultimate limitations on transforming a given state or channel are given by the maximal rate at which the conversion $\E \to \T$ can be performed with an asymptotic number of channel uses, allowing for error that vanishes asymptotically. Specifically, we will be interested in protocols which transform $n$ uses of a quantum channel $\E$ to $r n$ copies of the target channel $\T^{\otimes rn}$ up to error $\ve_n$. Imposing that the transformation is achieved exactly in the asymptotic limit, that is, $\ve_n \to 0$ as $n \to \infty$, and maximising over all such $r$ gives us the optimal asymptotic rate of converting $\E$ to $\T$ with free protocols.
We will distinguish two different rates: an adaptive rate $r_{\text{adap}}$ which allows the most general, adaptive processes acting on the input channels, and the parallel rate $r_{\text{par}}$ which considers parallel transformations of the form $\E^{\otimes n} \to \T^{\otimes r n}$ (recall the comparison in Figure~\ref{fig:sc-many-copy}).

The rates of distillation of quantum channel resources are an important aspect of understanding the limitations on resource manipulation~\cite{diaz_2018-2,liu_2020,kuroiwa_2020}, but little is known about them due to the difficulty in characterising the asymptotic properties of channel-based quantities~\cite{gour_2019-1,cubitt_2015}. Our methods allow us to establish two general bounds on the transformation rates. We can use the robustness $R_\OO$ to provide a general bound for the rate of any manipulation protocol, as well as obtain an improved bound for parallel protocols by suitably `smoothing' the definition of the robustness over channels within a small distance of the original input $\E$~\cite{diaz_2018-2,Fang2020max,liu_2019-1,gour_2019-1}. 
\begin{theorem}\label{qip:thm3}
If the target channel $\T$ satisfies $F_\OO(\T^{\otimes n}) = F_\OO(\T)^n$, then
\begin{align}\label{eqm:thm3}
	r_{\rm adap}(\E\to\T) &\leq \frac{\log {R_{\OO}}(\E)}{\log F_\OO(\T)^{-1}},\\
	r_{\rm par}(\E\to\T) &\leq \frac{D_\OO^\infty(\E)}{\log F_\OO(\T)^{-1}},
\end{align}
where $D_{\OO}^\infty(\E)\coloneqq \lim_{\delta\to 0}\limsup_{n\to\infty}\frac{1}{n}\log R_\OO^{\delta}(\E^{\otimes n})$ with $R_\OO^\delta(\E)\coloneqq \min_{F(\tilde\E,\E)\geq 1-\delta}R_\OO(\tilde\E)$.
\end{theorem}
The result establishes universal bounds on the achievable rate under any physical transformation protocol. 
Importantly, both of our bounds are \emph{strong converse} bounds, that is, they sharply characterise the threshold in achievable performance --- when a rate exceeds either of our bounds, the transformation fidelity necessarily goes to $0$, meaning that the error will grow very large and distillation cannot be reliably performed. 
The Theorem immediately applies in many settings of practical significance, as long as the condition $F_\OO(\T^{\otimes n}) = F_\OO(\T)^n$ is satisfied for the given target channel. This is a natural property which holds true both in dynamical resources such as communication, as well as in state-based channel resources such as entanglement, magic, coherence, or thermodynamics (see the forthcoming Table~\ref{tab:resources}).

In the majority of practically relevant settings, the robustness $R_\OO$ is submultiplicative under tensor product, meaning that $D_\OO^\infty(\E)\leq \log R_\OO(\E)$. Hence, the bound on $r_{\rm par}$ using $D_\OO^\infty(\E)$ might provide an improvement over the robustness-based bound, prompting the question of whether one can actually evaluate the tighter bound.  Notably, the regularisation $D_\OO^\infty(\E)$ has been computed exactly for a set of channels relevant in the study of quantum communication~\cite{Fang2020max}, which we will discuss in more detail shortly. The recent work of Ref.~\cite{gour_2019-1} began a systematic investigation of different regularisations in channel-based resource theories, but their general computability remains an open question. For quantum states, the regularisation $D_\FF^\infty(\rho) = \lim_{\delta\to 0}\lim_{n\to\infty}\frac{1}{n}\log R_\FF^{\delta}(\rho^{\otimes n})$ can be computed exactly under very mild assumptions on the set $\FF$~\cite{brandao_2010-1}, and it reduces to the regularised relative entropy of a resource. In such cases, our result recovers the fact that rates of distillation in resource theories of states are limited by the regularised relative entropy~\cite{hayashi2016quantum,brandao_2010-1}. We note also that related asymptotic bounds were considered in Ref.~\cite{liu_2020} for the case of state-based channel resource theories.

\begin{table*}
\setlength{\extrarowheight}{5pt}
  \begin{tabular}{@{\hspace*{1.8em}}l >{\raggedright}p{3.7cm} >{\raggedright}p{3.3cm} >{\raggedright}p{2cm} p{2.3cm}}
    \toprule
\multicolumn{1}{l}{\bfseries Channel resource} & \bfseries Target channel $\U$ \hspace{5pt} & $F_\OO(\U)$ \hspace{5pt} & $F_\OO(\U^{\otimes m})$ $\!\stackrel{?}{=}\! F_\OO(\U)^m$ \hspace{5pt} & \bfseries Computability of $R_\OO$ and $W_\OO$ \\
    \midrule
    \multicolumn{1}{l}{Quantum communication assisted by:}  &&&&\\
                      no-signalling transformations~\cite{Leung2015NS}  & Identity channel ${\mathrm{id}}_d$ & $\frac{1}{d^2}$ & Yes & SDP\\
                      separability-preserving transformations \hspace*{10pt}& Identity channel ${\mathrm{id}}_d$ & $\frac{1}{d}$ \cite{shimony_1995}  & Yes & {\raggedright Convex program (NP-hard~\cite{gurvits_2003})}\\
                      PPT-preserving transformations & Identity channel ${\mathrm{id}}_d$ & $\frac{1}{d}$ \cite{rains_2001} &Yes & SDP\\
    \multicolumn{1}{l}{Magic of many-qubit quantum channels~\cite{seddon_2019}} & Qubit T gate $T = \mathrm{diag}(1, e^{i\pi/4})$ & $\frac{1}{4}(2+\sqrt{2})$ \cite{bravyi_2019} & Yes~\cite{bravyi_2019} & SDP\\
     & Controlled-controlled-Z gate & $\frac{9}{16}$ \cite{bravyi_2019} & Yes \cite{bravyi_2019} & SDP\\
    \multicolumn{1}{l}{Magic of many-qudit quantum channels~\cite{wang_2019-1}} & Qutrit T gate $T = \mathrm{diag}(e^{2\pi i/9}, 1, e^{-2\pi i/9})$ & $(1+2 \sin(\pi/18))^{-1}$~\cite{wang_2019-1} \hspace*{5pt}& Yes \cite{wang_2020} * & SDP\vspace*{5pt}\\
    \toprule%
    \multicolumn{1}{l}{\bfseries State resource} & \bfseries Target state $\ket\phi$  & $F_\FF(\phi)$ & $F_\FF(\phi^{\otimes m})$ $\!\stackrel{?}{=}\! F_\FF(\phi)^m$ & \bfseries Computability of $R_\FF$ and $W_\FF$ \\
    \midrule
    \multicolumn{1}{l}{Quantum entanglement~\cite{horodecki_2009}}  & Maximally entangled state $ \frac{1}{\sqrt{d}} \sum_{i=0}^{d-1} \ket{ii}$ & $\frac1d$ \cite{shimony_1995} & Yes & Convex program (NP-hard~\cite{gurvits_2003})\\
        \multicolumn{1}{l}{Non-positive partial transpose~\cite{horodecki_2009}} & Maximally entangled state $ \frac{1}{\sqrt{d}} \sum_{i=0}^{d-1} \ket{ii}$ & $\frac1d$ \cite{rains_2001} & Yes & SDP\\
    \multicolumn{1}{l}{Quantum coherence~\cite{streltsov_2017}}  & Maximally coherent state $\frac{1}{\sqrt{d}} \sum_{i=0}^{d-1} \ket{i}$ & $\frac1d$ & Yes & SDP\\
    \multicolumn{1}{l}{Magic of many-qubit states~\cite{veitch_2014,howard_2017}} & T state $\frac{1}{\sqrt{2}}( \ket{0} + e^{i \pi /4}\ket{1} )$ & $\frac{1}{4}(2+\sqrt{2})$ \cite{bravyi_2019} & Yes~\cite{bravyi_2019} & SDP\\
    & CCZ state $\frac{1}{8} (1,\ldots,1,-1)^T$ & $\frac{9}{16}$ \cite{bravyi_2019} & Yes~\cite{bravyi_2019} & SDP\\
    \multicolumn{1}{l}{Magic of many-qudit states~\cite{veitch_2014}} & Hadamard `$+$' state $\propto (1+\sqrt{3})\ket{0} + \ket{1} + \ket{2}$ & $ \frac{1}{6}(3+\sqrt{3})$ \cite{wang_2020} & Yes \cite{wang_2020} & SDP\\ 
     & Norrell state $\frac{1}{\sqrt{6}}(\ket{0}-2\ket{1}+\ket{2})$ & $\frac{2}{3}$ \cite{wang_2020} & Yes \cite{wang_2020} & SDP\\ 
      \multicolumn{1}{>{\raggedright}p{5cm}}{Quantum thermodynamics with~Hamiltonian $H = \sum_{i} E_i \proj{i}$~\cite{Goold2016role}} & Energy eigenstate $\ket{i}$ & $\displaystyle \frac{e^{-\beta E_i}}{Z}$ ($\beta$:\! inverse temp., $Z$:\!~partition function) & Yes & Analytical\\
    \bottomrule
    \end{tabular}
    \caption{\textbf{Applicability of our bounds to common quantum resources}. We give an overview of quantum resources together with natural choices of target channels or states which are often used in distillation tasks. The list is by no means complete, but is meant to facilitate the application of our bounds in a selection of important setting. We see in particular that the parameter $F_\OO(\T)$ admits an exact analytical expression for all the target states on the list,
    and in addition is multiplicative under tensor product, which means that all of the bounds of this work (including the asymptotic bounds of Theorem~\ref{qip:thm3}) apply immediately.
    Furthermore, we see that the majority of cases discussed here allow for the robustness and weight measures to be computed as semidefinite programs (SDP). In the main text and in Supplementary Note~\ref{sec:main-app}, we provide more details about the theories of quantum communication and magic of channels and states, showing exactly how the bounds can be applied and how they perform.
    \\{}* In the case of the qutrit $T$ gate, instead of the quantity $F_\OO$ as defined in Eq.~\eqref{eqm:fidelity_def}, a closely related fidelity-type measure called the `min-thauma'~\cite{wang_2020} is used, which allows for an easier computation while otherwise acting in the same way. This makes no difference in the statement or properties of our bounds, so we make no distinction here and instead refer to the Supplementary Note~\ref{sec:main-app} for details.%
    }
        \label{tab:resources}
\end{table*}

\subsection*{Applying the bounds in practice}
We stress again that our main results discussed in Theorems~\ref{qip:thm1}--\ref{qip:thm3} apply to general convex resource theories of quantum channels and states, encompassing a wide variety of use cases. Since our discussions so far have presented them in a rather abstract manner, we will now discuss how the bounds can be evaluated in specific theories of interest.

With the exception of the regularised asymptotic bound in Theorem~\ref{qip:thm3}, all of our results depend only on three quantities: the overlap $F_\OO(\T)$ of the target channel, and either the robustness $R_\OO(\E)$ or the weight $W_\OO(\E)$ of the input.  In practical settings of interest, the choice of the target $\T$ is motivated by physical considerations --- representing, for instance, a maximally resourceful channel or state, or a particularly costly resource --- and the value of the parameter $F_\OO(\T)$ is typically known, so we can directly plug these quantities into the bounds established in Theorems~\ref{qip:thm1}--\ref{qip:thm3}. We collect some of the most important examples of such resources, together with the values of $F_\OO(\T)$, in Table~\ref{tab:resources}. All that remains now is to evaluate $R_\OO$ or $W_\OO$ for desired input channels. Fortunately, in many theories of interest, these two quantifiers can be computed as semidefinite programs, and often even evaluated or bounded analytically by utilising their convex duality and constructing suitable feasible solutions.

It will be instructive to discuss in more detail the applications to two fundamental examples. Full technical details and additional results are provided in Supplementary Note~\ref{sec:main-app}.

\begin{figure*}[t!]
\includegraphics[width=.75\textwidth]{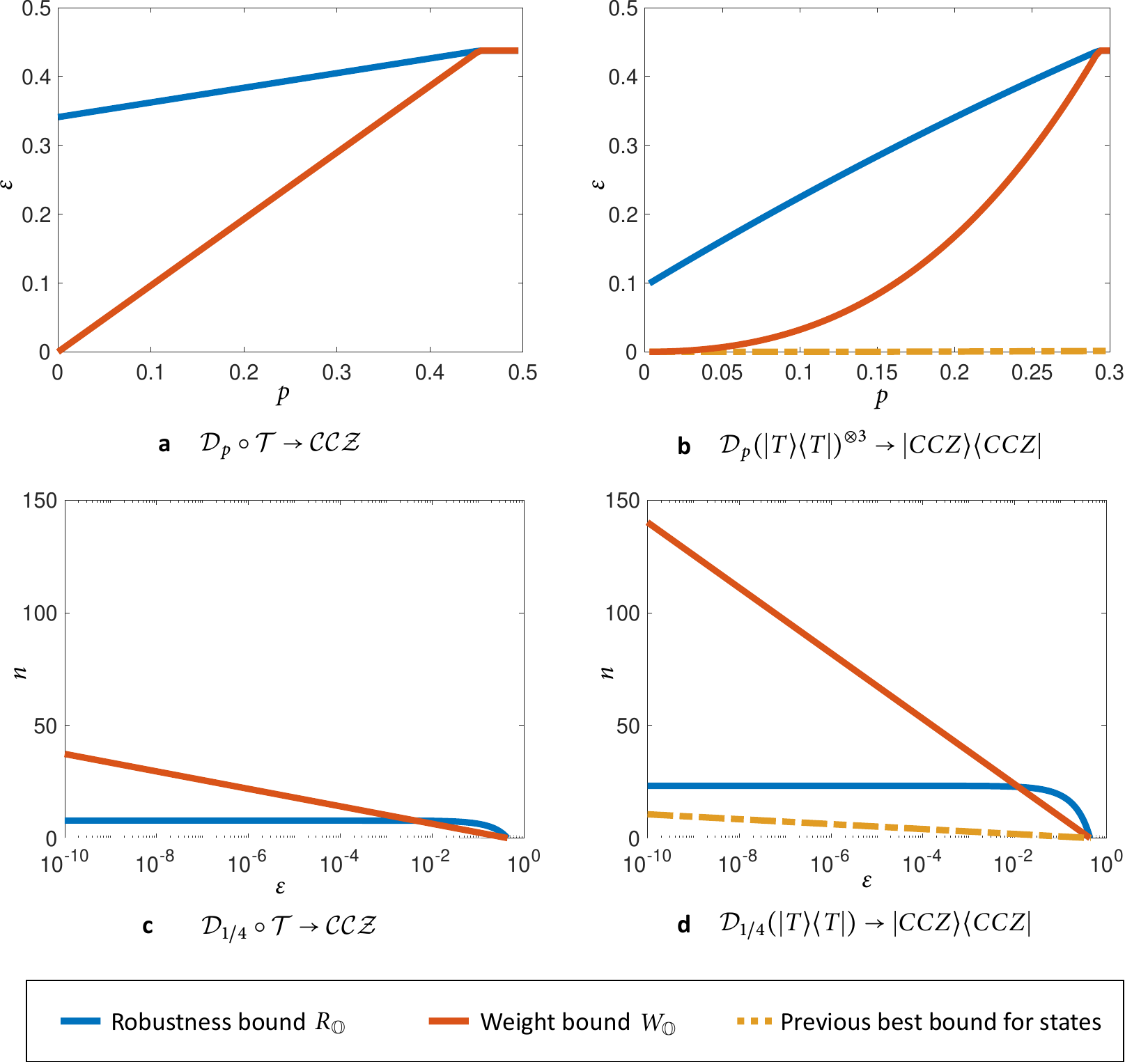}
		\caption{\textbf{Bounding the performance of gate synthesis and magic state distillation}. We plot lower bounds on: \textbf{a}--\textbf{b} the error $\ve$ necessarily incurred (as per Theorem~\ref{qip:thm1}); \textbf{c}--\textbf{d} the number of copies necessary (as per Theorem~\ref{qip:thm2}) in the given transformations between the depolarised T gate/state and the depolarised CCZ gate/state. The bounds are compared with the previous best general bound for magic state distillation introduced in~Ref.~\cite{fang_2020}. Here, $p$ is the noise parameter of the depolarising channel $\D_p(\rho) = (1-p) \rho + p \frac{\id}{2}$. 
		In \textbf{(a)} we explicitly see that the robustness bound indicates a significant error also in the noiseless case ($p=0$), whereas the weight bound becomes trivial for noiseless inputs. In \textbf{b}, we allow three copies of the noisy $T$ state to be used in the transformation --- here, both of our bounds significantly improve on the results of Ref.~\cite{fang_2020}, and in particular the robustness bound reveals that an error of $\approx 0.1$ is the best that one can hope for when converting $\ket{T}^{\otimes 3} \to \ket{CCZ}$ with any free transformation protocol. In \textbf{c} and \textbf{d}, we demonstrate the substantial advantages of the weight bound in bounding distillation overhead. Comparing the bounds for gate synthesis from the noisy $T$ gate in \textbf{c} and for magic state distillation from the noisy $T$ state in \textbf{d}, we can see that the bounds impose much higher requirements on the number of noisy states required to succeed.}
		\label{fig:magic-oneshot}
\end{figure*}

\subsection*{Application: gate synthesis and magic state distillation}%

Universal fault-tolerant quantum computation requires, in addition to the easily implementable Clifford gates, the use of costly non-Clifford unitaries such as the T gate~\cite{Bravyi2005magic}. Such gates are often implemented through the process of magic state injection~\cite{gottesman_1999}, which employs magic (non-stabiliser) states --- states that cannot be obtained with stabiliser operations alone --- to realise general quantum gates.
Magic states can provide feasible ways to synthesise general quantum circuits, but the main bottleneck in their efficient use is the resource cost associated with the required magic state distillation protocols~\cite{campbell_2017}. Understanding the limitations of such protocols and characterising the precise relations between magic state distillation and gate synthesis is thus highly important in paving the way to fault-tolerant quantum computation~\cite{campbell_2017-1,campbell_2017}.

In this setting, our results can be employed in two different ways: either directly at the level of channel manipulation (gate synthesis), or through an application to the task of magic state distillation. They therefore advance the resource-theoretic approach to magic~\cite{veitch_2014,howard_2017,seddon_2019,wang_2020,wang_2019-1} by explicitly shedding light on the precise quantitative connections between the channel-based theory and the underlying state-based resource. Here, the set of free channels $\OO$ can be understood as all stabiliser operations, or the larger set of completely stabiliser-preserving operations~\cite{seddon_2019}. We can then directly apply Theorems~\ref{qip:thm1}--\ref{qip:thm3} to immediately establish a number of bounds which can characterise the ultimate limitations in both exact and approximate transformations between channels and states in these resource theories. The relevant quantities $R_\OO$, $W_\OO$, and $F_\OO$ are computable as semidefinite programs in this setting, and for many channels of interest, such as quantum gates from the third level of the Clifford hierarchy~\cite{gottesman_1999}, the measures simplify to known quantities like the state-based stabiliser fidelity $F_\FF$~\cite{bravyi_2019} (see Supplementary Note~\ref{sec:main-app}). Applied at the level of states, our approach --- and in particular the weight-based bound --- constitutes a substantial improvement over the recent findings of~Refs.~\cite{fang_2020,seddon_2020} where lower bounds on the resource cost of magic state distillation were established.

Our results are demonstrated in Figure~\ref{fig:magic-oneshot}, where we plot the performance of our bounds in the transformation of the T gate $T = \mathrm{diag}(1, e^{i\pi/4})$~\cite{Bravyi2005magic} or the associated $\ket{T}$ state, affected by depolarising noise, to the controlled-controlled-Z gate CCZ.
We see that our results give non-trivial bounds on the error in all parameter regimes, revealing large errors even in cases where previous bounds could not do so. Notably, this yields state-of-the-art lower bounds on the overhead of magic-state distillation, as well as general bounds directly for the task of quantum gate synthesis.

\subsection*{Application: quantum communication}
The \emph{quantum capacity} $Q(\E)$ characterises the rate at which quantum information can be communicated through a channel, and bounding this quantity is a fundamental problem in quantum communication~\cite{lloyd_1997,barnum_1998,hayashi2016quantum}. It is often useful to allow the communicating parties to use some assistance --- in the form of shared correlations, or the ability to perform some limited set of joint operations --- in order to aid the communication. This traditional setting of quantum communication can be encompassed in our resource-theoretic framework of channel manipulation: the goal can be understood as using free superchannels (encoding and decoding operations) in order to purify a noisy quantum channel to the qubit identity channel ${\mathrm{id}}_2$, with the latter representing perfect noiseless communication.
Here, we will see that our general results can be readily applied to assess several fundamental limitations in this task.

\begin{figure*}[t]
\centering
\includegraphics[width=.75\textwidth]{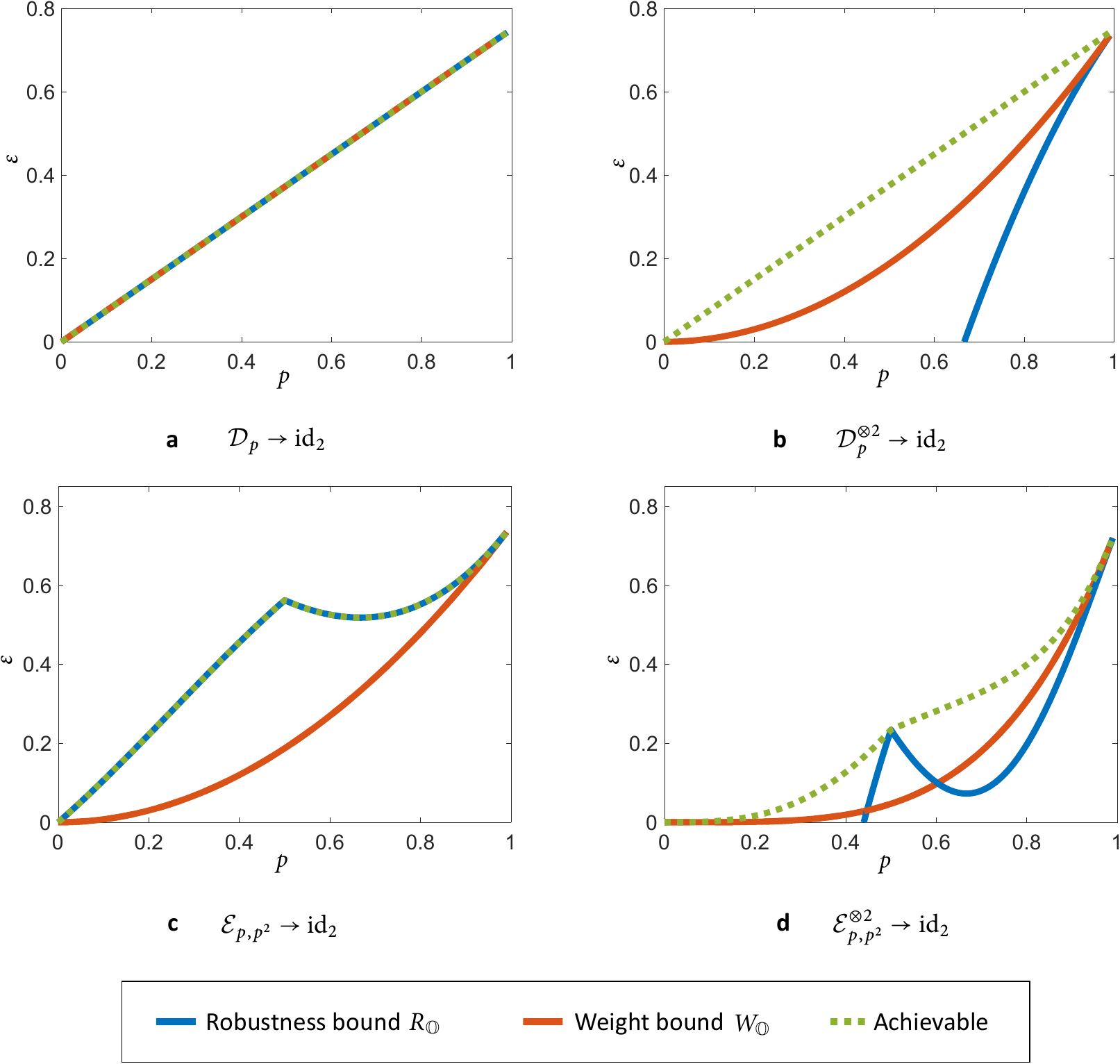}
		\caption{\textbf{Errors in one-shot quantum communication}. Lower bounds on the error $\ve$ in the transformation of channels to the qubit identity channel ${\mathrm{id}}_2$ under no-signalling codes $\mathbb{S}_{\rm NS}$. We plot the bounds obtained from Theorem~\ref{qip:thm1} for: \textbf{a}--\textbf{b} the qubit depolarising channel $\D_p (\rho) = (1-p) \rho + p \frac{\id}{d}$; \textbf{c}--\textbf{d} the dephrasure channel $\E_{p,q}(\rho) = (1-q) [(1-p) \rho + p Z \rho Z]$ $+\, q \proj{2}$ \cite{leditzky_2018-1}, where we set $q=p^2$. In \textbf{a}, the depolarising channel satisfies $W_{\OO} (\D_p) = p$, and the robustness and weight-based bounds are actually equal: we have $\ve \geq p (d^2 - 1)/d^2$. In fact, here the bounds match the achievable fidelity, meaning that Theorem~\ref{qip:thm1} quantifies the error in the one-shot transformation $\D_p \to {\mathrm{id}}_2$ under $\mathbb{S}_{\rm NS}$ exactly. The robustness bound is also seen to be tight in \textbf{c}. The weight bound can become more effective than the robustness bound when we consider more copies of the input channel, as seen in \textbf{b} and \textbf{d}. The importance of the weight bound is highlighted here, as it can certify that zero-error ($\ve=0$) communication is impossible for all $p>0$.}
		\label{fig:ns-oneshot}
\end{figure*}

For instance, the setting of no-signalling (NS) assisted communication~\cite{Leung2015NS,Duan2016nosignalling,wang_2017,takagi_2020,Fang2020max} allows Alice and Bob to perform joint coding protocols which obey the no-signalling condition from Alice to Bob and vice versa.  First, it is insightful to see what the bounds of Theorem~\ref{qip:thm1} tell us about one-shot transformations $\E \to {\mathrm{id}}_2$ in this setting. Here, the maximal fidelity achievable under no-signalling codes can be computed with an SDP~\cite{Leung2015NS}, which allows us to gauge the performance of our bounds exactly.
We demonstrate this with a numerical investigation in Figure~\ref{fig:ns-oneshot}, showing that our results can become exact in some cases, and the two bounds can complement each other in different situations.
Beyond such single-shot transformations, when multiple uses of channels are considered, our bounds can lead to tight asymptotic results.
In particular, Theorem~\ref{qip:thm3} gives a strong converse bound on the NS-assisted quantum capacity as $Q_{\rm NS}(\E) \leq D_{\OO}^\infty(\E)$. The quantity $D_{\OO}^\infty(\E)$ can be computed exactly in this case~\cite{Fang2020max}, and it equals the mutual information of the channel~\cite{bennett_2002}. Moreover, this is actually an achievable rate of communication~\cite{bennett_2002,Leung2015NS}, which means that $Q_{\rm NS}(\E)$ is given exactly by the mutual information of $\E$. In addition to recovering this tight bound, our results also show the strong converse property of NS-assisted communication~\cite{Fang2020max,takagi_2020}, which says that the capacity $Q_{\rm NS}(\E)$ is a strong converse rate of communication.

Another fundamental paradigm is quantum communication assisted by local operations and classical communication (LOCC)~\cite{bennett_1996}. Due to the complexity of describing LOCC, various approximations of this class of channels are often employed~\cite{rains_1997,rains_2001,Leung2015NS,hayashi2016quantum,tomamichel_2017}, the most common ones being the set of separable channels $\OO_{\rm SEP}$~\cite{rains_1997} (maps whose Choi matrix is separable) and positive partial transpose (PPT) channels $\OO_{\rm PPT}$~\cite{rains_2001} (maps whose Choi matrix is PPT). We can readily apply our results in two different ways, by considering either the capacity $Q_{\rm SEP}$ of communication assisted by protocols which preserve $\OO_{\rm SEP}$, or the capacity $Q_{\rm PPT}$ where Alice and Bob can perform joint manipulation protocols which preserve $\OO_{\rm PPT}$.
 Theorems~\ref{qip:thm1}--\ref{qip:thm3} then immediately provide a number of bounds on quantum capacity assisted by the most general adaptive protocols, in both the one-shot and asymptotic setting.
Notably, we obtain that the robustness $R_{\OO_{\rm SEP}}$ gives a strong converse bound to $Q_{\rm SEP}$. We show, in fact, that the robustness in this case equals a quantity known as the max-relative entropy of entanglement~\cite{christandl_2017}, therefore recovering and extending a bound of~Ref.~\cite{christandl_2017}, while providing an arguably simpler proof technique. In the PPT setting, we get an analogous result, which is closely related to bounds based on the so-called max-Rains information~\cite{wang_2019-3,berta_2017} --- these constitute, in many cases, the best known efficiently computable bounds on LOCC-assisted capacity.

The above shows the direct applicability of our formalism to upper bounding channel capacities in a number of different settings. Our methods thus not only provide useful benchmarks of practical significance, but also unify different specialised approaches and show them all to be part of a broader, resource-theoretic framework for the manipulation of quantum channels, which extends beyond entanglement and communication theory.

\subsection*{Extension to probabilistic protocols}
Our results can also be extended to the setting where the channel transformations are not realised deterministically, but can fail with a certain probability. Here, we establish general trade-offs between the success probability and the error in the transformation, extending the insights and establishing bounds which take into account the non-deterministic character of the transformations. The results suggest that potential advantages of probabilistic protocols over deterministic ones can be found in some cases. Because of the technicality of such extensions caused by the more complicated nature of probabilistic channel transformations~\cite{burniston_2020}, we defer the details to the Methods and Supplementary Note~\ref{sec:prob}.

\section*{Discussion}
We introduced universal quantitative bounds on the achievable performance of any deterministic distillation protocol in general quantum resource theories. We showed that our bounds can be used to establish fundamental no-go relations in the manipulation of quantum channels, introduce powerful restrictions on the overhead of any physical distillation protocol using the most general, adaptive manipulation schemes with indefinite causal order, and lead to several strong converse bounds for the asymptotic transformations of channels. We demonstrated the versatility of our bounds through applications to quantum communication and stabiliser-state quantum computation, using the bounds to characterise these important resource theories. Finally, we extended the insights provided by our bounds to distillation protocols which allow probabilistic implementations.

A key feature of our bounds is their generality: using a general resource-theoretic approach, we were able to establish limitations on the manipulation of quantum resources under only the most basic assumptions enforced by the structure of quantum mechanics. This reveals common aspects shared by all types of quantum resources, establishing our bounds as fundamental quantitative limitations on channel manipulation. Importantly, however, such breadth of this approach does not diminish its usefulness in concrete settings of interest --- we have shown that all of our bounds can be directly applied in a multitude of relevant resources, and we expect them to find use also in many settings that we have not considered here explicitly. On the practical side, our results shed light in particular on the important problem of purifying noisy resources. Because of the unavoidably noisy character of near-term quantum technologies~\cite{preskill_2018}, such distillation schemes are often necessary, and we therefore anticipate our bounds to find use in the practical investigation of the limitations of quantum information processing in a broad range of settings such as channel discrimination~\cite{pirandola_2019-2,Zhuang2020ultimate}, noise estimation~\cite{Pirandola2017noise}, programming of quantum channels~\cite{Yang2020programming,banchi_2020}, and covariant error correction~\cite{Faist2019covariant,Woods2020covariant,Zhou2020covariant}.

An interesting direction to consider in further research would be to understand precisely when and how the bounds can be achieved in practical setups, and how the different types of transformations --- parallel, adaptive, or ones with an indefinite causal order --- perform in various types of manipulation tasks. It would also be intriguing to apply our methods to the transformation of other types of dynamical resources, such as Bell nonlocality or quantum contextuality, which may provide further operational insights into the fundamental advantages enabled by quantum theory in different settings.

\textit{Note.{\,---}} During the completion of this manuscript, we became aware of a related work by Fang and Liu~\cite{fang_2020-2} where the authors independently considered the resource weight and obtained results related to the weight-based bounds in our Theorems~\ref{qip:thm1}~and~\ref{qip:thm2}.


\section*{Methods}

We give an overview of the main techniques used to establish our results in Theorems~\ref{qip:thm1}--\ref{qip:thm3}. The complete technical details, along with additional discussion and extensions, can be found in the Supplementary Information.

\subsection*{One-shot bounds (Theorem~\ref{qip:thm1})}

We will use the fact that both the robustness and weight measures can be expressed in terms of the max-relative entropy $D_{\max} ( \rho \| \sigma ) \coloneqq  \log \inf \lset \lambda \bar \rho \leq \lambda \sigma \rset$~\cite{datta_2009}. Defining $R_{\max} ( \rho \| \sigma) \coloneqq 2^{D_{\max}(\rho \| \sigma)}$, for any channels $\E, \F : A \to B$ one can define the optimised channel divergence~\cite{wilde_2020}
\begin{equation}\begin{aligned}
	R_{\max} ( \E \| \F ) \coloneqq \max_{\psi} R_{\max} ({\mathrm{id}} \otimes \E (\psi) \| {\mathrm{id}} \otimes \F (\psi)).
\end{aligned}\end{equation}
This generalisation of the max-relative entropy to channels obeys some useful properties, and in particular it holds that~\cite{wilde_2020}
\begin{equation}\begin{aligned}
	R_{\max} ( \E \| \F ) = R_{\max} ( J_\E \| J_\F),
\end{aligned}\end{equation}
that is, it suffices to consider the Choi matrices of the channels to evaluate the max-relative entropy. Exploiting the properties of $R_{\max}$ and the convex structure of the involved optimisation, we can then express the robustness measure ${R_{\OO}}$ as
\begin{align}
	{R_{\OO}} (\E) &= \min_{\M \in \OO} R_{\max} ( \E \| \M) \\
	&= \min_{\M \in \OO} \max_{\psi} R_{\max} ({\mathrm{id}} \otimes \E (\psi) \| {\mathrm{id}} \otimes \M (\psi))\\
	&= \max_{\psi} \min_{\M \in \OO} R_{\max} ({\mathrm{id}} \otimes \E (\psi) \| {\mathrm{id}} \otimes \M (\psi)).
\end{align}
In a very similar way, the weight $W_\OO$ can be written
\begin{align}
	W_\OO(\E)^{-1} &= \min_{\M \in \OO} R_{\max} ( {J}_\M \| {J}_\E)\\
	&= \min_{\M \in \OO} \max_{\psi} R_{\max} ({\mathrm{id}} \otimes \M (\psi) \| {\mathrm{id}} \otimes \E (\psi))\\
	&= \max_{\psi} \min_{\M \in \OO} R_{\max} ({\mathrm{id}} \otimes \M (\psi) \| {\mathrm{id}} \otimes \E (\psi)).
\end{align}
The next step is to use convex duality to express the robustness and weight as
\begin{align}
  {R_{\OO}} (\E) = \max_{\psi} \big\{ &\Tr\!\left( X\, {\mathrm{id}} \otimes \E (\psi) \right) \;\big|\\
   &X \geq 0,\; \Tr\!\left( X \,{\mathrm{id}} \otimes \M (\psi) \right) \leq 1 \; \forall \M\in \OO \big\},\nonumber\\
   W_\OO (\E) = \min_{\psi} \big\{ &\Tr\!\left( X\, {\mathrm{id}} \otimes \E (\psi) \right) \;\big|\\
   &X \geq 0,\; \Tr\!\left( X \,{\mathrm{id}} \otimes \M (\psi) \right) \geq 1 \; \forall \M\in \OO \big\}.\nonumber
\end{align}

The core of the idea behind the proof of Theorem~\ref{qip:thm1} is then as follows. Due to the purity of the target channel $\T$ (whether it is a unitary channel $\U$ or a replacement channel $\R_\phi$), the expression for the fidelity $F_\OO$ simplifies: we either have
\begin{equation}\begin{aligned}
   F_\OO(\U) = \max_{\M\in \OO} \Tr \left( {\mathrm{id}} \otimes \U (\psi^\star) \,{\mathrm{id}} \otimes \M (\psi^\star) \right)
\end{aligned}\end{equation}
for some optimal pure state $\psi^\star$, or, in the state case, we can write
\begin{equation}\begin{aligned}
	F_\FF(\phi) = \max_{\sigma \in \FF} \Tr( \phi \sigma ).
\end{aligned}\end{equation}
This allows us to use either the target channel $\U$ or the target state $\phi$ to construct feasible solutions for the dual form of ${R_{\OO}}$ and $W_\OO$. Specifically, the operator $\frac{1}{F_\OO(\U)} \left({\mathrm{id}} \otimes \U (\psi^\star)\right)$ or $\frac{1}{F_\FF(\phi)} \phi$ can be used to lower bound ${R_{\OO}}$, while the operator $\frac{1}{1-F_\OO(\U)} \left( \id - {\mathrm{id}} \otimes \U(\psi^\star)\right)$ or $\frac{1}{1-F_\FF(\phi)} \left( \id - \phi\right)$ gives an upper bound on $W_\OO$. These bounds immediately lead to the restrictions stated in Theorem~\ref{qip:thm1}.

\subsection*{Many-copy bounds (Theorem~\ref{qip:thm2})}

Mathematically, an $n$-channel quantum process $\Upsilon$ --- the most general physically realisable manipulation protocol involving multiple quantum channels  --- is an $n$-linear map which takes $n$ channels as input and outputs a single channel. Although the property of complete positivity is sometimes expected of such transformations~\cite{chiribella_2008-1,oreshkov_2012}, we do not require it, and all of our results are valid as long as the maps in consideration satisfy $\Upsilon(\N_1, \ldots, \N_n) \in {\mathrm{CPTP}}$ for any $\N_1, \ldots, \N_n \in {\mathrm{CPTP}}$. We can then define the set of free quantum processes as those which always result in a free channel, provided that all inputs are free:
\begin{equation}\begin{aligned}
\mathbb{S}_{(n)} \coloneqq \lset \Upsilon \bar \Upsilon(\M_1,\ldots,\M_n) \in \OO \;\, \forall \M_1, \ldots, \M_n \in \OO \rset.
\label{eqm:free comb}
\end{aligned}\end{equation}
In this sense, superchannels can be understood as (completely positive) processes acting on a single input.

Our main technical contribution is to show a very general type of sub- or super-multiplicativity that the robustness and weight measures obey. In particular, we show that, given any collection of $n$ channels $(\E_1, \ldots, \E_n)$, it holds that
\begin{equation}\begin{aligned}\label{eqm:weight-comb-submult}
  W_\OO\left(\Upsilon(\E_1,\ldots,\E_n)\right) \geq \prod_i W_\OO(\E_i)
\end{aligned}\end{equation}\vspace*{-.3\baselineskip}
and
\begin{equation}\begin{aligned}\label{eqm:rob-comb-submult}
   {R_{\OO}}\left(\Upsilon(\E_1,\ldots,\E_n)\right) \leq \prod_i R_\OO(\E_i)
\end{aligned}\end{equation}
for any free process $\Upsilon \in \mathbb{S}_{(n)}$. The basic idea behind the proof is to take an optimal channels $\M_i$ such that each $\E_i$ satisfies $J_{\E_i} \geq \mu_i J_{\M_i}$ in the case of $W_\OO$ or $J_{\E_i} \leq \mu_i J_{\M_i}$ in the case of $R_\OO$. By showing that $\Upsilon(\mu_1 \M_1, \mu_2 \M_2, \ldots, \mu_n \M_n)$ forms a valid feasible solution for $W_\OO\left(\Upsilon(\E_1,\ldots,\E_n)\right)$ or $R_\OO\left(\Upsilon(\E_1,\ldots,\E_n)\right)$, we obtain our desired result. 
Notably, the proof uses only the positivity and $n$-linearity of the free process $\Upsilon$, requiring no additional assumptions about the structure of the transformation.

Combined with Theorem~\ref{qip:thm1}, our result then immediately leads to the statement of Theorem~\ref{qip:thm2}. However, we stress that the property of sub- or super-multiplicativity that we have shown is much more general: the target in the transformation need not be a pure (unitary or replacement) channel, meaning that the inequalities in Eqs.~\eqref{eqm:weight-comb-submult}--\eqref{eqm:rob-comb-submult} are valid for any channel manipulation protocol. Although in the main text we have focused on the application to the task of channel distillation, this general feature of the robustness and weight measures can find use in broader channel processing tasks which involve multiple channels.

For instance, the task of channel synthesis is concerned with simulating the action of the given channel $\E$ by employing multiple uses of another channel, $\F$, and processing them with a free transformation protocol $\Upsilon$. We then immediately obtain lower bounds on the required number of uses of $\F$ under any physical transformation protocol:
\begin{equation}\begin{aligned}
  n \geq \frac{\log {R_{\OO}}(\E)}{\log {R_{\OO}}(\F)},\quad n \geq \frac{\log W_\OO(\E)}{\log W_\OO(\F)},
\end{aligned}\end{equation}
where in the second inequality we have assumed that $W_\OO(\E)$ and $W_\OO(\F)$ are not both 0.
When $\F$ is chosen to be a pure resource channel such as the target $\T$, this can be understood as the opposite task to distillation --- resource dilution. 

\subsection*{Asymptotic bounds (Theorem~\ref{qip:thm3})}

Both of our asymptotic bound in Theorem~\ref{qip:thm3} are consequences of the results of Theorem~\ref{qip:thm1} and \ref{qip:thm2} coupled with the assumption that $F_\OO(\T^{\otimes m})=F_\OO(\T)^m$. This means in particular that they apply to general manipulation protocols $\Upsilon$ without making assumptions about their structure, in contrast to most previous asymptotic bounds in the literature which explicitly considered sequential manipulation protocols with a fixed causal order.

We note that the second, regularised bound for parallel channel transformations (Eq.~\eqref{eqm:thm3}) requires a more careful approach, relying also on some technical bounds on the fidelity distance between channels. In particular, the `smoothing' parameter $\delta$ encountered here is the reason why the result applies to parallel manipulation protocols only --- an extension to more general transformations would entail an optimisation in the space of quantum combs (or quantum processes), and a straightforward application of our methods to this case does not appear to be possible. Whether this can be circumvented with a different approach remains an open question.

\subsection*{Extension to probabilistic protocols}

We have so far focused our discussion on deterministic channel transformations where superchannels (and quantum processes) transform channels to channels. 
To investigate a probabilistic version of such protocols, we need to consider `sub-superchannels': the linear maps which transform quantum channels to probabilistic implementations of channels in the form of completely positive, trace--non-increasing maps (subchannels), even when acting only on a part of a larger system~\cite{burniston_2020}. The operational meaning of these maps becomes clear by considering them as constituents of superinstruments, i.e., collections of sub-superchannels $\{\tilde\Theta_i\}$ --- each representing a single outcome of a probabilistic transformation --- such that the overall transformation $\sum_i \tilde\Theta_i$ is a superchannel. 
Just as the usual quantum instrument, a superinstrument can be assumed to come with a classical register recording which sub-superchannel was applied. 
Then, probabilistic protocols are declared successful when we learn that $\tilde\Theta_0$ was realised and are judged to have failed otherwise. 
To introduce the notion of free transformation in this context, let us first define the set of free subchannels.
If we think of free subchannels as a probabilistic version of free channels, it is natural to impose that every free subchannel probabilistically realises a transformation implemented by some free channel. 
This observation motivates us to define the set of free subchannels $\tilde\OO$ with respect to the given set of free channels $\OO$ as 
\begin{equation}\begin{aligned}
\tilde\OO\coloneqq \big\{ \tilde\M \;\big|\;& \forall \rho\in\DD,\; \exists\, \M\in\OO,\, t\in[0,1]\\
&\mbox{ s.t. } {\mathrm{id}}\otimes\tilde\M(\rho) = t\cdot{\mathrm{id}}\otimes\M(\rho) \big\},
\end{aligned}\end{equation}
and we correspondingly define the set of free sub-superchannels as $\tilde{\mathbb{S}} \coloneqq \lset\tilde\Theta\sbar\forall \M\in\OO,\; \tilde\Theta(\M)\in\tilde\OO\rset$.

We also need to establish a figure of merit for the probabilistic purification protocol. 
A subtlety is that the probability of the occurrence of a sub-superchannel $\tilde\Theta$ depends not only on the input channel $\E$, but also the input state $\psi$ as $\Tr[{\mathrm{id}}\otimes\tilde\Theta(\E)(\psi)]$. 
Integrating this observation with the definition of the fidelity for channels $F(\E, \T)$, we define the fidelity between the target channel and an output subchannel conditioned on its occurrence as
\begin{equation}\begin{aligned}
 F_{\rm cond}(\tilde\Theta(\E),\T)\coloneqq\min_\psi F\left(\frac{{\mathrm{id}}\otimes\tilde\Theta(\E)(\psi)}{p(\psi)}, {\mathrm{id}}\otimes\T(\psi)\right)
 \label{eqm:fidelity condition}
\end{aligned}\end{equation}
where $p(\psi)=\Tr[{\mathrm{id}}\otimes\tilde\Theta(\E)(\psi)]$.

We can then establish an analogue of Theorem~\ref{qip:thm1} for probabilistic channel manipulation. Specifically, we show that if there exists a free sub-superchannel $\tilde\Theta \in \tilde{\mathbb{S}}$ which achieves the transformation $\E \to \T$ with fidelity $F_{\rm cond}(\tilde\Theta(\E), \T) \geq 1-\ve$ and probability $p=\Tr[{\mathrm{id}}\otimes\tilde\Theta(\E)(\psi)]$, then
\begin{equation}\begin{aligned}\label{eqm:main-nogo-rob prob}
  \ve &\geq 1 - \frac{{R_{\OO}}(\E) \, F_\OO^\psi(\U)}{p}
\end{aligned}\end{equation}
and\vspace*{-.5\baselineskip}
\begin{equation}\begin{aligned}\label{eqm:main-nogo-weight prob loose}
  \ve &\geq 1-\frac{1-(1 - F_\OO^\psi(\U) )W_\OO(\E)}{p}
\end{aligned}\end{equation}
where $F_\OO^\psi(\U) \coloneqq \max_{\M\in\OO}F({\mathrm{id}}\otimes\U(\psi),{\mathrm{id}}\otimes\M(\psi))$. This resembles our previous bounds, but now explicitly incorporates the dependence on a probability $p$.

Another type of bound for probabilistic transformations can be obtained by taking $\M\in\OO$ to be a free channel such that $J_\E\geq W_\OO(\E)J_\M$. We then obtain
\begin{equation}\begin{aligned}\label{eqm:main-nogo-weight prob}
  \ve &\geq (1 - F_\OO^\psi(\U) )\, \frac{W_\OO(\E)\Tr[({\mathrm{id}}\otimes\tilde\Theta(\M)(\psi)]}{p}.
\end{aligned}\end{equation}
This bound addresses a question of whether the no-go statement implied by Theorem~\ref{qip:thm1}, which says that perfect purification with $\ve=0$ is impossible for any channel with  $W_\OO(\E)>0$, remains valid in probabilistic cases. 
Eq.~\eqref{eqm:main-nogo-weight prob} implies that if $\Tr[{\mathrm{id}}\otimes\tilde\Theta(\M)(\psi)]>0$, the no-go theorem still holds.   
On the other hand, if $\Tr[{\mathrm{id}}\otimes\tilde\Theta(\M)(\psi)]=0$, meaning that the free part of $\E$ is completely cut off by the selective operation $\tilde\Theta$, then this does not give us any insight into $\ve$. 
This is actually a natural consequence because such a perfect purification is indeed possible, as we discuss in Supplementary Note~\ref{sec:prob} in detail.


\begin{acknowledgments}
We are grateful to Mark M.\ Wilde for helpful comments. We acknowledge useful discussions with Kun Fang and Zi-Wen Liu related to weight-based bounds and in particular their applications in probabilistic protocols. We also thank the authors for sharing a draft of their work~\cite{fang_2020-2} with us and agreeing to wait in order to make the concurrent posting of our preprints possible.

B.R.\ was supported by the Presidential Postdoctoral Fellowship from Nanyang Technological University, Singapore.
R.T.\ acknowledges the support of NSF, ARO, IARPA, AFOSR, the Takenaka Scholarship Foundation, the National Research Foundation (NRF) Singapore, under its NRFF Fellow programme (Award No. NRF-NRFF2016-02), and the Singapore Ministry of Education Tier 1 Grant 2019-T1-002-015. Any opinions, findings and conclusions or recommendations expressed in this material are those of the author(s) and do not reflect the views of National Research Foundation, Singapore.
\end{acknowledgments}

\bibliographystyle{apsrev4-1a}
\bibliography{bib_bartosz,bib_ryuji}

\titleformat{\section}
  {\centering\normalfont\bfseries\large}
  {Supplementary Note \thesection:}{.5em}{}

\clearpage

\onecolumngrid
\begin{center}
\vspace*{\baselineskip}
{\textbf{\large Fundamental limitations on distillation of quantum channel resources\\[2pt] --- Supplementary Information ---}}\\[1pt] \quad \\
\end{center}

\addtocontents{toc}{\protect\setcounter{tocdepth}{1}}

 
The Supplementary Information is structured as follows. In Supplementary Note~\ref{sec:setting}, we introduce in full technical detail our notation and the general setting of channel manipulation and quantum resources, as well as establish some auxiliary results which will help us in characterising the robustness and weight monotones. 
 In Supplementary Note~\ref{sec:nogos}, we prove our general no-go theorems on the performance of channel distillation protocols (Theorem~\ref{qip:thm1} in main text). We extend these insights in Supplementary Note~\ref{sec:many-copy}, where we study many-copy channel manipulation protocols, use our results to provide bounds on the distillation overhead (Theorem~\ref{qip:thm2} in main text), establish bounds applicable beyond distillation protocols, and apply our main findings to the asymptotic setting (Theorem~\ref{qip:thm3} in main text). We discuss explicit applications in Supplementary Note~\ref{sec:main-app} with a detailed discussion of how our bounds can shed light on important tasks in quantum communication and fault-tolerant quantum computation.
We conclude in Supplementary Note~\ref{sec:prob} with a discussion of probabilistic distillation protocols and an extension of our approach to this setting.
The Supplementary Information may, in some places, overlap with the content of the main text in order to be as self-contained as possible and to provide additional details and motivations.

\tableofcontents


\section{Setting}\label{sec:setting}

All of our discussions will take place in finite-dimensional Hilbert spaces. Given the Hilbert space of a system $A$, we will write $\LL(A)$ for the linear operators, and $\DD(A)$ for the density operators acting on this space. We use $\CPTP(A \to B)$ to denote the set of quantum channels, i.e.\ completely positive and trace-preserving (CPTP) maps from $\LL(A)$ to $\LL(B)$. We associate with each channel $\E : A \to B$ its Choi matrix $J_{\E} \coloneqq \idc \otimes \E(\Phi^+) \in \LL(RB)$, where $\Phi^+ = \sum_{i,j} \ket{ii}\!\bra{jj}$ is the unnormalised maximally entangled state in $\LL(RA)$ and $R \cong A$. The normalised Choi state is then $\wt J_{\E} \coloneqq J_\E / d_A$. We use $\< A, B \> = \Tr(A^\dagger B)$ for the Hilbert-Schmidt inner product between operators. All logarithms will be taken to the base 2, unless specified otherwise.

We will be concerned with schemes for transformations of quantum channels, that is, maps $\CPTP(A \to B) \to \CPTP(C \to D)$. Recall that such superchannels~\cite{chiribella_2008} can be written as $\TT(\E) = \M_{RB \to D} \circ ( \idc_{R} \otimes \E ) \circ \N_{C \to RA}$ where $\N, \M$ are some pre- and post-processing quantum channels, $\idc$ is the identity channel, and $R$ denotes an ancillary system (see Fig.~\ref{fig:sc-def} in the main text). 


\subsection{Resource theories}

As discussed in the main text, a given set of channels $\OO$ is designated as \textit{free} --- these are the channels which are freely available to use within the constraints of the given physical setting, and all channels outside of this set are the resourceful ones. To ensure broad applicability, we only make the natural assumptions that the set $\OO$ is closed and convex. When discussing channels acting on different spaces, we will assume that each space in consideration has its own associated set of free channels $\OO$, and for simplicity of notation we often do not explicitly indicate the relevant spaces.

The transformations of channels, and the issue of what exactly constitutes a free channel transformation, are approached in various ways~\cite{yuan_2019,liu_2020,liu_2019-1,gour_2019-1}. Often, one is interested in the manipulation of channels under superchannels of the form $\TT(\E) = \M \circ ( \idc \otimes \E ) \circ \N$ where the pre- and post-processing channels $\N$ and $\M$ are both free. In order to apply our results to more general settings, we will make no such assumption, and instead take the weakest possible constraint on the considered set of free superchannels: that for any $\M \in \OO$, we necessarily have $\Theta(\M) \in \OO$. We use $\SS$ to denote the set of all such resource-preserving superchannels. By studying these transformations, we will therefore obtain the most general bounds on the achievable performance of \textit{any} free channel manipulation protocol, since any valid choice of free transformations will necessarily be a subset of $\SS$.

The first type of resources that we consider is concerned with the investigation of intrinsic channel resources; this includes the various resource theories of quantum communication~\cite{devetak_2008,liu_2019-1,Fang2020max,takagi_2020,kristjansson_2020}, the related setting of quantum memories~\cite{rosset_2018,yuan_2020}, and quantum error mitigation~\cite{takagi2020optimal,Endo2021hybrid}. The other type is concerned with an underlying state-based resource and the manipulation of channels in order to extract or utilise the state resource more effectively; this includes, for instance, quantum entanglement~\cite{berta_2013,pirandola_2017,christandl_2017,wilde_2018,gour_2019,bauml_2019}, coherence~\cite{bendana_2017,diaz_2018-2,theurer_2019,Theure2020dynamical,saxena_2020}, thermodynamics~\cite{faist_2018}, or non-stabiliser (magic) states~\cite{seddon_2019,wang_2019-1}.
In the case of such state-based theories, we will use $\FF$ to denote the corresponding (convex and closed) set of free states. We will then study three levels of resources: the free states $\FF$, for example separable states in entanglement theory; the free channels $\OO$, for example the closure of the set of local operations and classical communication (LOCC) or other chosen classes of operations in entanglement manipulation; and $\SS$, which we will always take to be the set of all superchannels that preserve the chosen set $\OO$, and which can be understood as the most general way to manipulate free channels while preserving their freeness.

We remark that although we focus here on the framework accounting for all channels as valid objects under study, it is sometimes useful to impose certain constraints on the considered channels (both free and non-free), restricting the attention to CPTP maps possessing a certain structure. Such an approach can be used to investigate settings such as Bell nonlocality~\cite{Bell1966hidden,Brunner2014Bell,de_Vicente_2014_nonlocality,Wolfe2020quantifyingbell,Schmid2020typeindependent,Rosset2020type}, post-quantum correlations~\cite{Cirelson1980,Popescu1994PR,Geller_2014,Schmid2020post}, and quantum contextuality~\cite{KOCHEN1967,Cabello2014graph,Amaral2018wiring,Duarte2018contextuality}.
Our results can indeed be extended to those cases, which we discuss in more detail elsewhere~\cite{regula2020oneshot}.

\subsection{Benchmarking channel transformations and resources}

In its general formulation, the task of distillation can be understood as the transformation of a noisy resource channel $\E : A \to B$ into `pure' or `perfect' resources, which are represented by some target channel $\T : C \to D$. Our task then is to understand when one can achieve transformations of the type $\TT(\E) = \T$, with $\TT \in \SS$ being a free transformation.

However, in any practical setting, it is important to allow for the possibility of error in the transformation, reflecting the physical imperfections in the manipulation of channels. To account for this, we will employ the worst-case fidelity between two channels $\E, \F : A \to B$~\cite{belavkin_2005,gilchrist_2005}\footnote{The optimisation in the definition of $F(\E,\F)$ can be understood as being over arbitrary ancillary spaces $R$, but in fact it suffices to take $R \cong A$~\cite{belavkin_2005,gilchrist_2005}.}
\begin{equation}\begin{aligned}\label{eq:fidelity_channels}
  F(\E, \F) \coloneqq& \min_{\rho_{RA}} F(\idc \otimes \E (\rho), \idc \otimes \F (\rho))\\
  =& \min_{\psi_{RA}} F(\idc \otimes \E (\psi), \idc \otimes \F (\psi)),\\
\end{aligned}\end{equation}
where $F(\rho,\sigma) = \norm{\sqrt{\rho}\sqrt{\vphantom{\rho}\sigma}}{1}^2$ is the state fidelity, and in the second line the optimisation is constrained to pure input states. We will thus aim to achieve a transformation such that $F(\Theta(\E),\T) \geq 1-\ve$ for some small error $\ve$.
This fidelity metric is very closely related to the channel diamond norm $\norm{\cdot}{\dia}$~\cite{kitaev_1997} as $1 - \sqrt{F(\E,\F)} \leq \frac{1}{2} \norm{\E - \F}{\dia} \leq \sqrt{1 - F(\E,\F)}$~\cite{belavkin_2005,fuchs_1999}, meaning that our results can be straightforwardly adapted also to cases when diamond norm error is considered.

In order to quantify the resources of general quantum channels, we will employ two different resource measures which were first defined in resource theories of states. These quantities are also directly related to more general geometric concepts based on Hilbert's projective metric~\cite{bushell_1973,reeb_2011}, and specifically the Funk metric~\cite{funk_1929}.  Let us begin by recalling their definitions in the state-based setting\footnote{%
We note that our notational conventions are slightly different from ones typically encountered in the literature; in particular, the robustness is often defined as $\RFg(\rho)-1$ in our notation, while the resource weight is often defined as $1-W_\FF(\rho)$.}. First, the (generalised) \textit{robustness}~\cite{vidal_1999,brandao_2015} is given by
\begin{equation}\begin{aligned}\label{eq:rob_st_def}
  \RFg (\rho) \coloneqq \min \lset \lambda \bar \rho \leq \lambda \sigma,\; \sigma \in \FF \rset,
\end{aligned}\end{equation}
where the inequality is with respect to the cone of positive semidefinite operators. This can be understood as the least amount $\lambda$ such that, when mixed with another state $\omega$, the mixture $\rho + (\lambda-1) \omega = \lambda \sigma$ has $\sigma \in \FF$. 
The measure previously found a number of applications in different operational tasks~\cite{harrow_2003,brandao_2011,brandao_2015,napoli_2016,anshu_2018-1,takagi_2019-2,takagi_2019,liu_2019,regula_2020,seddon_2020,regula_2020-1}. In order to avoid pathological cases, throughout this work we assume that the robustness $R_\FF(\rho)$ is finite (that is, the minimum in Eq.~\eqref{eq:rob_st_def} exists) --- this can be guaranteed for any state $\rho$ by requiring that the set of free states $\FF$ contains at least one state of full rank, which is a natural requirement in physical theories of interest.

A closely related quantity is the \textit{resource weight}:
\begin{equation}\begin{aligned}
  W_{\FF} (\rho) \coloneqq \max \lset \lambda \bar \rho \geq \lambda \sigma,\; \sigma \in \FF \rset.
\end{aligned}\end{equation}
The resource weight generalises a measure of entanglement known as the best separable approximation~\cite{lewenstein_1998}, and was recently shown to be a meaningful quantifier of general resources~\cite{uola_2020-1,ducuara_2020,Ducuara2020multiobject}. 
The optimal $\lambda$ here can be understood the largest weight that a free state $\sigma$ can take in a convex decomposition $\rho = \lambda \sigma + (1-\lambda) \omega$ for an arbitrary state $\omega$. This monotone is a peculiar quantity: although it gives a faithful and strongly monotonic measure in any convex resource theory~\cite{regula_2018}, it exhibits unusual behaviour. Specifically, if there exists a free state $\sigma \in \supp(\rho)$, we necessarily have $W_\FF(\rho) > 0$; conversely, if a given $\rho$ does not have any free states in its support, the weight achieves the minimum value $W_\FF(\rho) = 0$. In particular, any resourceful pure state $\psi \notin \FF$ has $W_\FF(\psi) = 0$.

The measures can be extended to quantum channels~\cite{diaz_2018-2,wilde_2020,takagi_2019,liu_2019-1,uola_2020-1} by considering the corresponding Choi matrices:
\begin{equation}\begin{aligned}
  \ROg (\E) \coloneqq \min &\lset \lambda \bar J_\E \leq \lambda J_\M,\; \M \in \OO \rset,\\
  W_\OO (\E) \coloneqq \max &\lset \lambda \bar J_\E \geq \lambda J_\M,\; \M \in \OO \rset.
\end{aligned}\end{equation}
This can be equivalently understood as the optimisation over free channels $\M$ such that $\lambda \M - \E$ (in the case of the robustness) or $\E - \lambda \M$ (for the resource weight) is a completely positive map. Both of the quantities are valid resource monotones in that they satisfy monotonicity under any free superchannel $\Theta \in \SS$, i.e.\ $\ROg(\Theta(\E)) \leq \ROg(\E)$, with the resource weight obeying a reverse monotonicity: $W_\OO(\Theta(\E)) \geq W_\OO(\E)$. Other useful properties of the measures, such as their submultiplicativity, can also be established (see Supplementary Notes~\ref{app:properties} and~\ref{sec:many-copy}). The quantities have a simple structure which allows for an efficient computation as a semidefinite program (SDP) whenever the set of channels $\OO$ can be characterised using semidefinite constraints, which we will see to be the case in many relevant settings.

We will use one additional resource monotone in order to quantify how close a target channel $\T$ is to a free channel, and thus characterise how difficult it is to purify noisy channels into the target $\T$. We thus define the fidelity-based overlap measure
\begin{equation}\begin{aligned}
F_\OO(\T) \coloneqq& \max_{\M\in \OO} F(\T, \M).
\end{aligned}\end{equation}
In state-based theories, one analogously has
\begin{equation}\begin{aligned}
  F_\FF(\rho) \coloneqq \max_{\sigma \in \FF} F(\rho, \sigma).
\end{aligned}\end{equation}

\subsection{Properties of the robustness and resource weight}\label{app:properties}

We will use the fact that both the robustness and weight measures can be expressed in terms of the max-relative entropy $D_{\max} ( \rho \| \sigma ) \coloneqq  \log \inf \lset \lambda \bar \rho \leq \lambda \sigma \rset$~\cite{datta_2009}, whose generalisation to channels obeys some useful properties~\cite{wilde_2020}. In particular, defining $R_{\max} ( \rho \| \sigma) \coloneqq 2^{D_{\max}(\rho \| \sigma)}$, for any channels $\E, \F : A \to B$ one can define the optimised channel divergence~\cite{wilde_2020}
\begin{equation}\begin{aligned}
	R_{\max} ( \E \| \F ) \coloneqq \max_{\psi_{RA}} R_{\max} (\idc \otimes \E (\psi) \| \idc \otimes \F (\psi)), 
\end{aligned}\end{equation}
where $R \cong A$ and the maximisation can be regarded as being over the convex and compact set of density operators on $R$, with $\psi_{RA}$ thought of as a purification of a given $\rho_R$. Crucially, it holds that~\cite[Lemma 12]{wilde_2020}
\begin{equation}\begin{aligned}
	R_{\max} ( \E \| \F ) = R_{\max} ( J_\E \| J_\F),
\end{aligned}\end{equation}
that is, it suffices to consider the Choi matrices of the channels to evaluate the max-relative entropy.

We can then write the robustness measure as
\begin{equation}\begin{aligned}
	\ROg (\E) &= \min_{\M \in \OO} R_{\max} ( {J}_\E \| {J}_\M) \\
	&= \min_{\M \in \OO} \max_{\psi} R_{\max} (\idc \otimes \E (\psi) \| \idc \otimes \M (\psi)).
	\label{eq:robustness def opt}
\end{aligned}\end{equation}
In a very similar way, we notice that
\begin{equation}\begin{aligned}
	W_\OO(\E)^{-1} &= \min_{\M \in \OO} R_{\max} ( {J}_\M \| {J}_\E)\\
	&= \min_{\M \in \OO} \max_{\psi} R_{\max} (\idc \otimes \M (\psi) \| \idc \otimes \E (\psi)).
\end{aligned}\end{equation}
We will use the facts that the minimisation and maximisation problems in the above can be interchanged, which can be shown through the application of Sion's minimax theorem~\cite{sion_1958} similarly to how it was done in~Refs.~\cite{gour_2019-1,wang_2019-1}.

\begin{boxxed}{white}
\begin{lemma}\label{lemma:minimax}
We have that
\begin{equation}\begin{aligned}
	\ROg(\E) &= \max_{\psi} \min_{\M \in \OO} R_{\max} (\idc \otimes \E (\psi) \| \idc \otimes \M (\psi)),\\
	W_\OO(\E) &= \min_{\psi} \max_{\M \in \OO} R_{\max} (\idc \otimes \M (\psi) \| \idc \otimes \E (\psi))^{-1},\\
	F_\OO(\E) &= \min_{\psi} \max_{\M \in \OO} F (\idc \otimes \E (\psi), \idc \otimes \M (\psi)).\\
\end{aligned}\end{equation}
\end{lemma}
\end{boxxed}
\begin{proof}
The minimax optimisation can be regarded as being over two convex, compact sets: the given set of free channels $\OO$ and the set of density matrices on $\rho_R$, with the state $\psi_{RA}$ taken to be a purification of $\rho_R$. To apply Sion's minimax theorem to $\ROg$, we then need that the objective function is quasi-concave in $\rho_R$ and quasi-convex in $\M$~\cite{sion_1958}. Since $R_{\max}$ is the composition of $D_{\max}$ with the non-decreasing function $2^x$, quasi-convexity in $\M$ follows from the quasi-convexity of $D_{\max}$~\cite{datta_2009}, and quasi-concavity in $\rho_R$ follows from the concavity of $D_{\max}$ (see \cite[Prop.\ 13]{wang_2019-1} and \cite[Thm.\ 2]{gour_2019-1}).  The case of $W_\OO$ follows similarly. The fidelity $F(\cdot, \cdot)$ is known to be concave in $\M$~\cite{belavkin_2005} and can be shown to be convex in $\rho_R$ following~Refs.~\cite{gour_2019-1,wang_2019-1}.
\end{proof}

We proceed to establish other useful properties of the two measures.

\begin{boxxed}{white}
\begin{lemma}\label{lem:duals}
The robustness and weight measures can be expressed in their dual forms as
\begin{align}\label{eq:dual_measures}
	\ROg(\E) = \max &\lset \< X, J_{\E} \> \bar X \geq 0,\; \< X, J_\M \> \leq 1 \; \forall \M\in \OO \rset\nonumber\\
	= \max_{\psi} &\lset \< X, \idc \otimes \E (\psi) \> \bar X \geq 0,\; \< X, \idc \otimes \M(\psi) \> \leq 1 \; \forall \M\in \OO \rset,\\
	W_\OO(\E) = \min &\lset \< X, J_{\E} \> \bar X \geq 0,\; \< X, J_\M \> \geq 1 \; \forall \M\in \OO \rset\nonumber\\
	= \min_{\psi} &\lset \< X, \idc \otimes \E (\psi) \> \bar X \geq 0,\; \< X, \idc \otimes \M(\psi) \> \geq 1 \; \forall \M\in \OO \rset,\\
	\RFg(\rho) = \max &\lset \< X, \rho \> \bar X \geq 0,\; \< X, \sigma \> \leq 1 \; \forall \sigma \in \FF \rset,\\
	W_\FF(\rho) = \min &\lset \< X, \rho \> \bar X \geq 0,\; \< X, \sigma \> \geq 1 \; \forall \sigma \in \FF \rset.
\end{align}
\end{lemma}
\end{boxxed}
\begin{proof}
Follows from standard convex duality results (see e.g.~\cite{boyd_2004,takagi_2019}) and Lem.~\ref{lemma:minimax}.
\end{proof}

\begin{boxxed}{white}
\begin{lemma}
For any superchannel $\TT \in \SS$, it holds that
\begin{equation}\begin{aligned}
 	\ROg(\TT(\E)) &\leq \ROg(\E),\\
 	W_\OO(\TT(\E)) &\geq W_\OO(\E).
 \end{aligned}\end{equation} 
\end{lemma}
\end{boxxed}
\begin{proof}
Let $\M \in \OO$ be any channel such that $J_\E \leq \lambda J_\M$. By definition of the set $\SS$, we have that $\TT(\M) = \M' \in \OO$. Since $\TT$ preserves positivity, it holds that $J_{\TT(\E)} \leq \lambda J_{\TT(\M)} = \lambda J_{\M'}$, so $\M'$ is a valid feasible solution for $\ROg(\TT(\E))$, which concludes the proof. The case of $W_\OO$ follows analogously.
\end{proof}

The result below applies to state-based resource theories, and shows an exact relation between the channel-based and state-based measures.

\begin{boxxed}{white}
\begin{lemma}\label{lem:prep}
For any replacement channel $\R_\omega : A \to B$ defined as $\R_\omega(\cdot) = \Tr(\cdot) \omega$ with some fixed $\omega$, it holds that
\begin{equation}\begin{aligned}
	\ROg(\R_\omega) &\geq \RFg(\omega),\\
	W_\OO(\R_\omega) &\leq W_\FF(\omega),\\
	F_\OO(\R_\omega) &\leq F_\FF(\omega).
\end{aligned}\end{equation}
If the class of operations $\OO$ contains all replacement channels $\R_\sigma : A \to B$ with $\sigma \in \FF$, then equality holds in all of the above.
\end{lemma}
\end{boxxed}
\begin{proof}
Taking any feasible dual solution $X$ for $\RFg(\omega)$ in Eq.~\eqref{eq:dual_measures} and any $\tau \in \FF$, we can see that the operator $\tau^T \otimes X$ is feasible for $\ROg$ since $\< \tau^T \otimes X, J_\M \> = \< X, \M(\tau) \> \leq 1$ for any $\M \in \OO$ using the Choi-Jamiołkowski isomorphism. This immediately gives that
\begin{equation}\begin{aligned}
  \ROg &(\R_\omega) \\
  &\geq \max \lset \< \tau^T \otimes X, J_{\R_\omega} \> \bar \< X, \sigma \> \leq 1\; \forall \sigma \in \FF,\; X \geq 0 \rset\\
  &= \max \lset \< X, \omega \> \bar \< X, \sigma \> \leq 1\; \forall \sigma \in \FF,\; X \geq 0 \rset\\
  &= \RFg(\omega).
\end{aligned}\end{equation}
Now, assume that $\R_\sigma \in \OO \; \forall \sigma \in \FF$. From the definition of the robustness, we know that there exists a $\tau \in \FF$ such that $\omega \leq \RFg(\omega)\, \tau$. But then $\id \otimes \omega \leq \id \otimes \left[\RFg(\omega)\, \tau\right] = \RFg(\omega) \, J_{\R_\tau}$. Since $\R_{\tau} \in \OO$, this gives that $\ROg(\R_\omega) \leq \RFg(\omega)$.

The case of $W_\OO$ proceeds analogously.

For the fidelity, consider that
\begin{equation}\begin{aligned}
	F_\OO(\R_\omega) = \min_{\rho_{RA}} \max_{\M \in \OO} F(\idc \otimes \M(\rho), \idc \otimes \R_{\omega} (\rho)).
\end{aligned}\end{equation}
Taking the ansatz $\rho \otimes \tau$ for the input state, where $\rho \in \DD(R)$ is an arbitrary state and $\tau \in \FF$, we get
\begin{equation}\begin{aligned}
	F_\OO(\R_\omega) &\leq \max_{\M \in \OO}  F(\rho \otimes \M(\tau), \rho \otimes \omega)\\
	& = \max_{\M \in \OO} F(\M(\tau), \omega)\\
	&\leq \max_{\sigma \in \FF} F(\sigma, \omega)\\
	& = F_\FF(\omega),
\end{aligned}\end{equation}
where the second line follows from the data processing inequality for the fidelity, and the third line since $\M(\tau) \in \FF$ for any $\tau \in \FF$. For the converse inequality, assuming that any $\R_\sigma$ is in $\OO$ gives
\begin{equation}\begin{aligned}
	F_\OO(\R_\omega) &\geq \max_{\sigma \in \FF} \min_{\rho_{RA}} F(\idc \otimes \R_{\sigma}(\rho), \idc \otimes \R_{\omega} (\rho))\\
	&= \max_{\sigma \in \FF} \min_{\rho_{RA}} F([\Tr_A \rho] \otimes \sigma, [\Tr_A \rho] \otimes \omega)\\
	&= \max_{\sigma \in \FF} F( \sigma, \omega)\\
	& = F_\FF(\omega).
\end{aligned}\end{equation}
\end{proof}
%


\section{No-go theorems for quantum channel distillation}\label{sec:nogos}

The results of this section establish Theorem~\ref{qip:thm1} of the main text, explicitly divided into two cases depending on whether the target channel is a unitary or replacement channel.

\subsection{Unitary channel distillation}

Within the context of channel-based resource theories, it is natural to regard some unitary channel $\U(\cdot) = U \cdot U^\dagger$ as the target of distillation protocols.  Notice that in this case the measure $F_\OO$ simplifies to
\begin{equation}\begin{aligned}
F_\OO(\U) = \max_{\M\in \OO}\, \min_{\psi_{RC}} \<\idc \otimes \U (\psi), \idc \otimes \M (\psi) \>
\end{aligned}\end{equation}
since $\idc \otimes \U(\psi)$ is a pure state. We also recall that, in cases where this distance might not be easy to compute, one can instead choose to employ the more straightforward fidelity measure which uses the corresponding (normalised) Choi matrices:
\begin{equation}\begin{aligned}
 \wt F_\OO(\U) \coloneqq  \max_{\M\in \OO} \, F \left(\wt J_\U, \wt J_\M \right) = \max_{\M\in \OO} \<\wt J_\U, \wt J_\M \>.
\end{aligned}\end{equation}
Indeed, this quantity is the figure of merit in many communication scenarios~~\cite{barnum_2000,kretschmann_2004,buscemi_2010-2}, and can be alternatively expressed as the fidelity averaged over all input states~\cite{horodecki_1999-1,gilchrist_2005}. 

As our first main result, we then establish a general bound on the error necessarily incurred in any transformation of a channel under free superchannels.

{
\renewcommand\thetheorem{1(a)}
\begin{boxxed}{red}
\begin{theorem}\label{thm:main-nogo}
If there exists a free superchannel $\TT \in \SS$ such that $F(\Theta(\E), \U) \geq 1-\ve$ for some resourceful unitary channel $\U$, then
\begin{equation}\begin{aligned}\label{eq:main-nogo-rob}
  \ve &\geq 1 - F_\OO(\U) \, \ROg(\E)
\end{aligned}\end{equation}
and
\begin{equation}\begin{aligned}\label{eq:main-nogo-weight}
  \ve &\geq [1 - F_\OO(\U)]\, W_\OO(\E).
\end{aligned}\end{equation}
\end{theorem}
\end{boxxed}
\addtocounter{theorem}{-1}
}

Notice from the weight-based bound in Eq.~\eqref{eq:main-nogo-weight} that, when $\ve = 0$, the transformation $\E \to \U$ is impossible for any channel with $W_\OO(\E) > 0$. Recalling that $W_\OO(\E)$ is 0 if and only if $J_\E$ has no free Choi matrices $J_\M$ with $\M\in \OO$ in its support, we conclude that zero-error distillation is impossible whenever the given channel has any free channels in $\supp(J_\E)$. This recovers a no-go result of~Ref.~\cite{fang_2020} and extends insights from quantum state distillation~\cite{fang_2020,regula_2020}.

For clarity, we will divide the proof of Theorem~\ref{thm:main-nogo} into two parts and consider each bound separately. We begin with the robustness $\ROg$.

\begin{boxxed}{white}
\begin{proposition}\label{thm:nogo-rob-oneshot}
If there exists a free supermap $\TT \in \SS$ such that $\TT(\E) = \N$ for some channel $\N$ with $F(\N, \U) \geq 1-\ve$, then
\begin{equation}\begin{aligned}
  \ve &\geq 1 - \ROg(\E) \, F_\OO(\U).
\end{aligned}\end{equation}
\end{proposition}
\end{boxxed}
\begin{proof}
Recall from Lem.~\ref{lem:duals} that for any channel we have
\begin{equation}\begin{aligned}
  \ROg (\E) = \max_{\psi} \ROg^{\,\psi} (\E),
\end{aligned}\end{equation}
where 
\begin{equation}\begin{aligned}
  \ROg^{\,\psi} (\E) = \max \big\{ &\< X, \idc \otimes \E (\psi) \> \;\big|\;
   X \geq 0,\; \< X, \idc \otimes \M (\psi) \> \leq 1 \; \forall \M\in \OO \big\}.
\end{aligned}\end{equation}
Now, for the given target channel $\U: C \to C$, let $\psi^\star \in \DD(RC)$ with $R \cong C$ denote a state such that
\begin{equation}\begin{aligned}
   F_\OO(\U) = \max_{\M\in \OO} \< \idc \otimes \U (\psi^\star), \idc \otimes \M (\psi^\star) \>.
\end{aligned}\end{equation}
Notice now that $\frac{1}{F_\OO(\U)} \idc \otimes \U (\psi^\star)$ is a feasible witness to the dual formulation of the robustness $\ROg^{\,\psi^\star}$; specifically, we have that $\idc \otimes \U (\psi^\star) \geq 0$ and
\begin{equation}\begin{aligned}
\max_{\M\in \OO} \< \frac{\idc \otimes \U (\psi^\star)}{F_\OO(\U)}, \idc \otimes \M (\psi^\star) \> = 1
\end{aligned}\end{equation}
by definition of $\psi^\star$.
Using the monotonicity of the robustness under free superchannels, we then get
\begin{equation}\begin{aligned}
    \ROg(\E) &\geq \ROg(\TT(\E))\\
    &= \ROg(\N)\\
    &= \max_{\psi} \ROg^{\,\psi}(\N)\\
    &\geq \ROg^{\,\psi^\star}(\N)\\
    &\geq \< \idc \otimes \N (\psi^\star), \frac{\idc \otimes \U (\psi^\star)}{F_\OO(\U)} \>\\
    &\geq \frac{1-\ve}{F_\OO(\U)},
\end{aligned}\end{equation}
where in the last inequality we used the fact that
\begin{equation}\begin{aligned}\label{eq:unitary_nogo1}
  1-\ve &\leq  F(\N, \U) \\
  &\leq F(\idc \otimes \N (\psi^\star), \idc \otimes \U (\psi^\star))\\
  &= \< \idc \otimes \N (\psi^\star), \idc \otimes \U (\psi^\star) \>
\end{aligned}\end{equation}
where the first line is by assumption, second by definition of $F(\N,\U)$, and third since $\idc \otimes \U (\psi^\star) $ is rank one. 
\end{proof}

\begin{boxxed}{white}
\begin{proposition}\label{thm:nogo-weight-oneshot}
If there exists a free supermap $\TT \in \SS$ such that $\TT(\E) = \N$ for some channel $\N$ with $F(\N, \U) \geq 1-\ve$, then
\begin{equation}\begin{aligned}
  \ve &\geq (1 - F_\OO(\U) )\, W_\OO(\E).
\end{aligned}\end{equation}
\end{proposition}
\end{boxxed}

\begin{proof}
Using Lem.~\ref{lem:duals} we have once again that
\begin{equation}\begin{aligned}
  W_\OO(\E) = \min_{\psi} W^{\psi}_\OO (\E),
\end{aligned}\end{equation}
where
\begin{equation}\begin{aligned}
  W^{\psi}_\OO (\E)  =& \min \big\{ \< X, \idc \otimes \E(\psi) \> \;\big|\;
  X \geq 0,\; \< X, \idc \otimes \M(\psi) \> \geq 1 \; \forall \M\in \OO \big\}.
\end{aligned}\end{equation}
In a way similar to the proof of Prop.~\ref{thm:nogo-rob-oneshot}, we let $\psi^\star \in \DD(RC)$ be a state achieving the minimum for $F_\OO(\U)$. We then notice that $\id_{RC} -  \idc \otimes \U(\psi^\star) \geq 0$ and that for each $\M\in \OO$ we have
\begin{equation}\begin{aligned}
  &\< \idc \otimes \M(\psi^\star), \id - \idc \otimes \U(\psi^\star) \> \\
  &= 1 - \< \idc \otimes \M(\psi^\star), \idc \otimes \U(\psi^\star) \>\\
   &\geq 1 - \max_{\M\in \OO}  \< \idc \otimes \M(\psi^\star), \idc \otimes \U(\psi^\star) \>\\
   &= 1- F_\OO(\U),
\end{aligned}\end{equation}
which means that $\frac{1}{1-F_\OO(\U)} \left( \id - \idc \otimes \U(\psi^\star)\right)$ is a valid feasible dual solution for $W_\OO^\psi$.
Using the reverse monotonicity of $W_\OO(\Theta)$, we then get
\begin{equation}\begin{aligned}
  W_\OO(\E) &\leq W_\OO(\N)\\
  &\leq W^{\psi^\star}_\OO (\N)\\
  &\leq \< \idc \otimes \N(\psi^\star), \frac{\id - \idc \otimes \U(\psi^\star)}{1 - F_\OO(\U)} \>,
\end{aligned}\end{equation}
and using Eq.~\eqref{eq:unitary_nogo1} we conclude that
\begin{equation}\begin{aligned}
  W_\OO(\E) \leq \frac{1 - (1-\ve)}{1 - F_\OO(\U)} =  \frac{\ve}{1 - F_\OO(\U)}
\end{aligned}\end{equation}
which is precisely the statement of the Proposition.
\end{proof}


\subsection{State-based resources}

Although the idea of distilling noisy resources into pure ones makes sense in many physical settings, some resource theories are instead concerned with extracting state-based resources. Here, we will assume that there is an underlying set of free states $\FF$, and the operations $\OO$ are free operations in this theory. In such cases, the target channel in distillation can be chosen as the replacement channel $\R_\phi$, which substitutes any input with a target state: $\R_\phi (\cdot) = \Tr(\cdot) \phi$. A special case of such channels are preparation channels $\P_\phi$, which have trivial input and simply prepare a single copy of a chosen resourceful pure state $\phi$. To characterise the resourcefulness of the target channel, we consider the overlap
\begin{equation}\begin{aligned}
  F_\FF(\phi) = \max_{\sigma \in \FF} \<\phi, \sigma\>.
\end{aligned}\end{equation}
We then obtain an analogous bound for all transformations into replacement channels.
{
\renewcommand\thetheorem{1(b)}
\begin{boxxed}{red}
\begin{theorem}\label{thm:main-nogo-statebased}
If there exists a free superchannel $\TT \in \SS$ such that $F(\Theta(\E), \R_\phi) \geq 1-\ve$ for some resourceful pure state $\phi$, then
\begin{align}\label{eq:main-nogo-rob-statebased}
  \ve &\geq 1 - F_\FF(\phi) \, \ROg(\E),\\
  \ve &\geq [1 - F_\FF(\phi) ]\, W_\OO(\E).\label{eq:main-nogo-weight-statebased}
\end{align}
\end{theorem}
\end{boxxed}
\addtocounter{theorem}{-1}
}
Once again, the weight bound~\eqref{eq:main-nogo-weight-statebased} gives a no-go result: no resourceful pure state replacement channel can be distilled with $\ve=0$ from a channel $\E$ such that $W_\OO(\E) > 0$, that is, such that $\supp(J_\E)$ contains any free channels.

As a special case, the results apply also to the manipulation of states themselves, that is, transformations of resourceful quantum states under free transformations in the form of channels $\OO$. Thus, we obtain:

\begin{boxxed}{red}
\begin{corollary}\label{thm:main-nogo-states}
If there exists a free channel $\M \in \OO$ such that $F(\M(\rho), \phi) \geq 1-\ve$ for some resourceful pure state $\phi$, then
\begin{align}\ve &\geq 1 - F_\FF(\phi) \, \RFg(\rho),\label{eq:main-nogo-rob-states}\\
\label{eq:main-nogo-weight-states}
  \ve &\geq [1- F_\FF(\phi)] \, W_\FF(\rho).
\end{align}
\end{corollary}
\end{boxxed}
Here we note that an analogous robustness bound for states~\eqref{eq:main-nogo-rob-states} previously appeared in Ref.~\cite{regula_2020}.

Another approach to no-go results in the distillation of resources from quantum states was studied in Ref.~\cite{fang_2020}. Our new weight-based bound~\eqref{eq:main-nogo-weight-states} strictly improves on that result. Let us explicitly compare our result with the bound of Ref.~\cite{fang_2020}, which applies only to full-rank input states $\rho$, and is given by $ {\ve \geq [1- F_\FF(\phi)] \, \lambda_{\min}(\rho)}$, 
where $\lambda_{\min}$ denotes the smallest eigenvalue of $\rho$. First, our bounds require no assumption about the rank of the input state $\rho$, thus extending the applicability of the fundamental restrictions on quantum resource distillation. More importantly, our approach replaces the dependence on the eigenvalues of the input state with a bound which explicitly takes into consideration the resources contained in $\rho$, which provides more accurate restrictions. Further, for a full-rank state one can notice that $\lambda_{\min}$ can be written as $\min_{\omega \in \DD} \max \lset \lambda \bar \rho \geq \lambda \omega \rset$. From this it follows that $W_\FF(\rho) \geq \lambda_{\min}(\rho)$ in any resource theory (with the inequality typically strict), and so the weight-based bound in Eq.~\eqref{eq:main-nogo-weight-states} is tighter than the result of~Ref.~\cite{fang_2020}. Indeed, in Supplementary Note~\ref{sec:main-app} we will see this improvement to be significant.

As before, we split the proof for clarity.
\begin{boxxed}{white}
\begin{proposition}\label{thm:nogo-states-weight}
Consider the replacement channel $\R_\phi : C \to D$. If there exists a free supermap $\TT$ such that $\TT(\E) = \N$ for some channel $\N : C \to D$ with $F(\N, \R_\phi) \geq 1-\ve$, then
\begin{equation}\begin{aligned}
  \ve \geq \left(1- F_\FF(\phi) \right) W_\OO(\E).
\end{aligned}\end{equation}
When the input is a preparation channel $\P_\rho : \CC \to B$ and the target is the preparation channel $\P_\phi : \CC \to D$, the problem reduces to manipulating quantum states, and we have
\begin{equation}\begin{aligned}
  \ve \geq \left(1- F_\FF(\phi) \right) W_\FF(\rho).
\end{aligned}\end{equation}
\end{proposition}
\end{boxxed}
\begin{proof}
Noticing that, for a fixed $\E$, the function $W^{\psi}_\OO (\E)$ that we considered in Prop.~\ref{thm:nogo-weight-oneshot} is concave in $\psi$, we can relax the optimisation to write
\begin{equation}\begin{aligned}
  W_\OO(\E) = \min_{\rho_{RA}} W^{\rho}_\OO (\E)
\end{aligned}\end{equation}
since the minimum will be achieved on a pure state $\psi_{RA}$ anyway. Choosing $\rho^\star = \rho \otimes \tau$ for arbitrary $\rho \in \DD(R)$ and $\tau \in \FF$, we use the reverse monotonicity of $W_\OO$ to obtain
\begin{equation}\begin{aligned}
	W_\OO(\E) &\leq W_\OO(\N)\\
	&\leq W^{\rho^\star}_\OO(\N)\\
	&= \max \lset \lambda \bar \idc \otimes \N(\rho^\star) \geq \idc \otimes \M (\rho^\star) ,\; \M\in \OO \rset\\
	&= \max \lset \lambda \bar \N(\tau) \geq \M (\tau),\; \M\in \OO  \rset\\
	&\leq \max \lset \lambda \bar \N(\tau) \geq \sigma,\; \sigma \in \FF \rset\\
	&= W_\FF(\N(\tau))
\end{aligned}\end{equation}
where we used that $\M(\tau) \in \FF$. Notice now that $\id - \phi \geq 0 \text{ and } \< \id - \phi, \sigma \> \geq 1 - F_\FF(\phi)  \; \forall \sigma \in \FF$, which means that $\frac{\id - \phi}{1 - F_\FF(\phi)}$ is a valid feasible dual solution for $W_\FF(\N(\tau))$. Using the fact that
\begin{equation}\begin{aligned}
	1-\ve &\leq F(\N, \R_\phi)\\
	&\leq F(\idc \otimes \N(\rho^\star), \idc \otimes \R_\phi(\rho^\star))\\
	&= F(\N(\tau), \phi)\\
	&= \< \N(\tau), \phi \>,
\end{aligned}\end{equation}
we then have
\begin{equation}\begin{aligned}
	W_\FF(\N(\tau)) &\leq \< \N(\tau), \frac{\id - \phi}{1 - F_\FF(\phi)} \>\\
	&\leq \frac{\ve}{1-F_\FF(\phi)}
\end{aligned}\end{equation}
which concludes the proof.

The above reduces to the case of quantum states when the input and target are preparation channels, since this constrains any output of the transformation to also be a preparation channel.
\end{proof}

\begin{boxxed}{white}
\begin{proposition}\label{thm:nogo-states-rob}
If there exists a free supermap $\TT$ such that $\TT(\E) = \N$ for some channel $\N: C \to D$ with $F(\N, \R_\phi) \geq 1-\ve$, then
\begin{equation}\begin{aligned}
  \ve \geq 1- F_\FF(\phi)\, \ROg(\E).
\end{aligned}\end{equation}
If $\E = \P_\rho$ and the system $C$ is trivial, we have that
\begin{equation}\begin{aligned}
  \ve \geq 1- F_\FF(\phi)\, \RFg(\rho).
\end{aligned}\end{equation}
\end{proposition}
\end{boxxed}
\begin{proof}
Analogously as in Prop.~\ref{thm:nogo-states-weight}, we choose $\rho^\star = \rho \otimes \tau$ for some $\tau \in \FF$ to get
\begin{align*}
  \ROg(\E) &\geq \ROg(\N)\\
  &\geq \ROg^{\rho^\star} (\N)\\
  &\geq \RFg(\N(\tau)) \tag{\stepcounter{equation}\theequation}\\
  &\geq  \< \N(\tau), \frac{\phi}{F_\FF(\phi)} \>\\
  &\geq  \frac{1 - \ve}{F_\FF(\phi)},
\end{align*}
where in the first line we used the monotonicity of $\ROg$, in the second line we used that $\ROg^\psi$ is convex in $\psi$ so we can optimise over mixed states, and in the fourth line we used that $\phi \geq 0 \text{ and } \< \phi, \sigma \> \leq F_\FF(\phi)  \; \forall \sigma \in \FF$ which means that $\frac{\phi}{F_\FF(\phi)}$ is a valid feasible dual solution for the robustness $\RFg(\omega)$.
\end{proof}


\section{Many-copy transformations and protocols beyond distillation}\label{sec:many-copy}

\subsection{Bounds for many-copy manipulation}
Recall that the most general physically realisable manipulation protocols involving multiple quantum channels are dubbed \textit{quantum processes}~\cite{chiribella_2008-1,chiribella_2009,gutoski_2007,chiribella_2013,oreshkov_2012,araujo_2017}. We will be interested in processes which transform multiple uses of a quantum channel into one output map --- we thus understand a quantum process $\CT$ as any transformation such that $\CT(\N_1, \ldots, \N_n) \in \CPTP$ for any $\N_1, \ldots, \N_n \in \CPTP$, and we take the set of free quantum processes as
\begin{equation}\begin{aligned}
\SS_{(n)} \coloneqq \lset \CT \bar \CT(\M_1,\ldots,\M_n) \in \OO \;\, \forall \M_1, \ldots, \M_n \in \OO \rset.
\label{eq:free comb}
\end{aligned}\end{equation}

Depending on the given resource theory, different ways to manipulate multiple channels might be of interest. For instance, when the theory is closed under tensor product, i.e. $\M, \M' \in \OO \Rightarrow \M\otimes \M' \in \OO$, then any free protocol which manipulates $n$ copies of a channel in parallel as $\TT(\E^{\otimes n})$ is a free quantum process. Similarly, when the theory is closed under composition, i.e. $\M, \M' \in \OO \Rightarrow \M\circ \M' \in \OO$, then any sequential protocol of the form $\TT_n(\E) \circ \cdots \circ \TT_1(\E)$ belongs to the set $\SS_{(n)}$. However, a general channel theory need not be closed under tensor product or composition --- for instance, the tensor product of operations which preserve separability in entanglement theory is not always separability preserving itself~\cite{horodecki_2001}. Therefore, to take into consideration the most general way of manipulating quantum channels allowed by the constraints of the given resource theory, we employ the formalism of free quantum processes $\SS_{(n)}$. By considering such transformations, we can establish fundamental bounds on the performance of any adaptive, multi-copy protocol for manipulating channels or states.

{
\renewcommand\thetheorem{2}
\begin{boxxed}{red}
\begin{theorem}\label{thm:main-nogo-parallel}
Given any distillation protocol $\CT \in \SS_{(n)}$ --- parallel, sequential, or adaptive, with or without a definite causal order --- which transforms $n$ uses of a channel $\E$ to some target unitary $\U$ up to accuracy $\ve > 0$, it necessarily holds that
\begin{align}\label{eq:main-copies-weight}
  n &\geq \log_{1/W_\OO(\E)}\, \frac{1- F_\OO(\U)}{\ve},\\
  n &\geq \log_{\ROg(\E)}\, \frac{1-\ve}{F_\OO(\U)}.\label{eq:main-copies-rob}
\end{align}
Analogously, when the target channel is a replacement channel $\R_\phi$ which prepares a pure state $\phi$, we have
\begin{align}\label{eq:main-copies-weight-statebased}
  n &\geq \log_{1/W_\OO(\E)}\, \frac{1- F_\FF(\phi)}{\ve},\\
  n &\geq \log_{\ROg(\E)} \,\frac{1-\ve}{F_\FF(\phi)}.\label{eq:main-copies-rob-statebased}
\end{align}
\end{theorem}
\end{boxxed}
\addtocounter{theorem}{-1}
}

The result applies also to the case of state manipulation, where we obtain that the number of copies of a state needed to perform the distillation $\M(\rho^{\otimes n}) \to \phi$ up to error $\ve$ must obey
\begin{align}\label{eq:main-copies-weight-states}
  n &\geq \log_{1/W_\FF(\rho)} \frac{1- F_\FF(\phi)}{\ve}, \quad n \geq \log_{\RFg(\rho)} \frac{1-\ve}{F_\FF(\phi)}.
\end{align}
This, again, improves on the bound obtained in Ref.~\cite{fang_2020} for resource theories of states.

The two bounds exhibit very different properties. Intuitively, we see that the weight-based bound \eqref{eq:main-copies-weight} will perform better for small $\ve$, establishing in particular that $n$ must scale as $\log(1/\ve)$ as $\ve \to 0$ for distillation to be possible.
On the other hand, the robustness-based bound \eqref{eq:main-copies-rob} increases in performance with decreasing $F_\OO(\U)$, i.e.\ with increasing resourcefulness of $\U$. One can use both of these insights to one's advantage when aiming to obtain more accurate bounds. For instance, a straightforward way to decrease $F_\OO(\U)$ is, instead of considering $\U$ as a target channel, to consider several copies of it. In practice, such an approach can be employed in block distillation protocols, which could provide a more efficient way of purifying the given resource. As long as one can bound or compute the quantity $F_\OO({\U^{\otimes m}})$ --- which we will shortly see to be possible in relevant cases --- this leads to an immediate improvement in the robustness bound (cf.~\cite{seddon_2020}).

The result of Thm.~\ref{thm:main-nogo-parallel} is a consequence of a general sub- or supermultiplicativity result for the monotones $R_\OO, W_\OO$.
\begin{boxxed}{red}
\begin{theorem}\label{thm:submult}
Consider a collection of $n$ channels $(\E_1, \ldots, \E_n)$. For any free protocol $\CT \in \SS_{(n)}$ it holds that
\begin{equation}\begin{aligned}\label{eq:weight-comb-submult}
  W_\OO\left(\CT(\E_1,\ldots,\E_n)\right) \geq \prod_i W_\OO(\E_i).
\end{aligned}\end{equation}
and
\begin{equation}\begin{aligned}\label{eq:rob-comb-submult}
   \ROg\left(\CT(\E_1,\ldots,\E_n)\right) \leq \prod_i R_\OO(\E_i).
\end{aligned}\end{equation}

In particular, for any free protocol $\CT \in \SS_{(n)}$ which takes $n$ uses of a quantum channel $\E$ to another quantum channel $\E'$, i.e. $\CT(\E, \ldots, \E) = \E'$, it necessarily holds that
\begin{equation}\begin{aligned}
  n \geq \frac{\log \ROg(\E')}{\log \ROg(\E)},\quad n \geq \frac{\log W_\OO(\E')}{\log W_\OO(\E)},
\end{aligned}\end{equation}
where in the second inequality we take $\log 0 = -\infty$ and assume that $W_\OO(\E')$ and $W_\OO(\E)$ are not both 0.
\end{theorem}
\end{boxxed}
Notice that this establishes bounds which go beyond distillation protocols, and imposes constraints on arbitrary manipulation of channels. In particular, the robustness-based bound gives a general restriction on the capabilities of channel dilution, i.e. transformations $\U \to \E'$ and $\R_\phi \to \E'$, which can be understood as simulating the action of a channel $\E'$ by employing the pure channel $\U$ or $\R_\phi$.

An interesting difference between the bounds of Thm.~\ref{thm:submult} emerges in the case when $\E'$ is a unitary or pure replacement channel with $W_\OO(\E') = 0$. Here, the weight-based bound shows that increasing the number of uses of a channel cannot allow perfect distillation when $W_\OO(\E) \in (0,1]$, strengthening the no-go result of Thm.~\ref{thm:main-nogo}. However, the bound based on $W_\OO$ does not provide information on the distillation of channels with $W_\OO(\E) = 0$ --- notably, unitary-to-unitary transformations --- while the robustness-based bound can also be applied in such cases. This complements the no-go results provided by $W_\OO$ and can reveal errors even in transformations where the weight bound becomes trivial.

\begin{proof}[\textbf{\textup{Proof of Thm.~\ref{thm:submult}}}]
We consider $W_\OO$ first. For each $\E_i$, let $\mu_i \in \RR_+$ and $\M_i \in \OO$ be such that $J_{\E_i} \geq \mu_i J_{\M_i}$. Using the $n$-linearity of the transformation $\CT$, we can expand
\begin{align*}
  \CT(\E_1,\ldots,\E_n) &= \CT(\E_1 - \mu_1 \M_1, \E_2, \ldots, \E_n) + \CT(\mu_1 \M_1, \E_2, \ldots, \E_n)\\
  &= \CT(\E_1 - \mu_1 \M_1, \E_2, \ldots, \E_n) + \CT(\mu_1 \M_1, \E_2 - \mu_2 \M_2, \E_3, \ldots, \E_n) + \CT(\mu_1 \M_1, \mu_2 \M_2, \E_3, \ldots, \E_n)\\
  &\vdotswithin{=} \tag{\stepcounter{equation}\theequation}\\
  &= \CT(\E_1 - \mu_1 \M_1, \E_2, \ldots, \E_n) + \ldots + \CT(\mu_1 \M_1, \ldots, \mu_{n-1} \M_{n-1}, \E_{n} - \mu_n \M_{n}) + \CT(\mu_1 \M_1, \ldots, \mu_n \M_n).
\end{align*}

By the positivity of $\CT$, each term on the right-hand side is positive semidefinite, and so
\begin{align*}
  0 &\leq \CT(\E_1,\ldots,\E_n) - \CT(\mu_1 \M_1, \ldots, \mu_n \M_n)\\
  &= \CT(\E_1,\ldots,\E_n) - \left(\prod_i \mu_i\right) \CT(\M_1, \ldots, \M_n)\tag{\stepcounter{equation}\theequation}\\
  &= \CT(\E_1,\ldots,\E_n) - \left(\prod_i \mu_i\right) \,\M'
\end{align*}
for some $\M' \in \OO$ due to the fact that $\CT$ is a free quantum process. Choosing $\M_i$ as optimal channels such that $\mu_i = W_\OO(\E_i)$, we have that $\prod_i W_\OO(\E_i)$ is a feasible optimal value for $W_{\OO}\left(\CT(\E_1,\ldots,\E_n)\right)$, which is precisely Eq.~\eqref{eq:weight-comb-submult}.

The case of the robustness $\ROg$ is shown analogously: recalling that $\ROg(\E)$ is given by the least coefficient such that $J_{\E_i} \leq \mu_i J_{\M_i}$ for $\M_i \in \OO$, we use the positivity and $n$-linearity of $\CT$ to show Eq.~\eqref{eq:rob-comb-submult} by the same argument.
\end{proof}

\begin{proof}[\textbf{\textup{Proof of Thm.~\ref{thm:main-nogo-parallel}}}]
Let $\CT \in \SS_{(n)}$ be any general distillation protocol such that $F\left(\CT(\E^{\times n}), \U \right) \geq 1-\ve$, where we use $\E^{\times n}$ to denote the $n$-tuple $(\E, \E, \ldots, \E)$ representing $n$ uses of $\E$.  Using Thm.~\ref{thm:submult} and Prop.~\ref{thm:nogo-weight-oneshot} gives
\begin{equation}\begin{aligned}
  W_\OO(\E)^{n} &\leq W_\OO(\CT(\E^{\times n}))\\
  &\leq  \frac{\ve}{1-F_\OO(\U)}.
\end{aligned}\end{equation}
Taking logarithm of both sides of the equation and recalling that $W_\OO(\E) \in [0,1]$, we get
\begin{equation}\begin{aligned}
n \geq \frac{\log \frac{\ve}{1-F_\OO(\U)}}{\log W(\E)}
\end{aligned}\end{equation}
as claimed.

The case of the robustness $\ROg$ follows analogously by using Thm.~\ref{thm:submult} and Prop.~\ref{thm:nogo-rob-oneshot}.

The state-based case follows in the same way, using Thm.~\ref{thm:submult} and Prop.~\ref{thm:nogo-states-weight} or Prop.~\ref{thm:nogo-states-rob}.
\end{proof}

A consequence of Thm.~\ref{thm:submult} is the sub- or supermultiplicativity of the measures under tensor product and composition.

\begin{boxxed}{white}
\begin{corollary}
If the resource theory is closed under tensor product, i.e. $\M, \M' \in \OO \Rightarrow \M\otimes \M' \in \OO$, then
\begin{equation}\begin{aligned}
  \ROg\left(\E_1 \otimes \E_2\right) &\leq \ROg(\E_1) \ROg(\E_2),\\
  W_\OO\left(\E_1 \otimes \E_2\right) &\geq W_\OO(\E_1) W_\OO(\E_2).\\
\end{aligned}\end{equation}
If the resource theory is closed under composition, i.e. $\M, \M' \in \OO \Rightarrow \M\circ \M' \in \OO$, then
\begin{equation}\begin{aligned}
  \ROg\left(\E_1 \circ \E_2\right) &\leq \ROg(\E_1) \ROg(\E_2),\\
  W_\OO\left(\E_1 \circ \E_2\right) &\geq W_\OO(\E_1) W_\OO(\E_2).\\
\end{aligned}\end{equation}
\end{corollary}
\end{boxxed}


\subsection{Asymptotic rates of distillation}\label{sec:rates}

To understand the ultimate limitations on transforming a given state or channel, one can study the maximal rate at which the conversion can be performed with an asymptotic number of channel uses, allowing for conversion error that vanishes asymptotically. Specifically, given two channels $\E, \F$, we define the maximal achievable rate of transformation under any adaptive protocol as
\begin{equation}\begin{aligned}
  &r_{\rm adap}(\E \!\to\! \F) \!\coloneqq\! \sup \lsetr r\! \barr \lim_{n \to \infty} \sup_{\CT_n \in \SS_{(n)}} \!F\left( \CT_n(\E^{\times n}), \F^{\otimes \floor{rn}}\right) \!=\! 1\! \rsetr\!,
 \label{eq:rate def}
\end{aligned}\end{equation}
Again, the transformations that we consider include both parallel and sequential protocols as relevant special cases, and thus provide an upper bound for both.
Although \eqref{eq:rate def} characterises the conversion rate with the perfect fidelity in asymptotic limit, it does not give insights into how robust the rate is against error. 
To characterise the maximum rate at which the asymptotic transformation is possible with some non-vanishing error, we define the strong converse rate as~\cite{hayashi2016quantum} 
\begin{equation}\begin{aligned}
  &r^\dagger_{\rm adap}(\E \!\to\! \F) \!\coloneqq\! \sup \lsetr r\! \barr \limsup_{n \to \infty} \sup_{\CT_n \in \SS_{(n)}} \!\!F\left( \CT_n(\E^{\times n}), \F^{\otimes \floor{rn}}\right) \!>\! 0\! \rsetr\!.
 \label{eq:strong converse rate def}
\end{aligned}\end{equation}
In other words, as soon as a rate exceeds the strong converse rate, the fidelity necessarily goes to 0. This places a threshold for the achievable distillation rates, even when non-zero error is allowed.

Applying our result in Thm.~\ref{thm:main-nogo-parallel} allows us to obtain a general bound on the rate of distillation protocols.
{
\renewcommand\thetheorem{3(a)}
\begin{boxxed}{red}
\begin{theorem}\label{thm:main-rate-bound}
If the target channel $\U$ satisfies $F_\OO(\U^{\otimes n}) = F_\OO(\U)^n$, then we have a strong converse bound as
\begin{equation}\begin{aligned}
  r_{\rm adap}(\E \to \U) \leq r^\dagger_{\rm adap}(\E\to\U)\leq \frac{\log \ROg(\E)}{\log F_\OO(\U)^{-1}}.
\label{eq:strong converse adaptive}
\end{aligned}\end{equation}

Alternatively, if the target is a replacement channel $\R_\phi$ such that $F_\FF(\phi^{\otimes n}) = F_\FF(\phi)^n$, then
\begin{equation}\begin{aligned}
 r_{\rm adap}(\E \to \R_{\phi}) \leq r^\dagger_{\rm adap}(\E \to \R_{\phi})\leq \frac{\log \ROg(\E)}{\log F_\FF(\phi)^{-1}}.
 \label{eq:strong converse adaptive state}
\end{aligned}\end{equation}

In the above, $F_\OO(\U)$ (or $F_\FF(\phi)$) can be replaced with any other multiplicative quantity $G_\OO(\U)$ s.t. $F_\OO(\U) \leq G_\OO(\U)$ (or $F_\FF(\phi) \leq G_\FF(\phi)$), and the results hold analogously.
\end{theorem}
\end{boxxed}
\addtocounter{theorem}{-1}
}
\begin{proof}
Let $r$ be any achievable rate at error threshold $\ve \in [0,1)$, that is, assume that there exists a sequence $\{\CT_n\}_n$ of free quantum processes such that $1 - F\left( \CT_n(\E^{\times n}), \U^{\otimes \floor{rn}}\right) \eqqcolon \ve_n$ with $\liminf_{n\to\infty}\ve_n=\ve < 1$.
From Thm.~\ref{thm:main-nogo-parallel}, for each $n$ we have that
\begin{equation}\begin{aligned}
  n \log \ROg(\E) &\geq \log(1-\ve_n) - \log F_\OO\left(\U^{\otimes \floor{rn}}\right)\\
  &\geq \log(1-\ve_n) + rn \log F_\OO(\U)^{-1}
\end{aligned}\end{equation}
where we used the multiplicativity of $F_\OO(\U)$. Dividing by $n$ and taking $\limsup_{n\to\infty}$ in both sides gives the claimed result.

When the target is a replacement channel, we can follow the proof of Lem.~\ref{lem:prep} to get
\begin{equation}\begin{aligned}
F_\OO(\R_{\phi}^{\otimes n}) \leq F_\FF(\phi^{\otimes n})
\end{aligned}\end{equation}
and the rest of the proof proceeds analogously.
\end{proof}

In some cases, it is reasonable to restrict our attention to parallel protocols (illustrated in Fig.~\ref{fig:sc-many-copy}(a) of the main text), which are indeed how many communication and channel manipulation schemes are often considered~\cite{barnum_2000,kretschmann_2004,Cooney2016constant}.
To separately characterise this scenario, we define the rate of transformation with parallel protocols as
\begin{equation}\begin{aligned}
  &r_{\rm par}(\E \!\to\! \F) \!\coloneqq\! \sup \lsetr r\! \barr \lim_{n \to \infty} \sup_{\TT_n \in \SS} \!F\left( \TT_n(\E^{\otimes n}), \F^{\otimes \floor{rn}}\right) \!=\! 1\! \rsetr\!,
 \label{eq:rate parallel def}
\end{aligned}\end{equation}
and its strong converse rate $r_{\rm par}^\dagger(\E \to \F)$ analogously. 
Then, we can get better strong converse bounds for parallel protocols by suitably `smoothing' the definition of the robustness over all channels within a small distance of the original input $\E$~\cite{diaz_2018-2,Fang2020max,liu_2019-1,gour_2019-1}. 
We thus define the regularised log-robustness (max-relative entropy) $D_{\OO}^\infty(\E)\coloneqq \lim_{\delta\to 0}\limsup_{n\to\infty}\frac{1}{n}\log R_\OO^{\delta}(\E^{\otimes n})$ where $R_\OO^\delta(\E)\coloneqq \min_{F(\tilde\E,\E)\geq 1-\delta}R_\OO(\tilde\E)$, and use it as follows.

{
\renewcommand\thetheorem{3(b)}
\begin{boxxed}{red}
\begin{theorem}\label{thm:main-rate-bound parallel}
If the target channel $\U$ satisfies $F_\OO(\U^{\otimes n}) = F_\OO(\U)^n$, then
\begin{equation}\begin{aligned}
  r_{\rm par}(\E \to \U) \leq r_{\rm par}^\dagger(\E\to\U)\leq \frac{D_\OO^\infty(\E)}{\log F_\OO(\U)^{-1}}.
\label{eq:strong converse par}
\end{aligned}\end{equation}

Alternatively, if the target is a preparation channel $\P_\phi$ such that $F_\FF(\phi^{\otimes n}) = F_\FF(\phi)^n$, then we have a strong converse bound as
\begin{equation}\begin{aligned}
 r_{\rm par}(\E \to \P_{\phi}) \leq r_{\rm par}^\dagger(\E \to \P_{\phi})\leq \frac{D_\OO^\infty(\E)}{\log F_\FF(\phi)^{-1}}.
 \label{eq:strong converse par state}
\end{aligned}\end{equation}
In the above, $F_\OO(\U)$ (or $F_\FF(\phi)$) can be replaced with any other multiplicative quantity $G_\OO(\U)$ s.t. $F_\OO(\U) \leq G_\OO(\U)$ (or $F_\FF(\phi) \leq G_\FF(\phi)$), and the results hold analogously.
\end{theorem}
\end{boxxed}
\addtocounter{theorem}{-1}
}

For both of the results of this section to be valid, we require the multiplicativity of the channel fidelity of the target unitary $\U$ or state $\phi$. In practice, the bound can be relaxed by using the Choi fidelity $\wt F_\OO(\U) \geq F_\OO(\U)$, which is often easier to show to be multiplicative. Although the multiplicativity condition might not hold in full generality, the majority of relevant resource theories do indeed satisfy it. This includes (cf.~Table~\ref{tab:resources} in the main text):
\begin{itemize}
    \item The theories of communication and quantum memories with the choice of the target channel $\U = \idc$, where we have $ F_\OO(\idc^{\otimes n}) =  F_\OO(\idc)^n$ (see Supplementary Note~\ref{sec:main-comm}).
    \item The theory of magic for multi-qubit quantum channels. Here, the fidelity $F_\OO(\R_\phi)$ of any replacement channel $\R_\phi$ reduces to $F_\FF(\phi)$ which is known to be multiplicative as long as $\phi$ is a state of up to three qubits~\cite{bravyi_2019}. We will furthermore show in Supplementary Note~\ref{sec:main-magic} the multiplicativity of the channel fidelity $F_\OO(\U)$ of any 1-, 2-, or 3-qubit diagonal unitary channel from the third level of the Clifford hierarchy, which are common choices of target channels in practical settings.
    \item The theory of magic for multi-qudit quantum channels, where a multiplicative bound for the fidelity of any target replacement channel $\R_\phi$ is given by the min-thauma of magic~\cite{wang_2020,wang_2019-1}. The fidelity $F_\FF$ itself can also be shown to be multiplicative in some relevant cases (see Supplementary Note~\ref{sec:magic_qudits}).
    \item Other state-based theories such as entanglement, coherence, athermality, and purity, where the quantities $F_\FF(\phi)$ are known to be multiplicative for any target pure state $\phi$.
\end{itemize}
Our results therefore give broadly applicable bounds for the asymptotic performance of general distillation protocols.

In order to prove Thm.~\ref{thm:main-rate-bound parallel}, we first show the following lemma.
\begin{boxxed}{white}
\begin{lemma}
\label{lem:fidelity and diamond}
 For any two channels $\E,\N$ and unitary $\U$, it holds that
 \bal
  |F(\E,\U)-F(\N,\U)|\leq \sqrt{1-F(\E,\N)}.
 \eal
\end{lemma}
\end{boxxed}
\begin{proof}
We first show $|F(\E,\U)-F(\N,\U)|\leq \frac{1}{2}\|\E-\N\|_\diamond$.
Suppose $F(\E,\U)\geq F(\N,\U)$ without loss of generality. Then,
 \begin{align*}
  &|F(\E,\U)-F(\N,\U)| \\
  &= F(\E,\U)-F(\N,\U)\\
  &=\min_\phi \braket{\phi|(\idc\otimes\U^\dagger)\idc\otimes\E(\phi)(\idc\otimes\U)|\phi} \\&\quad - \min_\phi \braket{\phi|(\idc\otimes\U^\dagger)\idc\otimes\N(\phi)(\idc\otimes\U)|\phi}\\
  &\leq \braket{\tilde\phi_\U|\idc\otimes\E(\tilde\phi)-\idc\otimes\N(\tilde\phi)|\tilde\phi_\U} \tag{\stepcounter{equation}\theequation}\\
  &\leq \max_\phi \max_{0\leq M \leq \id} \<M,\idc\otimes\E(\phi)-\idc\otimes\N(\phi)\>\\
  &=\frac{1}{2}\|\E-\N\|_\diamond
 \end{align*}
where in the third line we defined $\tilde\phi$ to be the minimiser of the second term in the second line and also defined $\tilde\phi_\U=\idc\otimes\U(\tilde\phi)$, and in the fourth line we used that $0\leq\dm{\tilde\phi_\U}\leq \id$.

The proof is concluded by noticing $\frac{1}{2}\|\E-\N\|_\diamond\leq \sqrt{1-F(\E,\N)}$ \cite{belavkin_2005,fuchs_1999}.

\end{proof}

\begin{proof}[\textbf{\textup{Proof of Thm.~\ref{thm:main-rate-bound parallel}}}]

 Let $r$ be any achievable rate at error threshold $\ve \in [0,1)$, that is, assume that there exists a sequence $\{\TT_n\}_n$ of free superchannels such that $1 - F\left( \TT_n(\E^{\otimes n}), \U^{\otimes \floor{rn}}\right) \eqqcolon \ve_n$ with $\liminf_{n\to\infty}\ve_n=\ve < 1$.
 Let $\delta$ with $0<\delta\leq 1$ be some constant and $\N_n$ be a channel such that $\log R_\OO^\delta(\TT_n(\E^{\otimes n}))=\log R_\OO(\N_n)$. 
 Then, using Lem.~\ref{lem:fidelity and diamond},
 
  \bal
  F(\N_n,\U^{\otimes \lfloor rn \rfloor})&\geq F(\TT_n(\E^{\otimes n}),\U^{\otimes \lfloor rn \rfloor})-\sqrt{1-F(\N_n,\TT_n(\E^{\otimes n}))}\\
  &\geq 1-\ve_n - \sqrt{\delta}.
 \eal
 
Since doing nothing is a valid free comb, $\N_n$ can be transformed to $\U^{\otimes \lfloor rn \rfloor}$ with fidelity  $1-\ve_n-\sqrt{\delta}$ for free. 
Applying Thm.~\ref{thm:main-nogo-parallel}, we get 
\bal
 \ve_n + \sqrt{\delta} &\geq 1-R_{\OO}(\N_n)  F_\OO\left(\U^{\otimes \lfloor rn \rfloor} \right)\\
 &= 1-R_{\OO}(\N_n)  F_\OO(\U)^{\lfloor rn \rfloor}
 \label{eq:each error bound}
\eal 
where we used the assumption of the multiplicativity $F_\OO\left(\U^{\otimes m} \right)=F_\OO(\U)^m$.
Let $\{m_n\}_n$ be a subsequence such that $\lim_{n\to\infty} \ve_{m_n}=\ve$.
Then, since $0\leq \ve <1$, there exist some integer $N$ and positive real number $\delta'$ such that $0< \ve_{m_n} + \sqrt{\delta'}<1$ for any $n>N$.
Using \eqref{eq:each error bound}, we get for $n>N$ and $\delta\leq\delta'$ that

\bal
 r&\leq \frac{\lfloor rm_n \rfloor +1}{m_n}  \\
 &\leq \frac{1}{\log F_\OO(\U)^{-1}}\frac{1}{m_n} \log R_{\OO}(\N_{m_n}) + \frac{1}{\log F_\OO(\U)^{-1}}\frac{1}{m_n}\log\frac{1}{1-(\ve_{m_n}+\sqrt{\delta})}+\frac{1}{m_n}\\
 &= \frac{1}{\log F_\OO(\U)^{-1}}\frac{1}{m_n} \log R_{\OO}^{\delta}(\TT(\E^{\otimes m_n})) + \frac{1}{\log F_\OO(\U)^{-1}}\frac{1}{m_n}\log\frac{1}{1-(\ve_{m_n}+\sqrt{\delta})}+\frac{1}{m_n}\\
 &\leq \frac{1}{\log F_\OO(\U)^{-1}}\frac{1}{m_n} \log R_{\OO}^{\delta}(\E^{\otimes m_n}) + \frac{1}{\log F_\OO(\U)^{-1}}\frac{1}{m_n}\log\frac{1}{1-(\ve_{m_n}+\sqrt{\delta})}+\frac{1}{m_n}
\label{eq:adaptive converse inequality}
\eal
where in the last inequality we used the monotonicity of max-relative entropy measure under free superchannels~\cite{takagi_2020}.
Noting that
\bal
\limsup_{n\to\infty}\frac{1}{m_n} \log R_{\OO}^{\delta}(\E^{\otimes m_n})\leq \limsup_{n\to\infty}\frac{1}{n} \log R_{\OO}^{\delta}(\E^{\otimes n}),
\eal
we take $\lim_{\delta\to 0}\limsup_{n\to\infty}$ in both sides of \eqref{eq:adaptive converse inequality} to get 
\bal
 r\leq \frac{D_\OO^\infty(\E)}{\log F_\OO(\U)^{-1}},\ \forall \ve \in [0,1)
\eal
showing that the quantity in the right hand side is a strong converse bound. 
The state case can be shown analogously.

\end{proof}



\section{Applications}\label{sec:main-app}


\subsection{Quantum communication}\label{sec:main-comm}

A central goal of quantum communication is to enable reliable transmission of quantum states to another party through a noisy channel. 
A way of accomplishing this task is to apply encoding (respectively, decoding) operations to input (output) states so that one can offset the effect of noise, and there has been an enormous amount of work to investigate the ultimate limit on how much information can be reliably sent~\cite{lloyd_1997,barnum_1998,bennett_2002,devetak_2005-1,holevo_2001,Bennett2014reverse,hayashi2016quantum,wang_2019-3,tomamichel_2016-1,tomamichel_2017,berta_2017,kaur_2017,pirandola_2017,pirandola_2019-1,fang_2019-1}.
By understanding this problem as the purification of a noisy quantum channel to the identity channel by means of some allowed free channel transformation, we can directly apply the results of our work.

To encompass the most general communication setting, we consider assisted communication scenarios, where both parties may share some correlations (e.g.\ entanglement) or are able to perform some joint operations (e.g.\ communicate classically) in order to enhance their quantum communication capabilities~\cite{bennett_1996,bennett_2002}.
The given type of assistance can be specified by the set of free superchannels in our framework, and our general results can be applied by considering $\OO$ as the set of free channels that are preserved by the free superchannels. 
We will thus investigate the maximum size of quantum systems that can be reliably sent using free superchannels $\SS_\mA$ where $\mA$ denotes the type of assisting operations. 
Of particular interest in the theory of quantum communication is the asymptotic capacity where many uses of the given channel are considered.  
We define the $\mA$-assisted adaptive quantum capacity
$Q_{\mA,{\rm adap}}$ as the rate at which the given channel can be converted to the qubit identity channel $\idc_2$ using the given choice of superchannels $\SS=\SS_\mA$, that is, $r_{\rm adap}(\E\rightarrow\idc_2)$ in the notation of Supplementary Note~\ref{sec:rates}. 
Another important quantity is the strong converse capacity, similarly defined as $Q_{\mA,{\rm adap}}^\dagger(\E)\coloneqq r^\dagger(\E\rightarrow\idc_2)$. We note that our methods naturally apply also to the setting of generalised resource theories of communication~\cite{kristjansson_2020}, where transformations with an indefinite causal order are allowed.

We also consider a simpler scenario where multiple uses of the given channel are applied in parallel.
We define the $\mA$-assisted quantum capacity under a parallel communication protocol as $Q_\mA(\E)\coloneqq r_{{\rm par}}(\E\rightarrow\idc_2)$ and similarly the strong converse capacity as $Q_\mA^\dagger(\E)\coloneqq r_{{\rm par}}^\dagger(\E\rightarrow\idc_2)$.

Considering assisted capacities is insightful in that they serve as upper bounds for capacities with weaker assistance, including the unassisted quantum capacity $Q(\E)$ whose single-letter formula is not known and thus hard to analyse in general~\cite{cubitt_2015}.
Below, we apply our method to two assisted settings that are often considered in the literature. 


\subsubsection{No-signalling assistance}\label{sec:main-comm-ns}

An important class of assisted communication codes are ones where the sending and receiving party are allowed to share no-signalling correlations~\cite{Leung2015NS,Duan2016nosignalling,wang_2017,wang_2019-3,wang_2019-2,takagi_2020}. 
It is the communication setting where the superchannel $\Theta:(A'\rightarrow B) \rightarrow (A\rightarrow B')$ is realised by a bipartite no-signalling channel $\Pi_{\rm NS}:AB\rightarrow A'B'$ satisfying
\bal
\Tr_{A'}\Pi_{\rm NS}(\rho_{A}^{(0)}\otimes\rho_{B})=\Tr_{A'}\Pi_{\rm NS}(\rho_{A}^{(1)}\otimes\rho_{B}) \\
\Tr_{B'}\Pi_{\rm NS}(\rho_{A}\otimes\rho_{B}^{(0)})=\Tr_{B'}\Pi_{\rm NS}(\rho_{A}\otimes\rho_{B}^{(1)}) \label{eq:BtoA no-signalling}
\eal
for any state $\rho_{A}$, $\rho_{B}$, and any pair of states $\{\rho_{A}^{(j)}\}_{j=0}^1$, $\{\rho_{B}^{(j)}\}_{j=0}^1$.
We denote the set of superchannels realised by no-signalling channels as $\SS_{\rm NS}$ and call them no-signalling superchannels.
Intuitively, such transformations do not create a side channel that allows for free communication.
In this setting, the free channels are the channels which are useless for transmitting any information: this set is formed by the replacement channels ${\OO_R}\coloneqq \lset \R_\sigma:A'\rightarrow B\sbar\R_\sigma(\cdot) = \Tr(\cdot) \sigma,\ \sigma\in\DD(B)\rset$, which simply replace the input with a fixed output state. Indeed, any superchannel $\SS_{\rm NS}$ preserves the set of replacement channels. Motivated by this, Ref.~\cite{takagi_2020} introduced a resource theory of communication with the set of free superchannels $\SS\coloneqq\lset \TT \bar \TT(\M) \in \OO_R \; \forall \M\in \OO_R \rset$, and showed that $\SS$ actually coincides with the set of no-signalling superchannels, i.e., $\SS=\SS_{\rm NS}$. 
Many-copy transformations with no-signalling assistance are then formally defined as the transformations by free combs as in~\eqref{eq:free comb}, which map multiple constant channels to a constant channel. 
A class of free many-copy transformations are the quantum feedback-assisted adaptive protocols where the receiving party is allowed to send a part of their quantum system back to the sender, followed by the application of a no-signalling bipartite channel between the channels uses. This includes the quantum feedback-assisted communication with entanglement assistance discussed in Refs.~\cite{Bowen2004feedback,Cooney2016constant}.

To evaluate $F_{\OO_R}(\idc_d)$, we will use the following lemma to relate it with the Choi fidelity $\wt F_{\OO_R} (\idc_d) = \max_{\M\in \OO_R} \<\wt J_{\idc_d}, \wt J_\M \>$, which is easier to compute. This is conceptually similar to an argument we made in the proof of Prop.~\ref{thm:nogo-rob-oneshot}.
\begin{boxxed}{white}
\begin{lemma}\label{lemma:golden_channel}
For any unitary channel $\U$, it holds that $R_\OO(\U) \geq F_\OO(\U)^{-1} \geq \wt F_\OO(\U)^{-1}$. In particular, if $R_\OO(\U) = \wt F_\OO(\U)^{-1}$, then $F_\OO(\U) = \wt F_\OO(\U)$.
\end{lemma}
\end{boxxed}
\begin{proof}
We have that
\begin{equation}\begin{aligned}
	R_{\OO}(\U) &= \max_{\psi} R_{\max}(\idc \otimes \U(\psi)\|\idc\otimes \M^\star(\psi))\\
		&\geq R_{\max}(\idc \otimes \U(\psi^\star)\|\idc\otimes \M^\star(\psi^\star))\\
		&\geq \< \idc \otimes \U(\psi^\star), \frac{\idc \otimes \U(\psi^\star)}{F(\U,\M^\star)} \>\\
		&=F(\U,\M^\star)^{-1}\\
		&\geq F_{\OO}(\U)^{-1}\\
		&\geq \wt F_\OO(\U)^{-1}.
\end{aligned}\end{equation}
where in the first line we chose $\M^\star\in\OO$ as an optimal channel for $R_{\OO}$ in~\eqref{eq:robustness def opt}, in the second line we chose $\psi^\star$ as a state such that $F(\U,\M^\star) = F(\idc \otimes \U(\psi^\star), \idc \otimes \M^\star(\psi^\star))$, in the third line we used that $\frac{\idc \otimes \U(\psi^\star)}{F(\U,\M^\star)}$ is a feasible solution for the dual form of $R_{\max}$ in \eqref{eq:dual_measures}, and the last line is by definition.
\end{proof}
For the theory of no-signalling channels, a direct computation tells us that $\wt F_{\OO_R}(\idc_d) = \frac{1}{d^2}$. Coupled with the fact $R_{\OO_R}(\idc_d)=d^2$~\cite{takagi_2020}, we obtain from Lem.~\ref{lemma:golden_channel} that $F_{\OO_R}(\idc_d) = \frac{1}{d^2}$.
This can alternatively be seen by a direct calculation: 
\bal
 F_{\OO_R}(\idc_d)&=\min_{\psi}\max_{\M\in\OO_R} \<\psi,\idc\otimes\M(\psi)\>\\
 &=\min_{\{\lambda_i\},\{\ket{i}\}}\max_{\sigma}\sum_{i=0}^{d-1} \lambda_i^2 \braket{i|\sigma|i} \\
 &= \min_{\{\lambda_i\},\{\ket{i}\}}\lambda_{\max}^2 \\
 &= \frac{1}{d^2},
\eal
where the first line is due to Lem.~\ref{lemma:minimax}, and in the second line we used that a bipartite pure state $\psi$ allows for a Schmidt decomposition $\ket{\psi}=\sum_{i=0}^{d-1} \lambda_i \ket{i}\ket{i}$ where $\lambda_{\max}\coloneqq \max_i \lambda_i$ in the third line. This implies in particular that $F_{\OO_R}(\idc_2^{\otimes n})=\frac{1}{2^{2n}}$.

First, it is insightful to see what the bounds of Thm.~\ref{thm:main-nogo} tell us about one-shot transformations $\E \to \idc_d$. Here, the maximal fidelity achievable under no-signalling codes can be computed with an SDP~\cite{Leung2015NS}. We will also use the fact that $R_{\OO_{\mathrm{R}}}$ has been computed analytically for some simple channels~\cite{Fang2020max}, and $W_{\OO_{\rm R}}$ can also be computed in such cases.

For instance, for the depolarising channel $\D_p (\rho) = (1-p) \rho + p \frac{\id}{d}$ we get $W_{\OO_{\rm R}} (\D_p) = p$, and the robustness and weight-based bounds are actually equal: we have $\ve \geq p (d^2 - 1)/d^2$. In fact, the bounds match the achievable fidelity, meaning that Thm.~\ref{thm:main-nogo} quantifies the error in the one-shot transformation $\D_p \to \idc_d$ under $\SS_{\rm NS}$ \textit{exactly}. On the other hand, channels such as the qubit dephasing channel $\mathcal{Z}_p(\rho) = (1-p) \rho + Z \rho Z$ or the amplitude damping channel $\mathcal{N}_\gamma$ have no constant channels in their support, and thus the weight bound becomes useless with $W_{\OO_{\rm R}}(\mathcal{Z}_p) = W_{\OO_{\rm R}}(\N_\gamma) = 0$. However, the robustness bound gives $\ve \geq \frac{1}{2} - |p-\frac{1}{2}|$ for the dephasing channel and $\ve \geq \frac{1}{4}(2+\gamma-2\sqrt{1-\gamma})$ for the amplitude damping channel, and these bounds, again, exactly equal the fidelity achievable under no-signalling assistance. Thus, we see that the error bounds of Thm.~\ref{thm:main-nogo} can be tight in relevant cases. We refer to Fig.~\ref{fig:ns-oneshot} of the main text for an explicit comparison.

We now apply our result to get insights into asymptotic capacity. From Thm.~\ref{thm:main-rate-bound}, we immediately obtain a strong converse bound for adaptive capacity as 
\bal
 Q_{\rm NS, adap}(\E)\leq Q_{\rm NS, adap}^\dagger(\E)\leq \frac{1}{2}\log R_{\OO_R}(\E).
\eal

For the capacity with parallel protocols, we can use the asymptotic equipartition property of Ref.~\cite{Fang2020max} which showed that $D^\infty_{\OO_R}(\E)=I(\E)$~\footnote{Although the asymptotic equipartition property was shown for the diamond norm smoothing, it can be straightforwardly extended to the fidelity smoothing using the relation between the diamond norm and fidelity~\cite{belavkin_2005,fuchs_1999}.}, where $I(\E)$ is the channel mutual information defined as $I(\E_{A'\rightarrow B})\coloneqq\max_{\phi_{EA'}}I(\idc\otimes\E_{A'\rightarrow B}(\phi_{EA'}))$ with $I(\idc\otimes\E(\phi_{EA'}))$ being the mutual information between $E$ and $B$. Combining this result with Thm.~\ref{thm:main-rate-bound parallel}, we obtain 
\bal
Q_{\rm NS}(\E)\leq Q_{\rm NS}^\dagger(\E)\leq \frac{1}{2} I(\E).
\eal
Together with the known achievable capacity $Q_{\rm NS}(\E)\geq\frac{1}{2}I(\E)$~\cite{bennett_2002,Leung2015NS}, this ensures the strong converse property of no-signalling--assisted quantum capacity $Q_{\rm NS}(\E)=Q_{\rm NS}^\dagger(\E)=\frac{1}{2}I(\E)$.

Importantly, no-signalling--assisted communication includes entanglement-assisted communication as a subclass, and thus the strong converse bound for the no-signalling--assisted capacity serves as that for the entanglement-assisted capacity. 
We note that the strong converse property for both settings was previously shown in several different ways~\cite{Bennett2014reverse,Berta2011reverse,Gupta2015strong,Cooney2016constant,Fang2020max,takagi_2020}; our general resource-theoretic approach provides an alternative perspective that employs the asymptotic equipartition property of Ref.~\cite{Fang2020max} to show this important relation.


\subsubsection{PPT and separable assistance}\label{sec:main-comm-ppt}

We now study another approach to quantum communication, where assistance by local operations and (two-way) classical communication is allowed. Since such operations are extremely difficult to characterise mathematically~\cite{chitambar_2014}, various relaxations are frequently employed. A particularly useful set of operations, amenable to both an efficient numerical computations as well as a simplified analytical characterisation, is the set of PPT (positive partial transpose) codes~\cite{rains_2001,Leung2015NS}. In the general case of bipartite channels $A B \to A' B'$, a map $\M$ is called PPT if the partial transpose of $J_\M$ across the $AA' : BB'$ bipartition is positive~\cite{rains_2001}. Analogously, a map is separable if $J_\M$ is a separable operator~\cite{rains_1997}. Quantum communication through a channel $\E: A \to B$ with PPT (or SEP) assistance is then defined as allowing Alice and Bob to perform joint PPT (or separable) operations between their successive channel uses (see e.g.~\cite{kaur_2017}). The goal is then, again, to simulate a number of uses of the identity channel $\idc_2$. Using $\OO_\PPT$ to denote the set of all PPT channels and analogously for $\OO_\SEP$, for both classes of channels we have~\cite{shimony_1995,rains_2001}
\begin{equation}\begin{aligned}
\wt F_{\OO_\PPT}(\idc_d) = \wt F_{\OO_\SEP}(\idc_d) = \max_{\M \in \OO_\SEP} \< \wt J_{\idc_d}, \wt J_\M \> = \frac{1}{d}.
\end{aligned}\end{equation}
Using now the fact that $R_{\OO_\PPT}(\idc_d) = R_{\OO_\SEP}(\idc_d) = d$~\cite{yuan_2020}, from Lem.~\ref{lemma:golden_channel} we conclude that $F_{\OO_\PPT} (\idc_d) = F_{\OO_\SEP} (\idc_d) = \frac{1}{d}$, and in particular $F_{\OO}(\idc_2^{\otimes n}) = \frac{1}{2^n}$.

In order to describe these operations in our framework, one can notice that any PPT code, understood as a superchannel acting on a channel $\E: A \to B$, preserves the set of PPT channels~\cite{gour_2019,bauml_2019}. This motivates us to define the class of PPT-preserving superchannels: $\SS_\PPT \coloneqq \lset \TT \bar \TT(\M) \in \OO_\PPT \; \forall \M\in \OO_\PPT \rset$.
Any bound obtained for such superchannels will then upper bound the capabilities of PPT codes. Similarly, one defines $\SS_\SEP \coloneqq \lset \TT \bar \TT(\M) \in \OO_\SEP \; \forall \M\in \OO_\SEP \rset$ as the separability-preserving superchannels. More general PPT- and separability-preserving quantum processes are defined analogously.

Using this approach, our results immediately provide several bounds on the capabilities of both PPT and separable codes in assisting quantum communication, both in the non-asymptotic and asymptotic settings. We discuss in particular the applications to upper bounding quantum capacity.

Denoting by $Q_{\SEP,{\rm adap}}(\E)$ the quantum capacity of a channel $\E$ assisted by general (adaptive, two-way) separable codes, Thm.~\ref{thm:main-rate-bound} gives
\begin{equation}\begin{aligned}
   Q_{\SEP,{\rm adap}}(\E) \leq Q^\dagger_{\SEP,{\rm adap}}(\E) \leq \log R_{\OO_\SEP}(\E).
\end{aligned}\end{equation}
Here, $\log R_{\OO_\SEP}$ is the so-called max-relative entropy of entanglement of a channel, and our result recovers exactly the strong converse bound of~Ref.~\cite{christandl_2017} (see below for a clarification regarding this quantity). Our approach then provides a remarkable simplification of the proof methods used to show this relation. Note that the quantity $R_{\OO_\SEP}$ was recently also considered in~Ref.~\cite{yuan_2020} in the different context of quantifying the memory of a quantum channel, where it was evaluated for a number of cases, and in particular shown to be computable exactly for low-dimensional channels: as long as $d_A \leq 3$ and $d_B = 2$, we have $R_{\OO_\SEP}(\E) = \max \{1, \norm{J_\E}{\infty} \}$. This provides an analytically computable strong converse bound for the capacities of all qubit channels.

Similarly, for PPT codes we get
\begin{equation}\begin{aligned}
   Q_{\PPT,{\rm adap}}(\E) \leq Q^\dagger_{\PPT,{\rm adap}}(\E) \leq \log R_{{\OO_\PPT}}(\E).
\end{aligned}\end{equation}
This gives a general, SDP-computable bound on the quantum capacity and is similar to the results of~Refs.~\cite{wang_2019-3,berta_2017}, although the robustness (max-relative entropy)-based quantities in those works optimise over a larger set of maps than our $R_{{\OO_\PPT}}$. Specifically, following the observations in~Refs.~\cite{bauml_2019,fang_2019-1}, the precise result of~Refs.~\cite{wang_2019-3,berta_2017} can be recovered exactly in the framework of this work if one considers the class of \textit{completely} PPT-preserving superchannels~\cite{gour_2019} instead of PPT-preserving superchannels $\SS_\PPT$ as we have done here. We also note that $R_{{\OO_\PPT}}(\E) = \max \{1, \norm{J_\E}{\infty} \}$ for any channel with $d_B = 2$~\cite{yuan_2020}.

Of note is the fact that previous approaches to establishing such bounds typically relied on so-called amortised monotones~\cite{takeoka_2014,christandl_2017,berta_2017,kaur_2017} defined at the level of quantum states, while our method is based on resource measures directly at the level of channels. Our results not only provide a streamlined approach to recovering the specialised bounds of Refs.~\cite{christandl_2017,berta_2017}, but also reveal these approaches to be a part of a broad resource-theoretic framework. 
A notable technical difference between our approach and the previous methods is that we do not need to make any assumptions about the structure of the communication protocol, and in particular we do not need to impose that it is constructed by sequentially composing free channel transformations, as previous works did.

Let us also note that the generality of our methods allows them to be immediately applicable to related settings such as secret key agreement~\cite{devetak_2005} (giving bounds on the private capacities and recovering a result of~Ref.~\cite{christandl_2017}) and bidirectional quantum communication~\cite{bennett_2003} (giving bounds on quantum and private capacities, recovering results similar to~Ref.~\cite{bauml_2018}).

Before proceeding further, let us clarify a difference between our measures $R_{\OO_\SEP}$, $R_{\OO_\PPT}$ and related quantities found in literature. 
Max-relative entropy quantities for channels are often defined in the literature with respect to a set of quantum states, rather than trace-preserving maps, which might seem different from our definitions. For instance, the max-relative entropy of entanglement of a channel was originally defined as~\cite{christandl_2017}
\begin{equation}\begin{aligned}
  \wt{D}_{{\SEP}}(\E) \coloneqq \max_{\psi} \min \lset \log \lambda \bar \idc \otimes \E(\psi) \leq \lambda \sigma,\; \sigma \in \FF_{\SEP} \rset
\end{aligned}\end{equation}
where $\sigma \in \FF_{\SEP}$ optimises over all separable states --- not necessarily valid Choi matrices --- which can be understood as an optimisation over a cone of completely positive maps. However, it is not difficult to show that this definition is equal to ours.
\begin{boxxed}{white}
\begin{lemma}\label{lem:rob_cptni}
Let $\OO$ denote either ${\OO_\SEP}$ or ${\OO_\PPT}$, and define $\wt \OO$ as the cone of separable operators on $A:B$ or PPT operators on $A:B$, respectively.
Then, we can equivalently optimise over completely positive and trace non-increasing maps in the definition of the robustness $\ROg$. That is, for $\E : A \to B$ we have
\begin{equation}\begin{aligned}\label{eq:rob_cptni}
  &\ROg(\E) \\
  &= \min \lset \lambda \bar J_{\E} \leq \lambda J_{\M},\; \M\in \wt\OO,\; \Tr_B J_\M= \id \rset\\
  &= \min \lset \lambda \bar J_{\E} \leq \lambda J_{\wt \M},\; \wt\M\in \wt\OO,\; \Tr_B J_{\wt\M} \leq \id \rset\\
    &= \min \lset \norm{\Tr_B J_{\wt\M}}{\infty} \bar J_{\E} \leq J_{\wt\M},\; \wt\M\in \wt\OO \rset\\
    &= \max_{\psi} \min \lset \lambda \bar \idc \otimes \E(\psi) \leq \lambda J_{\wt M} ,\; \wt \M \in \wt\OO,\; \Tr(J_{\wt M}) = 1 \rset.
\end{aligned}\end{equation}
\end{lemma}
\end{boxxed}
\begin{proof}
Let $\ROg'$ denote the quantity in the second line of \eqref{eq:rob_cptni}. Clearly, $\ROg(\E) \geq \ROg'(\E)$. To see the converse inequality, let $\wt\M\in \wt\OO$ be an optimal map in the optimisation for $\ROg'$. By definition, we have that $C \coloneqq \id_A - \Tr_B J_{\wt\M}\geq 0$. Define the completely positive and trace-preserving map $\M$ through $ J_{\M} = J_{\wt\M}+ C \otimes \frac{\id_B}{d_B}$. As $C \otimes \frac{\id_B}{d_B} \in \wt \OO$, convexity of $\wt \OO$ gives that $\M \in \wt \OO \cap \CPTP$. This means that $J_\E \leq \ROg'(\E) J_{\wt\M} \leq \ROg'(\E) J_{\M}$, hence $\ROg(\E) \leq \ROg'(\E)$ and so the two quantities must be equal. The third line is a simple restatement of the second, and the fourth equality has been previously shown in~Ref.~\cite[Lemma 7]{berta_2017}.
\end{proof}
This result shows a bound of~Ref.~\cite[Thm. 5.2]{christandl_2017} (also~\cite[Prop. 1]{rigovacca_2018}) to be tight.

In the formalism of `resource-generating power' found in Refs.~\cite{bendana_2017,liu_2020,liu_2019-1}, the result can be understood as the fact that the entanglement-generating power of a channel $\E : A\to B$ in terms of the generalised robustness of entanglement $R_{\FF_\SEP}$ is equal to the robustness $R_{\OO_\SEP}(\E)$.


\subsection{Quantum gate synthesis and magic states}\label{sec:main-magic}

To realise reliable computation under the presence of noise, it is essential to encode quantum states into higher-dimensional spaces by a quantum error correcting code and carry out computation within the logical subspace in a fault-tolerant manner~\cite{Preskill1998reliable}. 
Among the whole class of quantum gates, Clifford gates are known to be easier to implement on many codes of interest~\cite{Steane1996_7qubit,Shor1995_9qubit,gottesman_1998,Fowler2012surface,Bombin2006color} whereas logical implementation of non-Clifford gates results in a large overhead cost~\cite{gottesman_1999,Bravyi2005magic}.
The cost is often quantified in terms of the T gates, with T being the single-qubit unitary $T = \mathrm{diag}(1, e^{i\pi/4})$~\cite{Bravyi2005magic} which enables universal quantum computation together with Clifford gates and can be considered as the most difficult to implement. In particular, any circuit composed of only Clifford gates can be efficiently simulated on a classical computer~\cite{gottesman_1998}, but a circuit built with the Clifford+T gate set has a simulation time exponential in the number of T gates~\cite{aaronson_2004}. Much effort has therefore been devoted to studying the optimal ways to realise circuits with minimal T-gate cost~\cite{amy_2014,gosset_2014,campbell_2017}.

However, choices of relevant operations other than the commonly used set of Clifford+T gates can be made. Indeed, the Clifford unitaries are not the only types of channels that admit efficient classical simulation schemes: this holds true also for stabiliser operations~\cite{Veitch2014resource} and an even larger class of completely stabiliser-preserving channels~\cite{seddon_2019}. Such extensions have shown potential to significantly improve the cost of implementing quantum circuits in terms of the T-gate count~\cite{selinger_2013,jones_2013,gidney_2018} as well as the simulation cost of noisy circuits~\cite{Veitch2014resource,pashayan_2015,bravyi_2016,howard_2017,bravyi_2019,seddon_2020}. It is therefore of interest to understand the limitations on manipulating quantum circuits through the most general means.

The task of synthesising gates and circuits is often realised through magic state distillation~\cite{Bravyi2005magic}, which aims to prepare clean magic states --- states that cannot be obtained with stabiliser operations alone --- and use them to implement costly quantum gates through the scheme of state injection~\cite{gottesman_1999}. Magic state distillation can provide feasible ways to synthesise gates and circuits~\cite{campbell_2017}, and investigating the precise relations between distillation and gate synthesis is highly important in paving the way to fault-tolerant quantum computation~\cite{duclos-cianci_2013,campbell_2017-1,campbell_2017}.

The resource theory of magic was thus introduced to understand the limitations of manipulating and distilling states using different free operations~\cite{Veitch2014resource,howard_2017,seddon_2019}. We can apply our general results to this formalism to obtain fundamental limitations on both magic state distillation and more general gate manipulation protocols which act the level of quantum channels directly. In this resource theory, stabiliser operations $\OO_{\rm STAB}$ are built through Clifford gates, Pauli measurements, and preparations of ancillary states in the computational basis~\cite{Veitch2014resource}.  Then, the set of stabiliser states $\FF_{\rm STAB}$ is defined to be the states that can be created by such operations, which coincides with the convex hull of eigenstates of Pauli operators. The larger set of completely stabiliser-preserving operations~\cite{seddon_2019} is defined as all channels whose Choi state satisfies $\wt J_{\E} \in \FF_{\rm STAB}$, which allows the computation of quantifiers such as $\ROg$ and $W_\OO$ as semidefinite programs (although their size scales superexponentially with the number of qubits). By choosing this set as our free operations $\OO$, we ensure that all of our bounds will apply also to smaller sets of transformations, which includes all channels typically considered in the study of magic. The channel manipulation protocols $\SS$ that we consider include, in particular, any pre- and post-processing of the channel with other completely stabiliser-preserving channels.

Of particular importance are the channels which can be realised through state injection, such as unitary gates from the third level of the Clifford hierarchy $\C_3$~\cite{gottesman_1999}, that is, unitaries which map Pauli operators to Clifford gates. This includes many gates of practical relevance such as the T gate, the controlled-phase gate CS, the controlled-controlled-Z gate CCZ, or the Toffoli gate. All channels which can be implemented in this way obey the following relation, which we will prove and generalise in the next section.
\begin{boxxed}{white}
\begin{lemma}\label{lem:main-stab-seizable}
Let $\N: A \to B$ be any channel which can be implemented by state injection, that is, for which there exists a free state $\tau \in \STAB(R\otimes A)$ and free operation $\G \in \OO$ such that the state $\omega = \idc_R \otimes \N(\tau)$ allows for the implementation of the channel $\N$ as $\N(\rho) = \G( \rho \otimes \omega) \; \forall \rho$. Then
\begin{equation}\begin{aligned}
  \ROg(\N) &= R_{\STAB}(\omega)\\
  W_\OO(\N) &= W_{\STAB}(\omega)\\
  F_\OO(\N) &= F_{\STAB}(\omega).
\end{aligned}\end{equation}
That is, the channel resource measures all reduce to the corresponding state-based resource measures of the associated state $\omega = \idc \otimes \N(\tau)$.
\end{lemma}
\end{boxxed}
The Lemma can be considered as an application of the idea of resource simulability~\cite{kaur_2017,wang_2019-1}, which generalises the notions of teleportation-based simulation from entanglement theory~\cite{bennett_1996,horodecki_1999-1,pirandola_2017}. Importantly, this result is valid for all unitary channels $\U(\cdot) = U \cdot U^\dagger$ where $U$ is a $k$-qubit unitary $U \in \C_3$ and $\tau$ is the $2k$-qubit maximally entangled state~\cite{gottesman_1999}. When $U$ is additionally a diagonal gate, the state injection can be performed more easily with the state $\omega = U \proj{+}^{\otimes k} U^\dagger$. Notably, when $U$ is a 1-, 2-, or 3-qubit diagonal unitary in $\C_3$, then the quantities $\ROg(\U)$ and $F_\OO(\U)$ become multiplicative, owing to the multiplicativity of the associated state-based monotones~\cite{bravyi_2019,seddon_2020}. We will hereafter restrict our discussion to diagonal unitaries $U \in \C_3$ of up to 3 qubits for simplicity, as this is enough to encompass most gates of practical interest (or gates Clifford-equivalent thereto). We then use $\ket{U} \coloneqq U \ket{+}^{\otimes k}$ for the corresponding states. An interesting property of Lem.~\ref{lem:main-stab-seizable} for such gates is that the result actually does not depend on the choice of free operations, meaning that the robustness, weight, and fidelity measures all have the same values for any set of channels $\OO$ which can implement state injection gadgets --- this ranges from the practically relevant stabiliser operations to all stabiliser-preserving channels~\cite{seddon_2019}.

Let us now look at how our main results can be understood in this setting. In general, the bound of Thm.~\ref{thm:submult} provides insight into the best achievable performance of any free channel transformation protocol; namely, for any transformation $U \to U'$ which takes $n$ uses of a diagonal gate $U \in \C_3$ to $m$ copies of a target diagonal gate $U' \in \C_3$, it holds that
\begin{equation}\begin{aligned}
	\frac{n}{m} \geq \frac{\log R_{\STAB}(\proj{U'})}{\log R_{\STAB}(\proj{U})}.
\end{aligned}\end{equation}
For instance, using the known values of $R_{\STAB}(\proj{T})$ and $R_{\STAB}(\proj{CCZ})$~\cite{bravyi_2019} we conclude that $n \geq 3.6335$ and so 4 T gates are necessary to perfectly synthesise a CCZ gate --- this is a slight strengthening of previous results which showed this optimality in other settings~\cite{howard_2017,beverland_2019}, as we establish that even the most general, adaptive channel manipulation protocols cannot do better.

In practice, the input gate might be affected by noise, and similarly the output of the protocol might only be required to be a good approximation of the target gate up to small errors. We can then apply Thm.~\ref{thm:main-nogo-parallel} to show that, for any channel $\E$ and any deterministic gate synthesis protocol $\CT$ which transforms $n$ uses of a noisy channel $\E$ to $m$ copies of the target unitary channel $\U(\cdot) = U \cdot U^\dagger$ up to accuracy $F(\CT(\E, \ldots, \E), \U^{\otimes m}) \geq 1-\ve$, it holds that
\begin{equation}\begin{aligned}\label{eq:main-magic-channels}
  n &\geq \log_{1/W_\OO(\E)} \frac{1- F_{\STAB}(U)^m}{\ve},\\
  n &\geq \log_{\ROg(\E)} \frac{1-\ve}{F_{\STAB}(U)^m}.
\end{aligned}\end{equation}
We stress that the coefficient $F_{\STAB}(U) = \max_{\sigma \in \STAB} \braket{U | \sigma | U}$ is the stabiliser fidelity~\cite{bravyi_2019} of the associated state $\ket{U}$, which is known exactly for most gates of interest~\cite{bravyi_2019,beverland_2019}: for instance, $F_{\STAB}(T) = \frac{1}{4}(2+\sqrt{2})$, $F_{\STAB}(CCZ) = \frac{9}{16}$. The bounds thus establish universal, efficiently computable restrictions on gate synthesis protocols, providing in particular a fundamental lower bound on the associated resource cost that any physical protocol must satisfy. Considering the previous example of the transformation $T \to CCZ$, we obtain $n > 3$ for error values up to $\ve \approx 0.095$, showing that a large error is necessary if one employs fewer than 4 T gates in the transformation.

As a special case, Thm.~\ref{thm:main-nogo-parallel} gives also fundamental restrictions on the resource cost of magic state distillation protocols. In particular, for any protocol $\M \in \OO$ which satisfies $F(\M(\rho^{\otimes n}), \proj{U}^{\otimes m}) \geq 1-\ve$, we necessarily have
\begin{equation}\begin{aligned}\label{eq:main-magic-states}
  n &\geq \log_{1/W_\FF(\rho)} \frac{1- F_{\STAB}(U)^m}{\ve},\\
  n &\geq \log_{\ROg(\rho)} \frac{1-\ve}{F_{\STAB}(U)^m}.
\end{aligned}\end{equation}
Our weight-based bound improves on the previous bound of Ref.~\cite{fang_2020} and extends its applicability beyond full-rank input states. An explicit comparison in Fig.~\ref{fig:magic-oneshot} of the main text reveals that this improvement can be substantial. We note that the robustness bound previously appeared in~Ref.~\cite{seddon_2020} in this context.

Moreover, from Thm.~\ref{thm:main-rate-bound} we obtain a strong converse bound on all transformations of channels as
\begin{equation}\begin{aligned}
  r^\dagger_{\rm adap}(\E \to \U) \leq \frac{\log R_{\OO}(\E)}{\log F_\OO(U)^{-1}}
\end{aligned}\end{equation}
for any diagonal $U \in \C_3$ of up to 3 qubits. This establishes limitations on the capabilities of general adaptive protocols in gate synthesis in the asymptotic limit.

We now consider some explicit examples to compare the different bounds. In the examples below, we compute the robustness $\ROg$ and weight $W_\OO$ with the choice of $\OO$ as the completely stabiliser-preserving operations. We refer to Fig.~\ref{fig:magic-oneshot} of the main text.

First, we compare the bounds on the error incurred in one-shot transformations, as stated in Thm.~\ref{thm:main-nogo} and Cor.~\ref{thm:main-nogo-states}. In Fig.~\ref{fig:magic-oneshot}\textbf{a}, our aim is to transform the noisy T gate $\N_p = \D_p \circ \T$ to the CCZ gate, where $\D_p(\rho) = (1-p) \rho + p \frac{\id}{2}$ is the depolarising channel. Notably, we see that the robustness bound indicates a significant error also in the noiseless case ($p=0$), whereas the weight bound becomes trivial for noiseless inputs. In Fig.~\ref{fig:magic-oneshot}\textbf{b},  we study the error incurred in distilling the $\ket{CCZ}$ state from three copies of the noisy T state $\rho_p = \D_p(\proj{T})$. The result explicitly shows the substantial improvement of our bounds over the bounds of Ref.~\cite{fang_2020}, which become nearly 0 in this case. We also see the previously remarked fact that a significant ($\approx 0.1$) error is unavoidable in the conversion $\ket{T}^{\otimes 3} \to \ket{CCZ}$.

When considering the distillation of noisy resources, noise applied at the level of channels can affect the results in different ways than noise applied at the level of states. For example, in subfigures \textbf{c}--\textbf{d} we present a comparison between the bounds for magic state distillation from the noisy T state $\rho_p$ and for gate synthesis from the noisy T gate $\N_p$. The result shows a difference between the gate synthesis properties of the two: we can see that the bounds impose much higher requirements on the number of noisy states required to succeed. The comparatively weaker bound on the required copies of the noisy T gate indicates a possibility that, as long as a certain type of noise is fixed (e.g.\ depolarising noise here), manipulating the noisy gate at the level of channels might offer an improvement over methods which rely on distilling $\ket{T}$ states from the noisy state $\rho_p$.

The difference between the state and channel cases can be understood as follows. Letting $\G$ be the standard state injection gadget such that $\G(\cdot \otimes \proj{T}) = \T(\cdot)$~\cite{gottesman_1999}, one can easily verify that $\G(\cdot \otimes \frac{\id}{2})$ is the completely \textit{dephasing} channel $\mathcal{Z}_{1/2}(\rho) = \frac12 \rho + \frac12 Z\rho Z$. This means that the state $\rho_p$ is equivalent to the channel $\N^\mathcal{Z}_p = (1-p) \T + p \mathcal{Z}_{1/2}$, and thus depolarising noise at the level of states corresponds to dephasing, rather than depolarising, at the level of channels. In particular, Lem.~\ref{lem:main-stab-seizable} gives $\ROg(\N^\mathcal{Z}_p) = \RFg(\rho_p)$ and analogously for the other quantities.

\subsubsection{Remarks on the case of qudits}\label{sec:magic_qudits}

Our results apply in an analogous way to the resource theory of magic for qudit systems of prime dimension~\cite{Veitch2014resource}. Here, our robustness-based bounds can be used to recover results related to~Refs.~\cite{wang_2020,wang_2019-1}, although these works considered a slightly different approach based on the discrete Wigner function and associated norms. The technical differences in the framework of~Ref.~\cite{wang_2020}, and in particular the use of sub-normalised operators rather than normalised quantum states, mean that these results do not exactly fit within our framework. Nevertheless, we can use them to get insights about important states in this theory.

In particular, in Table~\ref{tab:resources} of the main text we reported the value of $F_\FF$ (and analogously $F_\OO$) for a selection of quantum states, and in particular the $H_+$ state and the Norrell state. These results are obtained by leveraging the findings of~Ref.~\cite{wang_2020}: for the aforementioned states, Ref.~\cite{wang_2020} showed that $\log R_\FF(\psi)$ (max-relative entropy of magic) equals a quantity called the `min-thauma' $\theta_{\min}(\psi)$. But in general it holds that
\begin{equation}\begin{aligned}
	 R_\FF(\psi) \geq F_\FF(\psi)^{-1} \geq 2^{\theta_{\min}(\psi)},
\end{aligned}\end{equation}
where the first inequality holds for any pure state~\cite{cavalcanti_2005,regula_2018} (cf.\ Lem.~\ref{lemma:golden_channel}), and the second is a consequence of the definition of $\theta_{\min}$~\cite{wang_2020}. Therefore, whenever $\theta_{\min}(\psi) = \log R_\FF(\psi)$, we conclude that also $\log F_\FF(\psi)$ has the same value.

For channels and states whose value of $F_{\OO}$ or $F_\FF$ has not been established or is not known to be multiplicative (such as the qutrit T state), one can instead substitute $F_\FF$ with $2^{-\theta_{\min}}$ in all of our bounds. In particular, since $\theta_{\min}$ is additive on tensor products~\cite{wang_2020}, the many-copy results of Thms.~\ref{thm:main-rate-bound} and \ref{thm:main-rate-bound parallel} can always be applied by choosing $G_\FF(\phi) = 2^{-\theta_{\min}(\phi)}$.


\subsection{Proof of Lem.~\ref{lem:main-stab-seizable} and extension to other resources}

The basic idea behind Lem.~\ref{lem:main-stab-seizable} is to exploit the fact that some channels can be reversibly interconverted with state resources through free operations, which means that the two types of resources can be considered as equivalent. This idea was first applied in the theory of entanglement~\cite[Sec. 5]{bennett_1996}, later extended to general entanglement manipulation protocols~\cite{pirandola_2017} and to other types of resource theories~\cite{wilde_2018,kaur_2017}. Within the theory of magic states, the idea can be understood as a generalisation of gate teleportation~\cite{gottesman_1999}.

A general formulation of this property is as follows (see also \cite[Sec.~VI]{wilde_2018}).
\begin{boxxed}{white}
\begin{lemma}\label{lemma:injection superchannel}
	Let $\OO$ be any class of free operations which can prepare all free states $\FF$, that is, $\R_\sigma \in \OO \; \forall \sigma \in \FF$. Let $\N : A \to B$ be any channel such that:
\begin{enumerate}[(i)]
\item there exists a free superchannel $\Gamma \in \SS$ and a state $\omega$ such that $\Gamma(\R_\omega) = \N$,
\item there exists a free superchannel $\TT \in \SS$ such that $\TT(\N) = \R_\omega$.
\end{enumerate}
Then
\begin{equation}\begin{aligned}
  \ROg(\N) &= R_{\FF}(\omega)\\
  W_\OO(\N) &= W_{\FF}(\omega)\\
  F_\OO(\N) &= F_{\FF}(\omega).
\end{aligned}\end{equation}
\end{lemma}
\end{boxxed}
\begin{proof}
Follows directly from the monotonicity of the measures $\ROg$, $W_\OO$, $F_\OO$ under free transformations $\SS$, coupled with Lem.~\ref{lem:prep}.
\end{proof}

In practice, the superchannel $\Gamma$ is often realised as a state injection protocol which provides $\omega$ as an ancillary system and processes the joint state with a free operation in $\OO$. This is the way in which this property is usually applied both in entanglement theory~\cite{bennett_1996,horodecki_1999-1,pirandola_2017} and in more general settings~\cite{wilde_2018,kaur_2017}. For completeness, we provide a statement of the property in this form (see also \cite{wang_2019-1}).
\begin{boxxed}{white}
\begin{lemma}
Consider any resource theory such that $\sigma, \sigma' \in \FF \Rightarrow \sigma \otimes \sigma' \in \FF$ and let $\OO$ be a chosen class of free operations. Let $\N : A \to B$ be any channel such that:
\begin{enumerate}[(i)]
\item it can be implemented through state injection, that is, there exists a free operation $\G \in \OO(AC \to B)$ and a state $\omega \in \DD(C)$ such that $\G(\cdot \otimes \omega) = \N(\cdot)$,
\item there exists a free transformation $\TT \in \SS$ such that $\TT(\N) = \R_\omega$; for example, $\N (\tau) = \omega$ for some $\tau \in \FF$, or $\idc \otimes \N (\tau) = \omega$ when $\M \in \OO \Rightarrow \idc \otimes \M \in \OO$.
\end{enumerate}
Then
\begin{equation}\begin{aligned}
  \ROg(\N) &= R_{\FF}(\omega)\\
  W_\OO(\N) &= W_{\FF}(\omega)\\
  F_\OO(\N) &= F_{\FF}(\omega).
\end{aligned}\end{equation}
\end{lemma}
\end{boxxed}
\begin{remark}In particular, the result holds in the resource theory of magic for any $k$-qubit unitary channel $\U(\cdot) = U\cdot U^\dagger$ from the third level of the Clifford hierarchy~\cite{gottesman_1999}. In this case: $\OO$ can be any subset of completely stabiliser-preserving operations which allows for the implementation of state injection gadgets; $\omega$ is given by the Choi state $(\id \otimes U) \proj{\psi^+} (\id \otimes U^\dagger)$ where $\ket{\psi^+}$ is the $2k$-qubit maximally entangled state; and the free superchannel $\TT$ consist of simply preparing the free state $\proj{\psi^+}$ before the application of $\U$. In the case of diagonal gates in the third level of the Clifford hierarchy, we can use $\omega = U \proj{+}^{\otimes k} U^\dagger$ and the result is valid for the larger class of stabiliser-preserving operations.
A similar state injection protocol can also be employed in qudit systems with odd prime dimensions as well~\cite{Howard2012qudit,Campbell2012magic}.\end{remark}
\begin{proof}
{\allowdisplaybreaks
The first part of the proof can be shown following~\cite[Prop.\ 22]{wang_2019-1}:
\begin{align*}
  \ROg(\N) &= \min_{\M\in \OO} R_{\max} \left( \N \| \M\right)\\
  &= \min_{\M\in \OO} R_{\max} \left( \G(\cdot \otimes \omega) \,\big\|\, \M\right)\\
  &\leq \min_{\sigma \in \FF} R_{\max} \left( \G(\cdot \otimes \omega) \,\big\|\, \G(\cdot \otimes \sigma) \right) \\
  &\leq \min_{\sigma \in \FF} R_{\max} \left( [\cdot \otimes \omega] \,\big\|\, [\cdot \otimes \sigma] \right) \tag{\stepcounter{equation}\theequation}\\
  &=\min_{\sigma \in \FF} \max_{\psi_{RA}} R_{\max} \left( \psi \otimes \omega \,\big\|\, \psi \otimes \sigma \right)\\
  &= \min_{\sigma \in \FF} R_{\max} \left( \omega \,\big\|\, \sigma \right) \\
  &= R_{\FF} (\omega),
\end{align*}
where the first inequality is since $\G(\cdot \otimes \sigma) \in \OO$ for any $\sigma \in \FF$, the second inequality is by the data processing inequality of $R_{\max}(\E\|\F)$~\cite{wilde_2020}, and the equality in the second-to-last line is by the data processing inequality of $R_{\max}(\rho\|\sigma)$~\cite{datta_2009}.
On the other hand, using the monotonicity of $\ROg$ under free superchannels, we have
\begin{equation}\begin{aligned}
  \ROg(\N) \geq \ROg(\TT(\N)) = \ROg(\R_\omega) \geq R_{\FF}(\omega)
\end{aligned}\end{equation}
where the last inequality is by Lem.~\ref{lem:prep}. 
The cases of $W_\OO$ and $F_\OO$ proceed in the same way.
}
\end{proof}

An interesting consequence of such an operational equivalence of channels $\N$ and states $\omega$ is that it allows us to simplify general, adaptive channel manipulation protocols into simply acting on $n$ copies of the state $\omega$ through the use of so-called teleportation stretching~\cite{pirandola_2017} (see also~\cite{kaur_2017}).


\section{Extension to probabilistic protocols}\label{sec:prob}

We now consider probabilistic sub-superchannels that need not preserve quantum channels, but can instead map them to their probabilistic implementations --- completely positive and trace--non-increasing maps. Recall that the set of free subchannels $\tilde\OO$ with respect to the given set of free channels $\OO$ is given by
\bal
\tilde\OO\coloneqq\lset \tilde\M \sbar \forall \rho\in\DD,\; \exists \M\in\OO, t\in[0,1]\right.\left.\mbox{ s.t. } \idc\otimes\tilde\M(\rho) = t\cdot\idc\otimes\M(\rho)\rset,
\label{eq:free subchannels def}
\eal
and the corresponding set of sub-superchannels is $\tilde\SS:=\lset\tilde\TT\sbar\forall \M\in\OO,\; \tilde\TT(\M)\in\tilde\OO\rset$. Our figure of merit in probabilistic channel transformations is then the conditioned fidelity
\bal
 F_{\rm cond}(\U,\tilde\TT(\E))\coloneqq\min_\psi F\left(\idc\otimes\U(\psi),\frac{\idc\otimes\tilde\TT(\E)(\psi)}{p(\psi)}\right)
 \label{eq:fidelity condition}
\eal
where $p(\psi)=\Tr[\idc\otimes\tilde\TT(\E)(\psi)]$.

It is instructive to see how the above notions translate into that for state theories. 
By taking preparation channels as free channels and relabeling $\OO\rightarrow\FF$ and $\SS\rightarrow\OO$, Eq.~\eqref{eq:free subchannels def} reduces to 
$\tilde\FF\coloneqq\lset\tilde\sigma\sbar \tilde\sigma\in\cone(\FF),\Tr[\tilde\sigma]\leq1\rset$ and correspondingly we get $\tilde\OO\coloneqq \lset\tilde\M\sbar\forall \sigma\in\FF,\ \tilde\M(\sigma)\in\tilde\FF\rset$. 
As for the conditional fidelity, since our free superchannels transform preparation channels to preparation channels, we replace the target unitary with target preparation channel that prepares target state $\phi$. 
Writing $\tilde\M(\rho)=p\eta,\ \eta\in\DD$, the conditional fidelity reduces to much clearer form $F_{\rm cond}(\phi,\tilde\M(\rho))=F(\phi, \eta)$. 

With this, we are now ready to present an extension of Thm.~\ref{thm:main-nogo} as follows.

\begin{boxxed}{red}
\begin{theorem} \label{thm:main-nogo probabilistic}
Suppose a free sub-superchannel $\tilde\TT \in \tilde\SS$ achieves $F_{\rm cond}(\tilde\TT(\E), \U) \geq 1-\ve$ for some resourceful unitary channel $\U$. Then, if $\tilde\TT$ occurs with input channel $\E$ and input state $\psi$ at probability $p=\Tr[\idc\otimes\tilde\Theta(\E)(\psi)]$, it holds that
\begin{equation}\begin{aligned}\label{eq:main-nogo-rob prob}
  \ve &\geq 1 - \frac{\ROg(\E) \, F_\OO^\psi(\U)}{p}
\end{aligned}\end{equation}
\vspace*{-\baselineskip}and
\begin{equation}\begin{aligned}\label{eq:main-nogo-weight prob loose}
  \ve &\geq 1-\frac{1-(1 - F_\OO^\psi(\U) )W_\OO(\E)}{p}
\end{aligned}\end{equation}
where $F_\OO^\psi(\U):=\max_{\M\in\OO}F(\idc\otimes\U(\psi),\idc\otimes\M(\psi))$.

Alternatively, taking $\M\in\OO$ to be a free channel such that $J_\E\geq W_\OO(\E)J_\M$, it holds that
\begin{equation}\begin{aligned}\label{eq:main-nogo-weight prob}
  \ve &\geq (1 - F_\OO^\psi(\U) )\, \frac{W_\OO(\E)\Tr[(\idc\otimes\tilde\TT(\M)(\psi)]}{p}
\end{aligned}\end{equation}
\end{theorem}
\end{boxxed}

We also obtain corresponding bounds for state-based theories.
\begin{boxxed}{white}
\begin{proposition}\label{prop:main-nogo probabilistic state}
If there exists a free subchannel $\tilde\M \in \tilde\OO$ such that $\tilde\M(\rho)=p\eta$ with $F(\eta, \phi) \geq 1-\ve$ for some resourceful pure state $\phi$. Then, it holds that
\begin{equation}\begin{aligned}\label{eq:main-nogo-rob prob state app}
  \ve &\geq 1 - \frac{R_\FF(\rho) \, F_\FF(\phi)}{p}
\end{aligned}\end{equation}
and
\begin{equation}\begin{aligned}\label{eq:main-nogo 1-weight prob state app}
\ve &\geq (1 - F_\FF(\phi) )\left(1- \frac{1- W_\FF(\rho)}{p}\right).
\end{aligned}\end{equation}

Alternatively, taking $\sigma\in\FF$ to be a free state such that $\rho\geq W_\FF(\rho)\sigma$, it holds that
\begin{equation}\begin{aligned}\label{eq:main-nogo-weight prob state app}
  \ve &\geq (1 - F_\FF(\phi) )\, \frac{W_\FF(\rho)\Tr[\tilde\M(\sigma)]}{p}.
\end{aligned}\end{equation}
\end{proposition}
\end{boxxed}

It can be easily seen that when $\tilde\TT$ is a superchannel, $p=1$ holds for any $\psi$ and the bounds of Thm.~\ref{thm:main-nogo probabilistic} reproduce Thm.~\ref{thm:main-nogo}.
An interesting question is whether the no-go statement implied by Thm.~\ref{thm:main-nogo}, which says that perfect purification with $\ve=0$ is impossible for any channel with  $W_\OO(\E)>0$, remains valid in probabilistic cases.
For the choice of $\M$ as an optimal channel satisfying $J_\E\geq W_\OO(\E)J_\M$, Eq.~\eqref{eq:main-nogo-weight prob} (and \eqref{eq:main-nogo-weight prob state app} for state theories) implies that \emph{if} $\Tr[\idc\otimes\tilde\TT(\M)(\psi)]>0$, the no-go theorem still holds.   
On the other hand, if $\Tr[\idc\otimes\tilde\TT(\M)(\psi)]=0$, meaning that the free part of $\E$ is completely cut off by the selective operation $\tilde\TT$, then this does not give us any insight into $\ve$. 
This is actually a natural consequence because such a perfect purification \emph{is} indeed possible, implying that one cannot expect a bound in which $\ve>0$ for $W_\OO(\E)>0$ (see also~\cite{fang_2020-2}).

As an illustrative example, let us consider a state theory, in particular, the theory of coherence where $\FF$ is the set $\conv\{ \dm{0}, \dm{1}, \dm{2}\}$ defined on a three-dimensional system. Take $\rho = (1/2)\dm{0} + (1/2)\dm{+_{12}}$ where $\ket{+_{12}}=(1/\sqrt{2})(\ket{1}+\ket{2})$.
Consider a free subchannel $\tilde\M(\cdot)=P_{12}\cdot P_{12}$ where $P_{12}\coloneqq\dm{1}+\dm{2}$ is the projector onto the subspace spanned by $\{\ket{1},\ket{2}\}$.
Then we get $\tilde\M(\rho)=(1/2)\dm{+_{12}}$, indicating that we can obtain a pure state $\ket{+_{12}}$ with probability 1/2, although $W_\FF(\rho)=1/2>0$. 
This shows a clear contrast to the deterministic case of Thm.~\ref{thm:main-nogo} and Cor.~\ref{thm:main-nogo-states} as well as the probabilistic bounds shown in Ref.~\cite{fang_2020}, which showed that perfect purification is impossible even probabilistically for full-rank states.   

Another interesting observation from Thm.~\ref{thm:main-nogo probabilistic} is that unlike Eqs.~\eqref{eq:main-nogo-rob prob} and \eqref{eq:main-nogo-weight prob loose} where the lower bounds for error decrease as success probability decreases, Eq.~\eqref{eq:main-nogo-weight prob} appears to have an opposite behaviour with respect to the success probability.
In fact, there is an intricate trade-off between the probability of detecting the free component $\Tr[\idc\otimes\tilde\TT(\M)(\psi)]$ and success probability $p$, and this provides a practical instruction that to accomplish a purification protocol with high accuracy, it is crucial to reduce the probability of detecting the free part of the given channel, as characterised by the resource weight $W_\OO$.

\begin{proof}[\textbf{\textup{Proof of Thm.~\ref{thm:main-nogo probabilistic}}}]
Define $\ket{U_\psi}\coloneqq \id\otimes U\ket{\psi}$ and let $\M\in\OO$ be a channel that satisfies $J_\E \leq R_\OO(\E) J_\M$. Then, by \eqref{eq:fidelity condition} we get
\bal
 p(1-\ve)&\leq \braket{U_\psi|\idc\otimes\tilde\TT(\E)(\psi)|U_\psi}\\
 &\leq R_\OO(\E)\braket{U_\psi|\idc\otimes\tilde\TT(\M)(\psi)|U_\psi} \\
 &\leq R_\OO(\E)F_\OO^\psi(\U).
\eal
The last line is because 
\bal
\braket{U_\psi|\idc\otimes\tilde\TT(\M)(\psi)|U_\psi} &= t\braket{U_\psi|\idc\otimes\M'(\psi)|U_\psi}\\
&\leq t F_\OO^\psi(\U)\leq F_\OO^\psi(\U)
\label{eq:bound free subsuperchannel}
\eal
where we used that $\tilde\TT(\M)\in\tilde\OO$, ensuring by \eqref{eq:free subchannels def} that there exist $\M'\in\OO$ and $t\in[0,1]$ such that 
\bal
\idc\otimes\tilde\TT(\M)(\psi)=t\cdot\idc\otimes\M'(\psi).
\label{eq:free subsuperchannel proportional}
\eal

To show the second bound, let us redefine $\M,\N$ as the channels satisfying $\E=W_\OO(\E)\M+(1-W_\OO(\E))\N,\ \M\in\OO$.  
Then,
\bal
 p(1-\ve)&\leq \braket{U_\psi|\idc\otimes\tilde\TT(\E)(\psi)|U_\psi}\\
 &=W_{\OO}(\E)\braket{U_\psi|\idc\otimes\tilde\TT(\M)(\psi)|U_\psi} \\
 &\quad + (1-W_{\OO}(\E)) \braket{U_\psi|\idc\otimes\tilde\TT(\N)(\psi)|U_\psi}\\
 &\leq W_{\OO}(\E)F_\OO^\psi(\U)\Tr[\idc\otimes\tilde\TT(\M)(\psi)]\\
 &\quad + (1-W_\OO(\E)) \Tr[\idc\otimes\tilde\TT(\N)(\psi)]\\
 &\leq 1-(1-F_\OO^\psi(\U))W_\OO(\E).
 \label{eq:weight subchannel}
\eal
To get the third line, we bound the first term in the second line by using \eqref{eq:bound free subsuperchannel} and identifying $t=\Tr[\idc\otimes\tilde\TT(\M)(\psi)]$, which can be seen by taking trace in both sides of \eqref{eq:free subsuperchannel proportional}. 
We used a similar reasoning to bound the second term; for any channel $\N$, sub-superchannel $\tilde\TT$, and state $\psi$, there exists another channel $\N'$ such that $\idc\otimes\tilde\TT(\N)(\psi)=t\cdot\idc\otimes\N'(\psi)$ with $t=\Tr[\idc\otimes\tilde\TT(\N)(\psi)]$.
Thus, we get $\braket{U_\psi|\idc\otimes\tilde\TT(\N)(\psi)|U_\psi}=\Tr[\idc\otimes\tilde\TT(\N)(\psi)]\braket{U_\psi|\idc\otimes\N'(\psi)|U_\psi}\leq \Tr[\idc\otimes\tilde\TT(\N)(\psi)]$.
Finally, to get the final line we used $\Tr[\idc\otimes\tilde\TT(\M)(\psi)]\leq 1$ and $\Tr[\idc\otimes\tilde\TT(\N)(\psi)]\leq 1$.

To get the third bound, noting 
\begin{equation}\begin{aligned}
\Tr\left[\idc\otimes\tilde\TT(\E)(\psi)\right]=W_\OO(\E)\Tr\left[\idc\otimes\tilde\TT(\M)(\psi)\right] +(1-W_\OO(\E))\Tr\left[\idc\otimes\tilde\TT(\N)(\psi)\right],
\end{aligned}\end{equation}
we equate the right-hand side of the third line of \eqref{eq:weight subchannel} to $\Tr[\idc\otimes\tilde\TT(\E)(\psi)]-(1-F_\OO^\psi(\U))W_\OO(\E)\Tr[\idc\otimes\tilde\TT(\M)(\psi)]$.
Plugging in $\Tr[\idc\otimes\tilde\TT(\E)(\psi)]=p$ leads to the inequality in the statement. 

\end{proof}

\begin{proof}[\textbf{\textup{Proof of Prop.~\ref{prop:main-nogo probabilistic state}}}]
 Proofs for \eqref{eq:main-nogo-rob prob state app} and \eqref{eq:main-nogo-weight prob state app} are analogous to Thm.~\ref{thm:main-nogo probabilistic}.
 To show \eqref{eq:main-nogo 1-weight prob state app}, we use the following strong monotonicity property of $W_\FF$.

\begin{boxxed}{white}
\begin{lemma}[\cite{regula_2018}]
\label{lem:strong monotonicity weight rmax}
Let $\{\tilde\M_i\}$ be a free instrument, i.e., $\tilde\M_i(\sigma)\in \cone(\FF),\forall \sigma\in\FF,\forall i$. Writing $p_i=\Tr[\tilde\M_i(\rho)]$, we get
\bal
 W_\FF(\rho)\leq \sum_i p_i W_\FF\left(\frac{\tilde\M_i(\rho)}{p_i}\right)
\eal
and
\bal
 R_\FF(\rho)\geq \sum_i p_i R_\FF\left(\frac{\tilde\M_i(\rho)}{p_i}\right)
\eal
\end{lemma}
\end{boxxed}

\begin{proof}
 Write $\rho \geq W_\FF(\rho)\sigma$ with a free state $\sigma\in\FF$ and some state $\tau$.
 Then,
 \bal
  \frac{\tilde\M_i(\rho)}{p_i} \geq \frac{W_\FF(\rho)}{p_i}\Tr[\tilde\M_i(\sigma)]\frac{\tilde\M_i(\sigma)}{\Tr[\tilde\M_i(\sigma)]}.
 \eal
Since $\tilde\M_i(\sigma)/\Tr[\tilde\M_i(\sigma)]\in\FF$, the above expression serves as a valid free decomposition of $\tilde\M_i(\rho)/p_i$. 
Thus, we get $W_\FF(\tilde\M_i(\rho)/p_i)\geq \Tr[\tilde\M_i(\sigma)]W_\FF(\rho)/p_i$. 
The statement is reached by multiplying by $p_i$ and summing over $i$ in both sides, as well as using $\sum_i\Tr[\M_i(\rho)]=1$.
The bound for $R_\FF$ can be shown analogously. 
\end{proof}
Using Lem.~\ref{lem:strong monotonicity weight rmax} we get that 
\bal
 1-W_\FF(\rho)&\geq 1 - \sum_i p_i W_\FF\left(\frac{\tilde\M_i(\rho)}{p_i}\right) \\&= \sum_i p_i\left[1-W_\FF\left(\frac{\tilde\M_i(\rho)}{p_i}\right)\right],
 \label{eq:1-weight strong monotoniciy}
\eal
namely, $1-W_\FF(\rho)$ also shows the strong monotonicity (with the opposite direction of inequality).  
Then, since one can always construct a free instrument by complementing $\tilde\M$ with a replacement subchannel $\tilde\R_\sigma$ defined by the Choi matrix $J_{\tilde\R_\sigma}=(\id-\Tr_B J_{\tilde\M})\otimes\sigma,\ \sigma\in\FF$, \eqref{eq:1-weight strong monotoniciy} implies that 
\bal
  1-W_\FF(\rho)\geq  p\left[1-W_\FF\left(\frac{\tilde\M(\rho)}{p}\right)\right] \geq p\left[1-\frac{\ve}{1-F_\FF(\phi)}\right],
\eal
where in the last inequality we used Cor.~\ref{thm:main-nogo-states}. 
A simple reordering of the terms leads to the bound in the statement.
 
\end{proof}


\end{document}

%% file: definitions_manuscript.tex
\usepackage{graphicx}
\usepackage[dvipsnames,svgnames]{xcolor}
\usepackage{bm,bbm}%
\usepackage{braket}
\usepackage{amsthm,amsmath,amssymb}
\usepackage{enumerate}
\usepackage[caption=false]{subfig}
\usepackage{booktabs,multirow}
\usepackage{mathtools,bbding}
\usepackage[percent]{overpic}
\usepackage{microtype}
\usepackage{array,adjustbox,afterpage,xstring,refcount}
 \usepackage[T1]{fontenc}
 \usepackage[utf8]{inputenc}
 \usepackage{tikz}
 \usetikzlibrary{calc}

\usepackage{newtxtext,newtxmath}
\usepackage[scr=esstix]{mathalpha}

\newtheoremstyle{red}{}{}{\itshape}{}{\color{red!80!black}\bfseries}{.}{ }{}

\definecolor{darkred}{rgb}{0.57,0,0.12}
\usepackage[colorlinks=true, linkcolor=MidnightBlue, urlcolor=darkred!75!black, citecolor=MidnightBlue, linktoc = all]{hyperref}

\usepackage{cleveref}
\usepackage[most,breakable]{tcolorbox}

\renewcommand{\thesection}{\arabic{section}}

\makeatletter
\renewcommand{\p@subsection}{}
\renewcommand{\p@subsubsection}{}
\makeatother

\AtBeginDocument{
\heavyrulewidth=.08em
\lightrulewidth=.05em
\cmidrulewidth=.03em
\belowrulesep=.65ex
\belowbottomsep=0pt
\aboverulesep=.4ex
\abovetopsep=0pt
\cmidrulesep=\doublerulesep
\cmidrulekern=.5em
\defaultaddspace=.5em
}

\let\nc\newcommand

\hypersetup{
    bookmarksnumbered=true, 
    breaklinks=true,
    unicode=false, 
    pdfstartview={FitH}, 
    pdftitle={}, 
    pdfnewwindow=true, 
    colorlinks=true, 
    allcolors=MidnightBlue!70!black!70!TealBlue
}

\DeclareMathOperator{\Tr}{Tr}

\DeclareMathOperator{\supp}{supp}

\DeclareMathOperator{\conv}{conv}

\newcommand{\norm}[2]{\left\lVert#1\right\rVert_{\,#2}}
\newcommand{\proj}[1]{\ket{#1}\!\bra{#1}}

\newcommand{\R}{\mathcal{R}}

\newcommand{\T}{\mathcal{T}}

\newcommand{\C}{\mathcal{C}}
\newcommand{\E}{\mathcal{E}}
\newcommand{\G}{\mathcal{G}}
\newcommand{\U}{\mathcal{U}}

\newcommand{\F}{\mathcal{F}}
\newcommand{\D}{\mathcal{D}}

\renewcommand{\P}{\mathcal{P}}
\newcommand{\M}{\mathcal{M}}
\newcommand{\N}{\mathcal{N}}

\newcommand{\mA}{\mathcal{A}}

\newcommand{\CPTP}{{\mathrm{CPTP}}}

\newcommand{\STAB}{{\FF_{\mathrm{STAB}}}}

\newcommand{\<}{\left\langle}
\renewcommand{\>}{\right\rangle}

\renewcommand{\bar}{\;\rule{0pt}{9.5pt}\right|\;}
\let\sbar\bar

\newcommand{\lset}{\left\{\left.}
\newcommand{\rset}{\right\}}

\DeclareMathOperator{\cone}{cone}

\newcommand{\RR}{\mathbb{R}}
\newcommand{\CC}{\mathbb{C}}
\renewcommand{\SS}{\mathbb{S}}

\newcommand{\DD}{\mathbb{D}}

\newcommand{\OO}{\mathbb{O}}
\newcommand{\FF}{\mathbb{F}}
\newcommand{\LL}{\mathbb{L}}

\newcommand{\ve}{\varepsilon}

\newcommand{\id}{\mathbbm{1}}
\newcommand{\idc}{\mathrm{id}}

\newcommand{\PPT}{\mathrm{PPT}}

\newcommand{\SEP}{\mathrm{SEP}}

\newenvironment{boxxed}[1]%
  {\expandafter\ifstrequal\expandafter{#1}{red}{\begin{tcolorbox}[colback=orange!3,colframe=orange!15]}{\begin{tcolorbox}[colback=white,colframe=gray!10,breakable,enhanced]}}%
  {\end{tcolorbox}}

\let\textsc\uppercase

{\begin{boxxed}{orange}\vspace{-\baselineskip}\begin{equation}}%
{\end{equation}\end{boxxed}}

\RenewDocumentEnvironment{boxed}{}%
{\begin{boxxed}{white}}%
{\end{boxxed}}

\newtheoremstyle{nc}
  {\topsep}
  {\topsep}
  {}
  {}
  {\itshape}
  {.}
  {.5em}
  {\thmname{#1}\thmnumber{ #2}\thmnote{ (#3)}}

\newtheorem{theorem}{Theorem}

\newtheorem{proposition}[theorem]{Proposition}
\newtheorem{corollary}[theorem]{Corollary}

\newtheorem{lemma}[theorem]{Lemma}

\theoremstyle{definition}
\newtheorem*{remark}{Remark}

\theoremstyle{red}

\let\oldproofname\proofname
\renewcommand{\proofname}{\rm\bf{\oldproofname}}

  \nc{\MIO}{{\text{\rm MIO}}}
\nc{\DIO}{{\text{\rm DIO}}}
\nc{\SIO}{{\text{\rm SIO}}}
\nc{\IO}{{\text{\rm IO}}}
\nc{\CRNG}{{\text{\rm CRNG}}}

\nc{\lsetr}{\left\{\,}
\nc{\rsetr}{\right.\right\}}
\nc{\barr}{\;\rule{0pt}{9.5pt}\left|\;}
\nc{\ketbra}[2]{\ket{#1}\!\bra{#2}}

\nc{\TT}{\Theta}
\nc{\CT}{\Upsilon}

\nc{\wt}{\widetilde}

\makeatletter
\def\l@f@section{%
 \addpenalty{\@secpenalty}%
 \addvspace{0.4em plus\p@}%
}%
\makeatother

\nc{\dia}{\!\!\Diamond}
\nc{\bdia}{\!\!\vardiamond}

\newcommand{\RFg}{R_{\FF}}

\newcommand{\ROg}{R_{\OO}}

\DeclarePairedDelimiter\floor{\lfloor}{\rfloor}

\newcommand{\bal}{\begin{equation}\begin{aligned}}
\newcommand{\eal}{\end{aligned}\end{equation}}
\newcommand{\dm}[1]{\ketbra{#1}{#1}}

\makeatletter
\def\renewtheorem#1{%
  \expandafter\let\csname#1\endcsname\relax
  \expandafter\let\csname c@#1\endcsname\relax
  \gdef\renewtheorem@envname{#1}
  \renewtheorem@secpar
}
\def\renewtheorem@secpar{\@ifnextchar[{\renewtheorem@numberedlike}{\renewtheorem@nonumberedlike}}
\def\renewtheorem@numberedlike[#1]#2{\newtheorem{\renewtheorem@envname}[#1]{#2}}
\def\renewtheorem@nonumberedlike#1{  
\def\renewtheorem@caption{#1}
\edef\renewtheorem@nowithin{\noexpand\newtheorem{\renewtheorem@envname}{\renewtheorem@caption}}
\renewtheorem@thirdpar
}
\def\renewtheorem@thirdpar{\@ifnextchar[{\renewtheorem@within}{\renewtheorem@nowithin}}
\def\renewtheorem@within[#1]{\renewtheorem@nowithin[#1]}

\renewcommand\onecolumngrid{
\do@columngrid{one}{\@ne}%
\def\set@footnotewidth{\onecolumngrid}
\def\footnoterule{\kern-6pt\hrule width 1.5in\kern6pt}%
}
\makeatother